\PassOptionsToPackage{pdftex,
pdfencoding=auto,
pdfnewwindow=true,
pdfusetitle=true,
bookmarks=true,
bookmarksnumbered=true,
bookmarksopen=true,
pdfpagemode=UseThumbs,
bookmarksopenlevel=1,
pdfpagelabels=false,
breaklinks=true,
colorlinks=true
}{hyperref}
\PassOptionsToPackage{usenames,dvipsnames,pdftex,table}{xcolor}
\PassOptionsToPackage{letterpaper,margin=0cm,bottom=0.1cm}{geometry}
\documentclass[prx,10pt,letterpaper,twocolumn,accepted=2024-03-06,allowfontchangeintitle]{quantumarticle}
\pdfoutput=1
\usepackage[letterpaper,margin=0cm,bottom=0.1cm]{geometry}
\usepackage[english]{babel}
\usepackage[numbers,sort&compress]{natbib}
\usepackage{amsmath,amssymb,graphicx,enumerate,bbm,xcolor,latexsym,relsize}
\usepackage{amsthm}
\usepackage{mathtools}
\usepackage{soul}
\usepackage{comment,makecell}
\usepackage{paralist, tabularx}
\usepackage{nicefrac}
\usepackage[T1]{fontenc} 
\usepackage[utf8]{inputenx}

\usepackage{subcaption}
\usepackage[labelformat=simple]{subcaption}

\usepackage{changepage}

\usepackage{enumitem}
\setlist{nosep}

\makeatletter
\g@addto@macro\normalsize{%
  \setlength\abovedisplayskip{5pt}
  \setlength\belowdisplayskip{5pt}
  \setlength\abovedisplayshortskip{2pt}
  \setlength\belowdisplayshortskip{5pt}
}
\makeatother

\usepackage[subfigure,titles]{tocloft}
\usepackage[page,title,toc]{appendix}

\usepackage{microtype}
\microtypecontext{spacing=nonfrench}
\microtypesetup{
expansion={true,nocompatibility},
protrusion={true,nocompatibility},
activate={true,nocompatibility},
tracking=true,
kerning=true,
spacing={true}
}
\usepackage[all=normal,floats=tight,mathspacing=tight,wordspacing=tight,paragraphs=normal,tracking=tight,charwidths=tight,mathdisplays=normal,sections=normal,margins=normal]{savetrees}
\everypar=\expandafter{\the\everypar\loosness=-1 }
\usepackage{amsmath,amsfonts,amssymb,amsthm, bbm}

\usepackage[normalem]{ulem}
\usepackage{enumerate}  

\newcommand\Op{\mathsf{Op}}
\newcommand\RL{\mathbf{FVect}_\mathds{R}}

\newcommand\QSS{\mathbf{QuasiSubStoch}}

\newcommand\SubS{\bf{SubStoch}}

\newcommand{\tikeq}[2][]{\begin{equation}\input{Diagrams/#2.tikz}\ #1\label{eq:#2}\end{equation}}
\newcommand\smallsquare{\mathsmaller{\mathsmaller{\mathsmaller\square}}}
\newcommand\smallboxtimes{\mathsmaller{\mathsmaller{\mathsmaller\boxtimes}}}

\usepackage[colorlinks=true]{hyperref} 

\usepackage{dsfont}
\usepackage{rotating}
\usepackage{floatpag}
\rotfloatpagestyle{empty}
\renewcommand{\theenumi}{\arabic{enumi}}
\usepackage{hyperref}

\usepackage{amsmath,amsthm,amssymb}
\usepackage{xspace,enumerate,color,epsfig}
\usepackage{graphicx}
\usepackage{marvosym}

\usepackage{tikzfig}
\usepackage{docmute}
\usepackage{keycommand}

\usepackage{pifont}

\usepackage{enumitem}

\newkeycommand{\pointmap}[style=point][1]{\,\tikz{\node[style=\commandkey{style}] (x) at (0,-0.3) {$#1$};\node [style=none] (3) at (0, 1.0) {};\node [style=none] (3) at (0, -1.0) {};\draw (x)--(0,0.7);}\,}

\newkeycommand{\typedkpoint}[style=kpoint,edgestyle=][2]{%
\,\begin{tikzpicture}
  \begin{pgfonlayer}{nodelayer}
    \node [style=none] (0) at (0, 1) {};
    \node [style=\commandkey{style}] (1) at (0, -0.75) {$#2$};
    \node [style=right label] (2) at (0.25, 0.5) {$#1$};
  \end{pgfonlayer}
  \begin{pgfonlayer}{edgelayer}
    \draw [\commandkey{edgestyle}] (1) to (0.center);
  \end{pgfonlayer}
\end{tikzpicture}\,}

\newkeycommand{\typedkpointdag}[style=kpoint adjoint,edgestyle=][2]{%
\,\begin{tikzpicture}
  \begin{pgfonlayer}{nodelayer}
    \node [style=none] (0) at (0, -1) {};
    \node [style=\commandkey{style}] (1) at (0, 0.75) {$#2$};
    \node [style=right label] (2) at (0.25, -0.5) {$#1$};
  \end{pgfonlayer}
  \begin{pgfonlayer}{edgelayer}
    \draw [\commandkey{edgestyle}] (0.center) to (1);
  \end{pgfonlayer}
\end{tikzpicture}\,}

\newcommand{\bistateadj}[1]{\,\begin{tikzpicture}[yshift=-3mm]
  \begin{pgfonlayer}{nodelayer}
    \node [style=kpointadj, minimum width=1 cm, inner sep=2pt] (0) at (0, 0.5) {$\psi$};
    \node [style=none] (1) at (-0.75, 0.25) {};
    \node [style=none] (2) at (0.75, 0.25) {};
    \node [style=none] (3) at (-0.75, -0.5) {};
    \node [style=none] (4) at (0.75, -0.5) {};
  \end{pgfonlayer}
  \begin{pgfonlayer}{edgelayer}
    \draw (1.center) to (3.center);
    \draw (2.center) to (4.center);
  \end{pgfonlayer}
\end{tikzpicture}\,}

\newkeycommand{\pointketbra}[style1=copoint,style2=point][2]{\,%
\begin{tikzpicture}
  \begin{pgfonlayer}{nodelayer}
    \node [style=none] (0) at (0, -1.5) {};
    \node [style=\commandkey{style1}] (1) at (0, -0.75) {$#1$};
    \node [style=\commandkey{style2}] (2) at (0, 0.75) {$#2$};
    \node [style=none] (3) at (0, 1.5) {};
  \end{pgfonlayer}
  \begin{pgfonlayer}{edgelayer}
    \draw (0.center) to (1);
    \draw (2) to (3.center);
  \end{pgfonlayer}
\end{tikzpicture}\,}

\newkeycommand{\twocopointketbra}[style1=copoint,style2=copoint,style3=point][3]{\,%
\begin{tikzpicture}
  \begin{pgfonlayer}{nodelayer}
    \node [style=none] (0) at (-0.75, -1.5) {};
    \node [style=none] (1) at (0.75, -1.5) {};
    \node [style=\commandkey{style1}] (2) at (-0.75, -0.75) {$#1$};
    \node [style=\commandkey{style2}] (3) at (0.75, -0.75) {$#2$};
    \node [style=\commandkey{style3}] (4) at (0, 0.75) {$#3$};
    \node [style=none] (5) at (0, 1.5) {};
  \end{pgfonlayer}
  \begin{pgfonlayer}{edgelayer}
    \draw (0.center) to (2);
    \draw (1.center) to (3);
    \draw (4) to (5.center);
  \end{pgfonlayer}
\end{tikzpicture}\,}

\newkeycommand{\twopointketbra}[style1=copoint,style2=point,style3=point][3]{\,%
\begin{tikzpicture}
  \begin{pgfonlayer}{nodelayer}
    \node [style=none] (0) at (0, -1.5) {};
    \node [style=none] (1) at (-0.75, 1.5) {};
    \node [style=\commandkey{style1}] (2) at (0, -0.75) {$#1$};
    \node [style=\commandkey{style2}] (3) at (-0.75, 0.75) {$#2$};
    \node [style=\commandkey{style3}] (4) at (0.75, 0.75) {$#3$};
    \node [style=none] (5) at (0.75, 1.5) {};
  \end{pgfonlayer}
  \begin{pgfonlayer}{edgelayer}
    \draw (0.center) to (2);
    \draw (1.center) to (3);
    \draw (4) to (5.center);
  \end{pgfonlayer}
\end{tikzpicture}\,}

\newkeycommand{\pointbraket}[style1=point,style2=copoint][2]{\,%
\begin{tikzpicture}
  \begin{pgfonlayer}{nodelayer}
    \node [style=\commandkey{style1}] (0) at (0, -0.5) {$#1$};
    \node [style=\commandkey{style2}] (1) at (0, 0.5) {$#2$};
  \end{pgfonlayer}
  \begin{pgfonlayer}{edgelayer}
    \draw (0) to (1);
  \end{pgfonlayer}
\end{tikzpicture}\,}

\newkeycommand{\kpointbraket}[style1=kpoint,style2=kpoint adjoint][2]{\,%
\begin{tikzpicture}
  \begin{pgfonlayer}{nodelayer}
    \node [style=\commandkey{style1}] (0) at (0, -0.75) {$#1$};
    \node [style=\commandkey{style2}] (1) at (0, 0.75) {$#2$};
  \end{pgfonlayer}
  \begin{pgfonlayer}{edgelayer}
    \draw (0) to (1);
  \end{pgfonlayer}
\end{tikzpicture}\,}

\newkeycommand{\kpointmap}[style1=kpoint,style2=map][2]{\,%
\begin{tikzpicture}
  \begin{pgfonlayer}{nodelayer}
    \node [style=\commandkey{style1}] (0) at (0, -1.1) {$#1$};
    \node [style=\commandkey{style2}] (1) at (0, 0.2) {$#2$};
    \node [style=none] (2) at (0, 1.5) {};
  \end{pgfonlayer}
  \begin{pgfonlayer}{edgelayer}
    \draw (0) to (1);
    \draw (1) to (2.center);
  \end{pgfonlayer}
\end{tikzpicture}%
\,}

\newkeycommand{\sandwichtwo}[style1=point,style2=map,style3=map,style4=copoint][4]{\,%
\begin{tikzpicture}
  \begin{pgfonlayer}{nodelayer}
    \node [style=\commandkey{style2}] (0) at (0, -0.75) {$#2$};
    \node [style=\commandkey{style3}] (1) at (0, 0.75) {$#3$};
    \node [style=\commandkey{style1}] (2) at (0, -2) {$#1$};
    \node [style=\commandkey{style4}] (3) at (0, 2) {$#4$};
  \end{pgfonlayer}
  \begin{pgfonlayer}{edgelayer}
    \draw (2) to (0);
    \draw (0) to (1);
    \draw (1) to (3);
  \end{pgfonlayer}
\end{tikzpicture}\,}

\newkeycommand{\longbraket}[style1=point,style2=copoint][2]{\,%
\begin{tikzpicture}
  \begin{pgfonlayer}{nodelayer}
    \node [style=\commandkey{style1}] (0) at (0, -2) {$#1$};
    \node [style=\commandkey{style2}] (1) at (0, 2) {$#2$};
  \end{pgfonlayer}
  \begin{pgfonlayer}{edgelayer}
    \draw (0) to (1);
  \end{pgfonlayer}
\end{tikzpicture}\,}

\newkeycommand{\onb}[style=point]{\ensuremath{\left\{\tikz{\node[style=\commandkey{style}] (x) at (0,-0.3) {$j$};\draw (x)--(0,0.7);}\right\}}\xspace}

\newkeycommand{\copointmap}[style=copoint][1]{\,\tikz{\node[style=\commandkey{style}] (x) at (0,0.3) {$#1$};\draw (x)--(0,-0.7);}\,}

\tikzstyle{cdiag}=[matrix of math nodes, row sep=3em, column sep=3em, text height=1.5ex, text depth=0.25ex,inner sep=0.5em]
\tikzstyle{arrow above}=[transform canvas={yshift=0.5ex}]
\tikzstyle{arrow below}=[transform canvas={yshift=-0.5ex}]

\usetikzlibrary{shapes.multipart}

\tikzset{->-/.style={decoration={
  markings,
  mark=at position .5 with {\arrow{>}}},postaction={decorate}}}
\tikzset{-<-/.style={decoration={
  markings,
  mark=at position .5 with {\arrow{<}}},postaction={decorate}}}

\tikzstyle{bwSpider}=[
       rectangle split,
       rectangle split parts=2,
       rectangle split part fill={black,white},
 minimum size=3.6 mm, inner sep=-2mm, draw=black,scale=0.5,rounded corners=0.8 mm
       ]
 \tikzstyle{wbSpider}=[
       rectangle split,
       rectangle split parts=2,
       rectangle split part fill={white,black},
 minimum size=3.6 mm, inner sep=-2mm, draw=black,scale=0.5,rounded corners=0.8 mm
       ]

\tikzstyle{epiCopoint}=[regular polygon,regular polygon sides=3,draw,scale=0.75,inner sep=-0.5pt,minimum width=5mm,fill=white,regular polygon rotate=0,line width=1pt]
\tikzstyle{epiPoint}=[regular polygon,regular polygon sides=3,draw,scale=0.75,inner sep=-0.5pt,minimum width=5mm,fill=white,regular polygon rotate=180,line width=1pt]
\tikzstyle{epiPointWide}=[regular polygon,regular polygon sides=3,draw,scale=0.75,inner sep=-0.5pt,minimum width=8mm,fill=white,regular polygon rotate=180,line width=1pt]
\tikzstyle{epiBox}=[fill=white,draw, line width = 1pt,inner sep=0.6mm,font=\footnotesize,minimum height=3mm,minimum width=3mm]
\tikzstyle{epiBoxWide}=[fill=white,draw, line width = 1pt,inner sep=0.6mm,font=\footnotesize,minimum height=3mm,minimum width=5mm]
\tikzstyle{epiBoxVeryWide}=[fill=white,draw, line width = 1pt,inner sep=0.6mm,font=\footnotesize,minimum height=3mm,minimum width=7mm]
\tikzstyle{qWire}=[line width = 1pt, color=black]
\tikzstyle{cWire}=[color=gray,line width = .75pt]
\tikzstyle{CqWire}=[color=gray,line width = .75pt,->-]
\tikzstyle{CcWire}=[color=gray,line width = .75pt,->-]
\tikzstyle{RqWire}=[line width = 1pt, color=black,-<-]
\tikzstyle{RcWire}=[color=gray,line width = .75pt,-<-]
\tikzstyle{env}=[copoint,regular polygon rotate=0,minimum width=0.2cm, fill=black]

\tikzstyle{probs}=[shape=semicircle,fill=white,draw=black,shape border rotate=180,minimum width=1.2cm]

\tikzstyle{every picture}=[baseline=-0.25em,scale=0.5]
\tikzstyle{dotpic}=[] 
\tikzstyle{diredges}=[every to/.style={diredge}]
\tikzstyle{math matrix}=[matrix of math nodes,left delimiter=(,right delimiter=),inner sep=2pt,column sep=1em,row sep=0.5em,nodes={inner sep=0pt},text height=1.5ex, text depth=0.25ex]

\tikzstyle{inline text}=[text height=1.5ex, text depth=0.25ex,yshift=0.5mm]
\tikzstyle{label}=[font=\footnotesize,text height=1.5ex, text depth=0.25ex,yshift=0.5mm]
\tikzstyle{left label}=[label,anchor=east,xshift=1.5mm]
\tikzstyle{right label}=[label,anchor=west,xshift=-1mm]
\tikzstyle{up label}=[label,anchor=south,yshift=-1mm]

\tikzstyle{braceedge}=[decorate,decoration={brace,amplitude=2mm,raise=-1mm}]
\tikzstyle{small braceedge}=[decorate,decoration={brace,amplitude=1mm,raise=-1mm}]

\tikzstyle{doubled}=[line width=1.6pt] 
\tikzstyle{boldedge}=[doubled,shorten <=-0.17mm,shorten >=-0.17mm]
\tikzstyle{boldedgegray}=[doubled,gray,shorten <=-0.17mm,shorten >=-0.17mm]
\tikzstyle{singleedgegray}=[gray]

\tikzstyle{semidoubled}=[line width=1.4pt] 
\tikzstyle{semiboldedgegray}=[semidoubled,gray,shorten <=-0.17mm,shorten >=-0.17mm]

\tikzstyle{boxedge}=[semiboldedgegray]

\tikzstyle{boldedgedashed}=[very thick,dashed,shorten <=-0.17mm,shorten >=-0.17mm]
\tikzstyle{vboldedgedashed}=[doubled,dashed,shorten <=-0.17mm,shorten >=-0.17mm]
\tikzstyle{left hook arrow}=[left hook-latex]
\tikzstyle{right hook arrow}=[right hook-latex]
\tikzstyle{sembracket}=[line width=0.5pt,shorten <=-0.07mm,shorten >=-0.07mm]

\tikzstyle{causal edge}=[->,thick,gray]
\tikzstyle{causal nondir}=[thick,gray]
\tikzstyle{timeline}=[thick,gray, dashed]

\tikzstyle{cedge}=[<->,thick,gray!70!white]

\tikzstyle{empty diagram}=[draw=gray!40!white,dashed,shape=rectangle,minimum width=1cm,minimum height=1cm]
\tikzstyle{empty diagram small}=[draw=gray!50!white,dashed,shape=rectangle,minimum width=0.6cm,minimum height=0.5cm]

\tikzstyle{dot}=[inner sep=0mm,minimum width=2mm,minimum height=2mm,draw,shape=circle]
\tikzstyle{bigdot}=[inner sep=0mm,minimum width=5mm,minimum height=5mm,draw,shape=circle]
\tikzstyle{leak}=[white dot, shape=regular polygon, minimum size=3.3 mm, regular polygon sides=3, outer sep=-0.2mm, regular polygon rotate=270]
\tikzstyle{proj}=[regular polygon,regular polygon sides=4,draw,scale=0.75,inner sep=-0.5pt,minimum width=6mm,fill=white]
\tikzstyle{projOut}=[regular polygon,regular polygon sides=3,draw,scale=0.75,inner sep=-0.5pt,minimum width=7.5mm,fill=white,regular polygon rotate=180]
\tikzstyle{projIn}=[regular polygon,regular polygon sides=3,draw,scale=0.75,inner sep=-0.5pt,minimum width=7.5mm,fill=white]
\tikzstyle{Vleak}=[white dot, shape=regular polygon, minimum size=3.3 mm, regular polygon sides=3, outer sep=-0.2mm, regular polygon rotate=90]
\tikzstyle{dleak}=[white dot, line width=1.6pt, shape=regular polygon, minimum size=3.3 mm, regular polygon sides=3, outer sep=-0.2mm, regular polygon rotate=270]

\tikzstyle{Wsquare}=[white dot, shape=regular polygon, rounded corners=0.8 mm, minimum size=3.3 mm, regular polygon sides=3, outer sep=-0.2mm]
\tikzstyle{Wsquareadj}=[white dot, shape=regular polygon, rounded corners=0.8 mm, minimum size=3.3 mm, regular polygon sides=3, outer sep=-0.2mm, regular polygon rotate=180]
\tikzstyle{ddot}=[inner sep=0mm, doubled, minimum width=2.5mm,minimum height=2.5mm,draw,shape=circle]

\tikzstyle{clear dot}=[dot,fill=none,text depth=-0.2mm,draw=gray, line width = .75pt]
\tikzstyle{tall clear dot}=[dot,fill=none,text depth=-0.2mm,draw=gray, line width = .75pt,shape=ellipse, minimum height=5mm]
\tikzstyle{wide clear dot}=[dot,fill=none,text depth=-0.2mm,draw=gray, line width = .75pt, shape=ellipse, minimum width = 5mm]
\tikzstyle{very wide clear dot}=[dot,fill=none,text depth=-0.2mm,draw=gray, line width = .75pt, shape=ellipse, minimum width = 7mm ]

\tikzstyle{black dot}=[dot,fill=black]
\tikzstyle{white dot}=[dot,fill=white,,text depth=-0.2mm]
\tikzstyle{white Wsquare}=[Wsquare,fill=gray,,text depth=-0.2mm]
\tikzstyle{white Wsquareadj}=[Wsquareadj,fill=white,,text depth=-0.2mm]
\tikzstyle{green dot}=[white dot] 
\tikzstyle{gray dot}=[dot,fill=gray!40!white,,text depth=-0.2mm]
\tikzstyle{red dot}=[gray dot] 

\tikzstyle{black ddot}=[ddot,fill=black]
\tikzstyle{white ddot}=[ddot,fill=white]
\tikzstyle{gray ddot}=[ddot,fill=gray!40!white]

\tikzstyle{gray edge}=[gray!60!white]

\tikzstyle{small dot}=[inner sep=0.2mm,minimum width=0pt,minimum height=0pt,draw,shape=circle]

\tikzstyle{small black dot}=[small dot,fill=black]
\tikzstyle{small white dot}=[small dot,fill=white]
\tikzstyle{small gray dot}=[small dot,fill=gray,draw=gray]

\tikzstyle{causal dot}=[inner sep=0.4mm,minimum width=0pt,minimum height=0pt,draw=white,shape=circle,fill=gray!40!white]

\tikzstyle{phase dimensions}=[minimum size=5mm,font=\footnotesize,rectangle,rounded corners=2.5mm,inner sep=0.2mm,outer sep=-2mm]
\tikzstyle{dphase dimensions}=[minimum size=5mm,font=\footnotesize,rectangle,rounded corners=2.5mm,inner sep=0.2mm,outer sep=-2mm]

\tikzstyle{white phase dot}=[dot,fill=white,phase dimensions]
\tikzstyle{white phase ddot}=[ddot,fill=white,dphase dimensions]

\tikzstyle{white rect ddot}=[draw=black,fill=white,doubled,minimum size=5mm,font=\footnotesize,rectangle,rounded corners=2.5mm,inner sep=0.2mm]
\tikzstyle{gray rect ddot}=[draw=black,fill=gray!40!white,doubled,minimum size=6mm,font=\footnotesize,rectangle,rounded corners=3mm]

\tikzstyle{gray phase dot}=[dot,fill=gray!40!white,phase dimensions]
\tikzstyle{gray phase ddot}=[ddot,fill=gray!40!white,dphase dimensions]
\tikzstyle{grey phase dot}=[gray phase dot]
\tikzstyle{grey phase ddot}=[gray phase ddot]

\tikzstyle{small phase dimensions}=[minimum size=4mm,font=\tiny,rectangle,rounded corners=2mm,inner sep=0.2mm,outer sep=-2mm]
\tikzstyle{small dphase dimensions}=[minimum size=4mm,font=\tiny,rectangle,rounded corners=2mm,inner sep=0.2mm,outer sep=-2mm]

\tikzstyle{small gray phase dot}=[dot,fill=gray!40!white,small phase dimensions]
\tikzstyle{small gray phase ddot}=[ddot,fill=gray!40!white,small dphase dimensions]

\tikzstyle{small map}=[draw,shape=rectangle,minimum height=4mm,minimum width=4mm,fill=white]

\tikzstyle{cnot}=[fill=white,shape=circle,inner sep=-1.4pt]

\tikzstyle{asym hadamard}=[fill=white,draw,shape=NEbox,inner sep=0.6mm,font=\footnotesize,minimum height=4mm]
\tikzstyle{asym hadamard conj}=[fill=white,draw,shape=NWbox,inner sep=0.6mm,font=\footnotesize,minimum height=4mm]
\tikzstyle{asym hadamard dag}=[fill=white,draw,shape=SEbox,inner sep=0.6mm,font=\footnotesize,minimum height=4mm]

\tikzstyle{hadamard}=[fill=white,draw,inner sep=0.6mm,font=\footnotesize,minimum height=4mm,minimum width=4mm]
\tikzstyle{small hadamard}=[fill=white,draw,inner sep=0.6mm,minimum height=1.5mm,minimum width=1.5mm]
\tikzstyle{small hadamard rotate}=[small hadamard,rotate=45]
\tikzstyle{dhadamard}=[hadamard,doubled]
\tikzstyle{small dhadamard}=[small hadamard,doubled]
\tikzstyle{small dhadamard rotate}=[small hadamard rotate,doubled]
\tikzstyle{antipode}=[white dot,inner sep=0.3mm,font=\footnotesize]

\tikzstyle{scalar}=[diamond,draw,inner sep=0.5pt,font=\small]
\tikzstyle{dscalar}=[diamond,doubled, draw,inner sep=0.5pt,font=\small]

\tikzstyle{small box}=[rectangle,inline text,fill=white,draw,minimum height=5mm,yshift=-0.5mm,minimum width=5mm,font=\small]
\tikzstyle{small gray box}=[small box,fill=gray!30]
\tikzstyle{medium box}=[rectangle,inline text,fill=white,draw,minimum height=5mm,yshift=-0.5mm,minimum width=10mm,font=\small]
\tikzstyle{square box}=[small box] 
\tikzstyle{medium gray box}=[small box,fill=gray!30]
\tikzstyle{semilarge box}=[rectangle,inline text,fill=white,draw,minimum height=5mm,yshift=-0.5mm,minimum width=12.5mm,font=\small]
\tikzstyle{large box}=[rectangle,inline text,fill=white,draw,minimum height=5mm,yshift=-0.5mm,minimum width=15mm,font=\small]
\tikzstyle{large gray box}=[small box,fill=gray!30]

\tikzstyle{Bayes box}=[rectangle,fill=black,draw, minimum height=3mm, minimum width=3mm]

\tikzstyle{gray square point}=[small box,fill=gray!50]

\tikzstyle{dphase box white}=[dhadamard]
\tikzstyle{dphase box gray}=[dhadamard,fill=gray!50!white]
\tikzstyle{phase box white}=[hadamard]
\tikzstyle{phase box gray}=[hadamard,fill=gray!50!white]

\tikzstyle{point}=[regular polygon,regular polygon sides=3,draw,scale=0.75,inner sep=-0.5pt,minimum width=9mm,fill=white,regular polygon rotate=180]
\tikzstyle{infpoint}=[regular polygon,regular polygon sides=3,draw,scale=0.75,inner sep=-0.5pt,minimum width=9mm,fill=white,regular polygon rotate=90]
\tikzstyle{point nosep}=[regular polygon,regular polygon sides=3,draw,scale=0.75,inner sep=-2pt,minimum width=9mm,fill=white,regular polygon rotate=180]
\tikzstyle{infcopoint}=[regular polygon,regular polygon sides=3,draw,scale=0.75,inner sep=-0.5pt,minimum width=9mm,fill=white,regular polygon rotate=270]
\tikzstyle{copoint}=[regular polygon,regular polygon sides=3,draw,scale=0.75,inner sep=-0.5pt,minimum width=9mm,fill=white]
\tikzstyle{dpoint}=[point,doubled]
\tikzstyle{dcopoint}=[copoint,doubled]

\tikzstyle{pointgrow}=[shape=cornerpoint,kpoint common,scale=0.75,inner sep=3pt]
\tikzstyle{pointgrow dag}=[shape=cornercopoint,kpoint common,scale=0.75,inner sep=3pt]

\tikzstyle{wide copoint}=[fill=white,draw,shape=isosceles triangle,shape border rotate=90,isosceles triangle stretches=true,inner sep=0pt,minimum width=1.5cm,minimum height=6.12mm]
\tikzstyle{wide point}=[fill=white,draw,shape=isosceles triangle,shape border rotate=-90,isosceles triangle stretches=true,inner sep=0pt,minimum width=1.5cm,minimum height=6.12mm,yshift=-0.0mm]
\tikzstyle{wide point plus}=[fill=white,draw,shape=isosceles triangle,shape border rotate=-90,isosceles triangle stretches=true,inner sep=0pt,minimum width=1.74cm,minimum height=7mm,yshift=-0.0mm]

\tikzstyle{wide dpoint}=[fill=white,doubled,draw,shape=isosceles triangle,shape border rotate=-90,isosceles triangle stretches=true,inner sep=0pt,minimum width=1.5cm,minimum height=6.12mm,yshift=-0.0mm]

\tikzstyle{tinypoint}=[regular polygon,regular polygon sides=3,draw,scale=0.55,inner sep=-0.15pt,minimum width=6mm,fill=white,regular polygon rotate=180]

\tikzstyle{white point}=[point]
\tikzstyle{white dpoint}=[dpoint]
\tikzstyle{green point}=[white point] 
\tikzstyle{white copoint}=[copoint]
\tikzstyle{gray point}=[point,fill=gray!40!white]
\tikzstyle{gray dpoint}=[gray point,doubled]
\tikzstyle{red point}=[gray point] 
\tikzstyle{gray copoint}=[copoint,fill=gray!40!white]
\tikzstyle{gray dcopoint}=[gray copoint,doubled]

\tikzstyle{white point guide}=[regular polygon,regular polygon sides=3,font=\scriptsize,draw,scale=0.65,inner sep=-0.5pt,minimum width=9mm,fill=white,regular polygon rotate=180]

\tikzstyle{black point}=[point,fill=black,font=\color{white}]
\tikzstyle{black copoint}=[copoint,fill=black,font=\color{white}]

\tikzstyle{tiny gray point}=[tinypoint,fill=gray!40!white]

\tikzstyle{diredge}=[->]
\tikzstyle{ddiredge}=[<->]
\tikzstyle{rdiredge}=[<-]
\tikzstyle{thickdiredge}=[->, very thick]
\tikzstyle{pointer edge}=[->,very thick,gray]
\tikzstyle{pointer edge part}=[very thick,gray]
\tikzstyle{dashed edge}=[dashed]
\tikzstyle{thick dashed edge}=[very thick,dashed]
\tikzstyle{thick gray dashed edge}=[thick dashed edge,gray!40]
\tikzstyle{thick map edge}=[very thick,|->]

\makeatletter
\newcommand{\boxshape}[3]{%
\pgfdeclareshape{#1}{
\inheritsavedanchors[from=rectangle] 
\inheritanchorborder[from=rectangle]
\inheritanchor[from=rectangle]{center}
\inheritanchor[from=rectangle]{north}
\inheritanchor[from=rectangle]{south}
\inheritanchor[from=rectangle]{west}
\inheritanchor[from=rectangle]{east}
\backgroundpath{
\southwest \pgf@xa=\pgf@x \pgf@ya=\pgf@y
\northeast \pgf@xb=\pgf@x \pgf@yb=\pgf@y

\@tempdima=#2
\@tempdimb=#3

\pgfpathmoveto{\pgfpoint{\pgf@xa - 5pt + \@tempdima}{\pgf@ya}}
\pgfpathlineto{\pgfpoint{\pgf@xa - 5pt - \@tempdima}{\pgf@yb}}
\pgfpathlineto{\pgfpoint{\pgf@xb + 5pt + \@tempdimb}{\pgf@yb}}
\pgfpathlineto{\pgfpoint{\pgf@xb + 5pt - \@tempdimb}{\pgf@ya}}
\pgfpathlineto{\pgfpoint{\pgf@xa - 5pt + \@tempdima}{\pgf@ya}}
\pgfpathclose
}
}}

\boxshape{NEbox}{0pt}{3pt}
\boxshape{SEbox}{0pt}{-3pt}
\boxshape{NWbox}{5pt}{0pt}
\boxshape{SWbox}{-5pt}{0pt}
\boxshape{EBox}{-3pt}{3pt}
\boxshape{WBox}{3pt}{-3pt}
\makeatother

\tikzstyle{cloud}=[shape=cloud,draw,minimum width=1.5cm,minimum height=1.5cm]

\tikzstyle{map}=[draw,shape=NEbox,inner sep=1pt,minimum height=4mm,fill=white]
\tikzstyle{dashedmap}=[draw,dashed,shape=NEbox,inner sep=2pt,minimum height=6mm,fill=white]
\tikzstyle{mapdag}=[draw,shape=SEbox,inner sep=1pt,minimum height=4mm,fill=white]
\tikzstyle{mapadj}=[draw,shape=SEbox,inner sep=2pt,minimum height=6mm,fill=white]
\tikzstyle{maptrans}=[draw,shape=SWbox,inner sep=2pt,minimum height=6mm,fill=white]
\tikzstyle{mapconj}=[draw,shape=NWbox,inner sep=2pt,minimum height=6mm,fill=white]

\tikzstyle{medium map}=[draw,shape=NEbox,inner sep=2pt,minimum height=6mm,fill=white,minimum width=7mm]
\tikzstyle{medium map dag}=[draw,shape=SEbox,inner sep=2pt,minimum height=6mm,fill=white,minimum width=7mm]
\tikzstyle{medium map adj}=[draw,shape=SEbox,inner sep=2pt,minimum height=6mm,fill=white,minimum width=7mm]
\tikzstyle{medium map trans}=[draw,shape=SWbox,inner sep=2pt,minimum height=6mm,fill=white,minimum width=7mm]
\tikzstyle{medium map conj}=[draw,shape=NWbox,inner sep=2pt,minimum height=6mm,fill=white,minimum width=7mm]
\tikzstyle{semilarge map}=[draw,shape=NEbox,inner sep=2pt,minimum height=6mm,fill=white,minimum width=9.5mm]
\tikzstyle{semilarge map trans}=[draw,shape=SWbox,inner sep=2pt,minimum height=6mm,fill=white,minimum width=9.5mm]
\tikzstyle{semilarge map adj}=[draw,shape=SEbox,inner sep=2pt,minimum height=6mm,fill=white,minimum width=9.5mm]
\tikzstyle{semilarge map dag}=[draw,shape=SEbox,inner sep=2pt,minimum height=6mm,fill=white,minimum width=9.5mm]
\tikzstyle{semilarge map conj}=[draw,shape=NWbox,inner sep=2pt,minimum height=6mm,fill=white,minimum width=9.5mm]
\tikzstyle{large map}=[draw,shape=NEbox,inner sep=2pt,minimum height=6mm,fill=white,minimum width=12mm]
\tikzstyle{large map conj}=[draw,shape=NWbox,inner sep=2pt,minimum height=6mm,fill=white,minimum width=12mm]
\tikzstyle{very large map}=[draw,shape=NEbox,inner sep=2pt,minimum height=6mm,fill=white,minimum width=17mm]

\tikzstyle{medium dmap}=[draw,doubled,shape=NEbox,inner sep=2pt,minimum height=6mm,fill=white,minimum width=7mm]
\tikzstyle{medium dmap dag}=[draw,doubled,shape=SEbox,inner sep=2pt,minimum height=6mm,fill=white,minimum width=7mm]
\tikzstyle{medium dmap adj}=[draw,doubled,shape=SEbox,inner sep=2pt,minimum height=6mm,fill=white,minimum width=7mm]
\tikzstyle{medium dmap trans}=[draw,doubled,shape=SWbox,inner sep=2pt,minimum height=6mm,fill=white,minimum width=7mm]
\tikzstyle{medium dmap conj}=[draw,doubled,shape=NWbox,inner sep=2pt,minimum height=6mm,fill=white,minimum width=7mm]
\tikzstyle{semilarge dmap}=[draw,doubled,shape=NEbox,inner sep=2pt,minimum height=6mm,fill=white,minimum width=9.5mm]
\tikzstyle{semilarge dmap trans}=[draw,doubled,shape=SWbox,inner sep=2pt,minimum height=6mm,fill=white,minimum width=9.5mm]
\tikzstyle{semilarge dmap adj}=[draw,doubled,shape=SEbox,inner sep=2pt,minimum height=6mm,fill=white,minimum width=9.5mm]
\tikzstyle{semilarge dmap dag}=[draw,doubled,shape=SEbox,inner sep=2pt,minimum height=6mm,fill=white,minimum width=9.5mm]
\tikzstyle{semilarge dmap conj}=[draw,doubled,shape=NWbox,inner sep=2pt,minimum height=6mm,fill=white,minimum width=9.5mm]
\tikzstyle{large dmap}=[draw,doubled,shape=NEbox,inner sep=2pt,minimum height=6mm,fill=white,minimum width=12mm]
\tikzstyle{large dmap conj}=[draw,doubled,shape=NWbox,inner sep=2pt,minimum height=6mm,fill=white,minimum width=12mm]
\tikzstyle{large dmap trans}=[draw,doubled,shape=SWbox,inner sep=2pt,minimum height=6mm,fill=white,minimum width=12mm]
\tikzstyle{large dmap adj}=[draw,doubled,shape=SEbox,inner sep=2pt,minimum height=6mm,fill=white,minimum width=12mm]
\tikzstyle{large dmap dag}=[draw,doubled,shape=SEbox,inner sep=2pt,minimum height=6mm,fill=white,minimum width=12mm]
\tikzstyle{very large dmap}=[draw,doubled,shape=NEbox,inner sep=2pt,minimum height=6mm,fill=white,minimum width=19.5mm]

\tikzstyle{muxbox}=[draw,shape=rectangle,minimum height=3mm,minimum width=3mm,fill=white]
\tikzstyle{dmuxbox}=[muxbox,doubled]

\tikzstyle{box}=[draw,shape=rectangle,inner sep=2pt,minimum height=6mm,minimum width=6mm,fill=white]
\tikzstyle{dbox}=[draw,doubled,shape=rectangle,inner sep=2pt,minimum height=6mm,minimum width=6mm,fill=white]
\tikzstyle{dmap}=[draw,doubled,shape=NEbox,inner sep=2pt,minimum height=6mm,fill=white]
\tikzstyle{dmapdag}=[draw,doubled,shape=SEbox,inner sep=2pt,minimum height=6mm,fill=white]
\tikzstyle{dmapadj}=[draw,doubled,shape=SEbox,inner sep=2pt,minimum height=6mm,fill=white]
\tikzstyle{dmaptrans}=[draw,doubled,shape=SWbox,inner sep=2pt,minimum height=6mm,fill=white]
\tikzstyle{dmapconj}=[draw,doubled,shape=NWbox,inner sep=2pt,minimum height=6mm,fill=white]

\tikzstyle{ddmap}=[draw,doubled,dashed,shape=NEbox,inner sep=2pt,minimum height=6mm,fill=white]
\tikzstyle{ddmapdag}=[draw,doubled,dashed,shape=SEbox,inner sep=2pt,minimum height=6mm,fill=white]
\tikzstyle{ddmapadj}=[draw,doubled,dashed,shape=SEbox,inner sep=2pt,minimum height=6mm,fill=white]
\tikzstyle{ddmaptrans}=[draw,doubled,dashed,shape=SWbox,inner sep=2pt,minimum height=6mm,fill=white]
\tikzstyle{ddmapconj}=[draw,doubled,dashed,shape=NWbox,inner sep=2pt,minimum height=6mm,fill=white]

\boxshape{sNEbox}{0pt}{3pt}
\boxshape{sSEbox}{0pt}{-3pt}
\boxshape{sNWbox}{3pt}{0pt}
\boxshape{sSWbox}{-3pt}{0pt}
\tikzstyle{smap}=[draw,shape=sNEbox,fill=white]
\tikzstyle{smapdag}=[draw,shape=sSEbox,fill=white]
\tikzstyle{smapadj}=[draw,shape=sSEbox,fill=white]
\tikzstyle{smaptrans}=[draw,shape=sSWbox,fill=white]
\tikzstyle{smapconj}=[draw,shape=sNWbox,fill=white]

\tikzstyle{dsmap}=[draw,dashed,shape=sNEbox,fill=white]
\tikzstyle{dsmapdag}=[draw,dashed,shape=sSEbox,fill=white]
\tikzstyle{dsmaptrans}=[draw,dashed,shape=sSWbox,fill=white]
\tikzstyle{dsmapconj}=[draw,dashed,shape=sNWbox,fill=white]

\boxshape{mNEbox}{0pt}{10pt}
\boxshape{mSEbox}{0pt}{-10pt}
\boxshape{mNWbox}{10pt}{0pt}
\boxshape{mSWbox}{-10pt}{0pt}
\tikzstyle{mmap}=[draw,shape=mNEbox]
\tikzstyle{mmapdag}=[draw,shape=mSEbox]
\tikzstyle{mmaptrans}=[draw,shape=mSWbox]
\tikzstyle{mmapconj}=[draw,shape=mNWbox]

\tikzstyle{mmapgray}=[draw,fill=gray!40!white,shape=mNEbox]
\tikzstyle{smapgray}=[draw,fill=gray!40!white,shape=sNEbox]

\makeatletter

\pgfdeclareshape{cornerpoint}{
\inheritsavedanchors[from=rectangle] 
\inheritanchorborder[from=rectangle]
\inheritanchor[from=rectangle]{center}
\inheritanchor[from=rectangle]{north}
\inheritanchor[from=rectangle]{south}
\inheritanchor[from=rectangle]{west}
\inheritanchor[from=rectangle]{east}
\backgroundpath{
\southwest \pgf@xa=\pgf@x \pgf@ya=\pgf@y
\northeast \pgf@xb=\pgf@x \pgf@yb=\pgf@y

\pgfmathsetmacro{\pgf@shorten@left}{\pgfkeysvalueof{/tikz/shorten left}}
\pgfmathsetmacro{\pgf@shorten@right}{\pgfkeysvalueof{/tikz/shorten right}}

\pgfpathmoveto{\pgfpoint{0.5 * (\pgf@xa + \pgf@xb)}{\pgf@ya - 5pt}}
\pgfpathlineto{\pgfpoint{\pgf@xa - 8pt + \pgf@shorten@left}{\pgf@yb - 1.5 * \pgf@shorten@left}}
\pgfpathlineto{\pgfpoint{\pgf@xa - 8pt + \pgf@shorten@left}{\pgf@yb}}
\pgfpathlineto{\pgfpoint{\pgf@xb + 8pt - \pgf@shorten@right}{\pgf@yb}}
\pgfpathlineto{\pgfpoint{\pgf@xb + 8pt - \pgf@shorten@right}{\pgf@yb - 1.5 * \pgf@shorten@right}}
\pgfpathclose
}
}

\pgfdeclareshape{cornercopoint}{
\inheritsavedanchors[from=rectangle] 
\inheritanchorborder[from=rectangle]
\inheritanchor[from=rectangle]{center}
\inheritanchor[from=rectangle]{north}
\inheritanchor[from=rectangle]{south}
\inheritanchor[from=rectangle]{west}
\inheritanchor[from=rectangle]{east}
\backgroundpath{
\southwest \pgf@xa=\pgf@x \pgf@ya=\pgf@y
\northeast \pgf@xb=\pgf@x \pgf@yb=\pgf@y

\pgfmathsetmacro{\pgf@shorten@left}{\pgfkeysvalueof{/tikz/shorten left}}
\pgfmathsetmacro{\pgf@shorten@right}{\pgfkeysvalueof{/tikz/shorten right}}

\pgfpathmoveto{\pgfpoint{0.5 * (\pgf@xa + \pgf@xb)}{\pgf@yb + 5pt}}
\pgfpathlineto{\pgfpoint{\pgf@xa - 8pt + \pgf@shorten@left}{\pgf@ya + 1.5 * \pgf@shorten@left}}
\pgfpathlineto{\pgfpoint{\pgf@xa - 8pt + \pgf@shorten@left}{\pgf@ya}}
\pgfpathlineto{\pgfpoint{\pgf@xb + 8pt - \pgf@shorten@right}{\pgf@ya}}
\pgfpathlineto{\pgfpoint{\pgf@xb + 8pt - \pgf@shorten@right}{\pgf@ya + 1.5 * \pgf@shorten@right}}
\pgfpathclose
}
}

\makeatother

\pgfkeyssetvalue{/tikz/shorten left}{0pt}
\pgfkeyssetvalue{/tikz/shorten right}{0pt}

\tikzstyle{kpoint common}=[draw,fill=white,inner sep=1pt,minimum height=4mm]
\tikzstyle{kpoint sc}=[shape=cornerpoint,kpoint common]
\tikzstyle{kpoint adjoint sc}=[shape=cornercopoint,kpoint common]
\tikzstyle{kpoint}=[shape=cornerpoint,shorten left=5pt,kpoint common]
\tikzstyle{kpoint adjoint}=[shape=cornercopoint,shorten left=5pt,kpoint common]
\tikzstyle{kpoint conjugate}=[shape=cornerpoint,shorten right=5pt,kpoint common]
\tikzstyle{kpoint transpose}=[shape=cornercopoint,shorten right=5pt,kpoint common]
\tikzstyle{kpoint symm}=[shape=cornerpoint,shorten left=5pt,shorten right=5pt,kpoint common]

\tikzstyle{wide kpoint sc}=[shape=cornerpoint,kpoint common, minimum width=1 cm]
\tikzstyle{wide kpointdag sc}=[shape=cornercopoint,kpoint common, minimum width=1 cm]

\tikzstyle{black kpoint}=[shape=cornerpoint,shorten left=5pt,kpoint common,fill=black,font=\color{white}]

\tikzstyle{black kpoint sm}=[shape=cornerpoint,shorten left=5pt,kpoint common,fill=black,font=\color{white},scale=0.75]

\tikzstyle{black kpoint adjoint}=[shape=cornercopoint,shorten left=5pt,kpoint common,fill=black,font=\color{white}]
\tikzstyle{black kpointadj}=[shape=cornercopoint,shorten left=5pt,kpoint common,fill=black,font=\color{white}]

\tikzstyle{black kpointadj sm}=[shape=cornercopoint,shorten left=5pt,kpoint common,fill=black,font=\color{white},scale=0.75]

\tikzstyle{black dkpoint}=[shape=cornerpoint,shorten left=5pt,kpoint common,fill=black, doubled,font=\color{white}]
\tikzstyle{black dkpoint adjoint}=[shape=cornercopoint,shorten left=5pt,kpoint common,fill=black, doubled,font=\color{white}]
\tikzstyle{black dkpointadj}=[shape=cornercopoint,shorten left=5pt,kpoint common,fill=black, doubled,font=\color{white}]

\tikzstyle{black dkpoint sm}=[shape=cornerpoint,shorten left=5pt,kpoint common,fill=black, doubled,font=\color{white},scale=0.75]
\tikzstyle{black dkpointadj sm}=[shape=cornercopoint,shorten left=5pt,kpoint common,fill=black, doubled,font=\color{white},scale=0.75]

\tikzstyle{kpointdag}=[kpoint adjoint]
\tikzstyle{kpointadj}=[kpoint adjoint]
\tikzstyle{kpointconj}=[kpoint conjugate]
\tikzstyle{kpointtrans}=[kpoint transpose]

\tikzstyle{big kpoint}=[kpoint, minimum width=1.2 cm, minimum height=8mm, inner sep=4pt, text depth=3mm]

\tikzstyle{wide kpoint}=[kpoint, minimum width=1 cm, inner sep=2pt]
\tikzstyle{wide kpointdag}=[kpointdag, minimum width=1 cm, inner sep=2pt]
\tikzstyle{wide kpointconj}=[kpointconj, minimum width=1 cm, inner sep=2pt]
\tikzstyle{wide kpointtrans}=[kpointtrans, minimum width=1 cm, inner sep=2pt]

\tikzstyle{wider kpoint}=[kpoint, minimum width=1.25 cm, inner sep=2pt]
\tikzstyle{wider kpointdag}=[kpointdag, minimum width=1.25 cm, inner sep=2pt]
\tikzstyle{wider kpointconj}=[kpointconj, minimum width=1.25 cm, inner sep=2pt]
\tikzstyle{wider kpointtrans}=[kpointtrans, minimum width=1.25 cm, inner sep=2pt]

\tikzstyle{gray kpoint}=[kpoint,fill=gray!50!white]
\tikzstyle{gray kpointdag}=[kpointdag,fill=gray!50!white]
\tikzstyle{gray kpointadj}=[kpointadj,fill=gray!50!white]
\tikzstyle{gray kpointconj}=[kpointconj,fill=gray!50!white]
\tikzstyle{gray kpointtrans}=[kpointtrans,fill=gray!50!white]

\tikzstyle{gray dkpoint}=[kpoint,fill=gray!50!white,doubled]
\tikzstyle{gray dkpointdag}=[kpointdag,fill=gray!50!white,doubled]
\tikzstyle{gray dkpointadj}=[kpointadj,fill=gray!50!white,doubled]
\tikzstyle{gray dkpointconj}=[kpointconj,fill=gray!50!white,doubled]
\tikzstyle{gray dkpointtrans}=[kpointtrans,fill=gray!50!white,doubled]

\tikzstyle{white label}=[draw,fill=white,rectangle,inner sep=0.7 mm]
\tikzstyle{gray label}=[draw,fill=gray!50!white,rectangle,inner sep=0.7 mm]
\tikzstyle{black label}=[draw,fill=black,rectangle,inner sep=0.7 mm]

\tikzstyle{dkpoint}=[kpoint,doubled]
\tikzstyle{wide dkpoint}=[wide kpoint,doubled]
\tikzstyle{dkpointdag}=[kpoint adjoint,doubled]
\tikzstyle{wide dkpointdag}=[wide kpointdag,doubled]
\tikzstyle{dkcopoint}=[kpoint adjoint,doubled]
\tikzstyle{dkpointadj}=[kpoint adjoint,doubled]
\tikzstyle{dkpointconj}=[kpoint conjugate,doubled]
\tikzstyle{dkpointtrans}=[kpoint transpose,doubled]

\tikzstyle{kscalar}=[kpoint common, shape=EBox, inner xsep=-1pt, inner ysep=3pt,font=\small]
\tikzstyle{kscalarconj}=[kpoint common, shape=WBox, inner xsep=-1pt, inner ysep=3pt,font=\small]

\tikzstyle{spekpoint}=[kpoint sc,minimum height=5mm,inner sep=3pt]
\tikzstyle{spekcopoint}=[kpoint adjoint sc,minimum height=5mm,inner sep=3pt]

\tikzstyle{dspekpoint}=[spekpoint,doubled]
\tikzstyle{dspekcopoint}=[spekcopoint,doubled]

 \tikzstyle{upground}=[circuit ee IEC,thick,ground,rotate=90,scale=2.5]
 \tikzstyle{downground}=[circuit ee IEC,thick,ground,rotate=-90,scale=2.5]
 \tikzstyle{infupground}=[circuit ee IEC,thick,ground,rotate=0,scale=2.5]
 \tikzstyle{infdownground}=[circuit ee IEC,thick,ground,rotate=180,scale=2.5]
 \tikzstyle{bigground}=[regular polygon,regular polygon sides=3,draw=gray,scale=0.50,inner sep=-0.5pt,minimum width=10mm,fill=gray]

\tikzstyle{arrs}=[-latex,font=\small,auto]
\tikzstyle{arrow plain}=[arrs]
\tikzstyle{arrow dashed}=[dashed,arrs]
\tikzstyle{arrow bold}=[very thick,arrs]
\tikzstyle{arrow hide}=[draw=white!0,-]
\tikzstyle{arrow reverse}=[latex-]
\tikzstyle{cdnode}=[]

\tikzstyle{tilde}=[draw=blue]
\tikzstyle{tildelabel}=[text=blue]

\let\olddagger\dagger
\renewcommand{\dagger}{\ensuremath{\olddagger}\xspace}

\usepackage{makeidx}
\makeindex

\theoremstyle{plain}
\newtheorem*{main theorem}{Main Theorem}
\newtheorem{theorem}{Theorem}[section]
\newtheorem{corollary}[theorem]{Corollary}
\newtheorem{lemma}[theorem]{Lemma}
\newtheorem{proposition}[theorem]{Proposition}

\newtheorem{definition}[theorem]{Definition}

\newtheorem{example}[theorem]{Example}

\newtheorem{example*}[theorem]{Example*}
\newtheorem{examples*}[theorem]{Examples*}
\newtheorem{remark}[theorem]{Remark}
\newtheorem{remark*}[theorem]{Remark*}

\newtheorem*{search problem}{Search Problem}

\usepackage{color}
\def\bR{\begin{color}{red}}
\def\bB{\begin{color}{blue}}
\def\bM{\begin{color}{magenta}}
\def\bC{\begin{color}{cyan}}
\def\bW{\begin{color}{white}}
\def\bBl{\begin{color}{black}}
\def\bG{\begin{color}{green}}
\def\bY{\begin{color}{yellow}}
\def\e{\end{color}\xspace}
\newcommand{\bit}{\begin{itemize}}
\newcommand{\eit}{\end{itemize}\par\noindent}
\newcommand{\ben}{\begin{enumerate}}
\newcommand{\een}{\end{enumerate}\par\noindent}
\newcommand{\beq}{\begin{equation}}
\newcommand{\eeq}{\end{equation}\par\noindent}
\newcommand{\beqa}{\begin{eqnarray*}}
\newcommand{\eeqa}{\end{eqnarray*}\par\noindent}
\newcommand{\beqn}{\begin{eqnarray}}
\newcommand{\eeqn}{\end{eqnarray}\par\noindent}

\def\jR{\begin{color}{black}}
\def\jB{\begin{color}{black}}
\def\jM{\begin{color}{magenta}}
\def\jC{\begin{color}{cyan}}
\def\jW{\begin{color}{white}}
\def\jBl{\begin{color}{black}}
\def\jG{\begin{color}{green}}
\def\jY{\begin{color}{yellow}}

\newcommand{\xiNC}{\xi_{\rm{\kern -0.8pt n \kern -0.7pt c}}}
\usepackage{soul}

\begin{document}
\title{\LARGE A structure theorem for generalized-noncontextual ontological models}
\author{David Schmid}
\affiliation{
International Centre for Theory of Quantum Technologies, University of Gda\'nsk, 80-308 Gda\'nsk, Poland}
\affiliation{Perimeter Institute for Theoretical Physics, 31 Caroline Street North, Waterloo, Ontario Canada N2L 2Y5}
\affiliation{Institute for Quantum Computing and Department of Physics and Astronomy, University of Waterloo, Waterloo, Ontario N2L 3G1, Canada}
\author{John H. Selby}
\affiliation{
International Centre for Theory of Quantum Technologies, University of Gda\'nsk, 80-308 Gda\'nsk, Poland}
\author{Matthew F. Pusey}
\affiliation{Department of Mathematics, University of York, Heslington, York YO10 5DD, United Kingdom}
\author{Robert W. Spekkens}
\affiliation{Perimeter Institute for Theoretical Physics, 31 Caroline Street North, Waterloo, Ontario Canada N2L 2Y5}
\begin{abstract}
It is useful to have a criterion for when the predictions of an operational theory should be considered classically explainable.  Here we take the criterion to be that the theory
admits of a generalized-noncontextual ontological model. Existing works on generalized noncontextuality have focused on experimental scenarios having a simple structure: typically, prepare-measure scenarios. Here, we formally extend the framework of ontological models as well as the principle of generalized noncontextuality to arbitrary compositional scenarios. We leverage a process-theoretic framework to prove that, under some reasonable assumptions,  every generalized-noncontextual ontological model of a tomographically local operational theory has a surprisingly rigid and simple mathematical structure---in short, it corresponds to a frame representation which is not overcomplete. One consequence of this theorem is that the largest number of ontic states possible in any such model is given by the dimension of the associated generalized probabilistic theory. This constraint is useful for generating noncontextuality no-go theorems as well as techniques for experimentally certifying contextuality. Along the way, we extend known results concerning the equivalence of different notions of classicality from prepare-measure scenarios to arbitrary compositional scenarios. Specifically, we prove a correspondence between the following three notions of classical explainability of an operational theory:  (i) existence of a noncontextual ontological model for it, (ii) existence of a positive quasiprobability representation for the generalized probabilistic theory it defines, and (iii)  existence of an ontological model for the generalized probabilistic theory it defines.
\end{abstract}
\maketitle
\vspace{-3mm}
{
  \hypersetup{linkcolor=purple}
  \tableofcontents
}

\section{Introduction}

For a given operational theory, under what circumstances is it appropriate to say that its predictions admit of a classical explanation? This article starts with the presumption that this question is best answered as follows:
the operational theory must admit of an ontological model that satisfies the principle of generalized noncontextuality, defined in Ref.~\cite{Spekkens2005}.
Admitting of a generalized-noncontextual ontological model subsumes several other notions of classical explainability,
such as admitting of a positive quasiprobability representation~\cite{Spekkens2008,ferrie2008frame}, being embeddable in a simplicial generalized probabilistic theory (GPT)~\cite{schmid2019characterization,
shahandeh2019contextuality,
selby2021contextuality,selby2021accessible}, and admitting of a locally causal model~\cite{Bell,Bellreview}.   (Note that the first two of these results are first proved for general compositional scenarios in this paper.)  Additionally, generalized noncontextuality can be motivated as an instance of a methodological principle for theory construction due to Leibniz, as argued in Ref.~\cite{Leibniz} and the appendix of Ref.~\cite{Mazurek2016}.
Finally, operational theories that fail to admit of a generalized-noncontextual ontological model provide advantages for information processing relative to their classically explainable counterparts~\cite{POM,RAC,RAC2,
Saha_2019,
saha2019preparation,
MESD,
Lostaglio2020contextualadvantage,schmid2021only}.  Because the notion of generalized noncontextuality is the only one we consider in this article, we will often refer to it simply as `noncontextuality'.

 To date, prepare-measure scenarios are the experimental arrangements for which the consequences of generalized noncontextuality have been most explored.
A few works have also studied experiments where there is a transformation or an instrument intervening between the preparation and the measurement~\cite{lillystone2019single,PP1,PP2,AWV,AWVrobust,Lostaglio2020contextualadvantage}.
However, generalized noncontextuality has not previously been considered in
experimental scenarios wherein the component procedures are connected together in arbitrary ways, that is, in arbitrary compositional scenarios. Indeed, generalized noncontextuality has not even been formally defined at the level of compositional theories prior to this work; rather, it and several related concepts have only been formally defined for particular types of scenarios. In this work, we give a process-theoretic~\cite{coecke2015categorical,coecke2017picturing,selbyReconstruction,gogioso2017categorical} formulation of the various relevant notions of operational theories and of representations thereof, enabling the study of noncontextuality in arbitrary compositional scenarios, and indeed of the noncontextuality of operational {\em theories} themselves. We then derive a number of results regarding the structure of noncontextual representations of operational theories, and we ultimately put strong constraints on the nature of these representations.

Like Ref.~\cite{schmid2019characterization}, this work is sensitive to the distinction between operational theories and {\em quotiented} operational theories, commonly termed {\em generalized probabilistic theories} (or GPTs)~\cite{hardy2001quantum,barrett2007,
hardy2011reformulating,gogioso2017categorical,selbyReconstruction}. In an  operational theory, one understands the primitive processes (e.g., preparation,  transformation, and measurement procedures) to be lists of laboratory instructions detailing actions that can be taken on some physical system. Such a theory also makes predictions for the statistics of outcomes in any given experimental arrangement (without making any attempt to {\em explain} these predictions). As lists of laboratory instructions, the processes in an operational theory contain details which are not relevant to the observed statistics; any such details are termed the {\em context} of the given process~\cite{Spekkens2005}. In contrast, a {\em quotiented} operational theory, or GPT, arises when one removes this context information by identifying any two processes that differ only by context---that is, which lead to all the same statistical predictions, and so are said to be operationally equivalent. 

Our formalization of both operational theories and generalized probabilistic theories follows that of quotiented and unquotiented operational probabilistic theories (OPTs), as in Refs.~\cite{chiribella2010probabilistic,chiribella2011informational,chiribella2016quantum}. For completeness, we provide a pedagogical introduction to our notation and conventions in Sec.~\ref{sec:preliminaries}. The framework presented here is also a precursor to a more novel framework presented in Ref.~\cite{schmid2020unscrambling}, which is motivated by the objective of cleanly separating the causal and inferential aspects of a theory. 

There are multiple different representations of operational theories and generalized probabilistic theories that one can consider, often motivated by the aim of {\em explaining} the predictions of the theory by appealing to some underlying realist model of reality. The quintessential sort of explanation is an ontological model of an operational theory, which presumes that the systems passing between experimental devices have properties, and that the outcomes of measurements reveal information about these properties. A complete specification of these properties for a given system is termed its {\em ontic state}. The variability in this ontic state mediates causal influences between the devices. Ontological models may be defined for either  operational theories or generalized probabilistic theories. Another type of representation which has been widely considered (particularly by the quantum optics community) is that of quasiprobabilistic representations. Such representations are only defined for generalized probabilistic theories, and can be viewed as ontological models using quasiprobability distributions---that is, analogues of probability distributions in which some of the values can be negative. 

Our formalization of ontological models is more general than that which is usually given, since we define them in a compositional manner, and for arbitrary theories rather than for particular scenarios.  (Although note that this was already done for the special case of quantum theory in Ref.~\cite{gheorghiu2019ontological}.)  Furthermore, our formalization of quasiprobabilistic representations is more general than that which is usually given, since we define them for arbitrary GPTs (not necessarily quantum). In this latter case, our formalization was strongly influenced by the prescription suggested in Ref.~\cite{van2017quantum}.

We can now return to the question of when a theory's predictions admit of a classical explanation. 

As argued above, our guiding principle is that of noncontextuality. The principle of noncontextuality is a constraint on ontological models of  operational theories: namely, that 
the representation of operational processes does not depend on their context. That is, operational processes which lead to identical predictions about operational facts are represented by ontological processes which lead to identical predictions about ontological facts. If such a noncontextual ontological model exists, we take it to be a classical realist explanation of the predictions of the  operational theory. Hence, the notion of classical-explainability for  operational theories is the existence of a generalized-noncontextual ontological model.

If one takes the processes of a generalized probabilistic theory as the domain of one's representation map, then there is no context on which a given representation (be it an ontological model or a quasiprobabilistic representation) could conceivably depend. 
This point was first made in Ref.~\cite{schmid2019characterization}, and we expand on it in this work, in particular in Appendix~\ref{contextsingpts}.  In such an approach, one cannot directly take noncontextuality---independence of context---as a notion of classicality for generalized probabilistic theories. Still, Ref.~\cite{schmid2019characterization} showed that, in prepare-measure scenarios, the notion of noncontextuality for operational theories induces a natural and equivalent notion of classicality for generalized probabilistic theories.  In particular, a generalized probabilistic theory admits of a classical explanation if and only if there is a simplex that embeds its state space and furthermore the hypercube of effects that is dual to this simplex embeds its effect space. Such an embedding can be viewed as an ontological model of a GPT~\cite{schmid2019characterization,shahandeh2019contextuality}. Hence, the resulting notion of classical explainability for a GPT is the existence of an ontological model for it. Our work extends this result to the case of arbitrary compositional theories and scenarios.

We also extend (from prepare-measure scenarios to arbitrary compositional scenarios) the proof that positive quasiprobabilistic models are in one-to-one correspondence with noncontextual ontological models~\cite{Spekkens2008,schmid2019characterization}. Note that our proof---like the special case given in Ref.~\cite{schmid2019characterization}---corrects some issues with the original arguments in Ref.~\cite{Spekkens2008}.\footnote{
Ref.~\cite{Spekkens2008} was not careful to distinguish between quotiented and unquotiented operational theories, and as such did not stipulate whether quantum theory was being considered as an operational theory or as a GPT.  As a result, it failed to note that the most natural domain for a quasi-probabilistic representation is quantum theory {\em as a GPT}, while the domain of a noncontextual ontological representation is necessarily quantum theory {\em as an operational theory}. Also as a consequence, it argued that a positive quasiprobability representation {\em is the same thing as a} noncontextual ontological model. Our recasting of the relation between quasiprobability representations and noncontextual ontological representations is explicit about such distinctions, and consequently we show that a positive quasiprobability representation is {\em not} in and of itself a noncontextual ontological model. Rather, it is just that the sets of these are in one-to-one correspondence.}

{Note that simplex-embeddability can also be motivated as a notion of classicality as follows.  First, a simplicial GPT, i.e., one in which all of the state spaces are simplices, transformations are arbitrary convex-linear maps between these simplices, effects are elements of the hypercubes which are dual to the simplices, and which is tomographically local, 
has been argued to capture a notion of classicality among operational theories~\cite{hardy2001quantum,barrett2007}. If a GPT satisfies simplex-embeddabiltiy, then it follows that the set of states and the set of effects therein can be conceptualized as a subset of those arising in a simplicial GPT, implying that every experiment describable by the GPT can be \emph{simulated} within the simplicial GPT.  It follows that simplex embeddability captures the possibility of simulatability within a classical operational theory, hence a notion of classicality.}
 Furthermore, the existence of a positive quasiprobabilistic representation is a notion of classicality in the sphere of quantum optics. Hence, our results ultimately show that three independently motivated notions of classicality (namely these two, and the existence of a noncontextual ontological model of an operational theory) all coincide in general compositional situations (such as is relevant, for example, in quantum computation).

Most importantly, this equivalence allows us to prove that every noncontextual ontological model of a tomographically local operational theory which satisfies
 an assumption of {\em diagram preservation}  has a rigid and simple mathematical structure.
In particular, every such model is given by a diagram-preserving positive quasiprobabilistic model of the GPT associated with the operational theory, and we prove that every such quasiprobabilistic model is in turn a {\em frame representation}~\cite{ferrie2008frame,Ferrie_2009} that is not overcomplete.
As a corollary, it follows that the number of ontic states in any such model is no larger than the dimension of the GPT space.

This rigid structure theorem and bound on the number of ontic states shows that there is much less freedom in constructing noncontextual ontological models than previously  recognized. In particular, it means that once the representation of the states is fixed (i.e., by choice of frame) then \emph{there is no remaining freedom in the representations of the measurements and transformations}. Moreover, in many ontological models the number of ontic states is taken to be infinite (e.g., corresponding to points on the surface of the Bloch ball); however, all such models are immediately ruled out by our bound on the number of ontic states. These results also imply new proofs of the fact that operational quantum theory does not admit of a generalized-noncontextual model and simplifies the problem of witnessing generalized contextuality experimentally.

 Categorically, we view the GPT as being a particular monoidal category, and the representations thereof as being particular strong monoidal functors into subcategories of $\RL$ (the category of linear maps between finite dimensional real vector spaces).  In the case of tomographically local GPTs, our structure theorem states that any such functor is naturally isomorphic to a standard representation of a tomographically local GPT within $\RL$. In particular, this means that ontological models for such theories, should they exist, are essentially unique.    

We now summarize our key results and main assumptions in more detail. 

\subsection{Results}

  We begin by providing informal statements of our main results. The first result in this list extends the results of Ref.~\cite{schmid2019characterization} from the case of prepare-measure scenarios to arbitrary scenarios. The 
  second entry in the list constitutes the main technical result of this work. The third is
a primary consequence of this for the study of noncontextual ontological models.

\ben
\item  We refine and generalize the notions of quasiprobabilistic models and ontological models of operational theories and of GPTs to arbitrary compositional scenarios and theories, and we show a triple equivalence between:
\ben
\item a positive quasiprobabilistic model of the GPT associated to the operational theory, 
 \item an ontological model of the GPT associated to the operational theory, and
\item a noncontextual ontological model of the operational theory.
\een
\item We then prove a structure theorem for representations of a GPT which implies that:
\ben
\item every diagram-preserving quasiprobabilistic model of a GPT is a frame representation that is not overcomplete, i.e., an {\em exact}  frame representation
\item every diagram-preserving ontological model of a GPT is a {\em positive} exact frame representation, and
\item  every diagram-preserving noncontextual ontological model of an operational theory can be used to construct a positive exact frame representation of the associated GPT, and vice versa. 
\een
\item A key corollary of these is that the cardinality of the set of ontic states for a given system in
any diagram-preserving ontological model is equal to the dimension of the state space of that system in
the GPT. For instance, a noncontextual ontological model of a qudit must have exactly $d^2$ ontic states. Similarly, the dimension of the sample space of any diagram-preserving quasiprobabilistic model is the GPT dimension.
\een 

These results show that by moving beyond prepare-measure scenarios, the concept of a noncontextual ontological model of an operational theory becomes constrained to a remarkably specific and simple mathematical structure. Moreover, our bound on the number of ontic states yields new proofs of the impossibility of a noncontextual model of quantum theory (e.g., via Hardy's ontological excess baggage theorem~\cite{Hardy2004}) and dramatic simplifications to algorithms for witnessing contextuality in experimental data (e.g. reducing the algorithm introduced in Ref.~\cite{schmid2019characterization} from a hierarchy of tests to a single test).

\subsection{Assumptions} \label{secassumptions}

The assumptions that are needed to prove our results
will be formally introduced as they become relevant.
 For the sake of having a complete list in one place, however, we provide an informal account of them here. These assumptions can be divided into two categories.

First, we have assumptions limiting the sorts of operational theories that we are considering.
\ben
\item {Unique deterministic effect:} We consider only operational theories
in which all deterministic effects (corresponding to implementing a measurement on the system and marginalizing over its outcome) are operationally equivalent~\cite{chiribella2010probabilistic}.
\item {Arbitrary mixtures:} We assume that every mixture of procedures within an operational theory is also an effective procedure within that operational theory. That is, for any pair of procedures in the theory, there exists a third procedure defined by flipping
a weighted coin and choosing to implement either the first or the second, depending on the outcome of the coin flip.
\item  {Finite dimensionality:} We assume that the dimension of the GPT associated to the operational theory is finite.
\item {Tomographic locality:} For some of our results, we moreover limit our analysis to operational theories 
 whose corresponding GPT is tomographically local (namely, where all GPT processes on composite systems can be fully characterized by probing the component systems locally \cite{hardy2001quantum}).  
\een

Second, we have assumptions that concern the ontological model (or quasiprobabilistic model).
\ben
\item {Deterministic effect preservation:} Any deterministic effect in the operational theory is represented by marginalization over the  sample 
 space of the system in the ontological (or quasiprobabilistic) model.
\item {Convex-Linearity:} The representation of a mixture of procedures is given by the mixture of their representations, and the representation of a coarse-graining of effects is given by the coarse-graining of their representations.
\item {Empirical adequacy:} The ontological (or quasiprobabilistic) representations must make the same predictions as the operational theory.
\item {Diagram preservation:} The compositional structure of the ontological (or quasiprobabilistic) representation must be the same as the compositional structure of the operational theory.  (Formally, this means that we take these representations to be strong monoidal functors.) 
\een

The most significant assumption regarding the scope of operational theories to which our results apply is that of tomographic locality.   Among the assumptions concerning the nature of the ontological (or quasiprobablistic) model, the only one that is not completely standard is that of diagram preservation.

As we will explain,
however, the assumption of
diagram preservation does not restrict the scope of applicability of our results; rather, it is a prescription for how one is to apply our formalism to a given scenario.
Furthermore, our main results do not require the full power of diagram preservation, but rather can be derived from the application of this assumption to a few simple scenarios: the identity operation, the prepare-measure scenario, and the measure-and-reprepare operation.
However, full diagram preservation is a natural generalization of these assumptions, as well as of a number of other standard assumptions that have been made throughout the literature on ontological models, and so we will build it into our definitions rather than endorsing only those particular instances that we need for the results in this paper.
We discuss these points in more detail in Section~\ref{revisassump}, and 
provide a defense of full diagram preservation in 
Ref.~\cite{schmid2020unscrambling}.

\section{Preliminaries}\label{sec:preliminaries}

In this section we provide a pedagogical introduction to the diagrammatic notation that we will employ and its application to operational theories, tomographically local GPTs, and their ontological and quaisiprobabilistic representations. This section should be treated largely as a review of the relevant literature, which we include to have a self-contained presentation of the necessary formalism for our main results. 

\subsection{Process theories}
In this paper we will represent
various types of theories as \emph{process theories}~\cite{coecke2017picturing}, which highlights the compositional structures within these theories.
We will express certain relationships that hold between these process theories in terms of \emph{diagram-preserving maps}\footnote{This work could also be presented in the language of category theory. In particular, process theories can be viewed as symmetric monoidal categories~\cite{coecke2017picturing} and diagram preserving maps as strong monoidal functors between them~\cite{mellies2006functorial}.}. We give a brief introduction to this formalism here. Readers who would like a deeper understanding of this approach can read, for example, Refs.~\cite{coecke2017picturing,coecke2015categorical,selbyReconstruction,gogioso2017categorical,schmid2020unscrambling}.

A process theory $\mathcal{P}$ is specified by a collection of systems $A, B, C, ...$ and a collection of \emph{processes} on these systems.
We will represent the processes diagrammatically, e.g.,

\begin{equation} \left\{
\right\} \subset \mathcal{P}, \end{equation}
where we work with the convention that input systems are at the bottom and output systems are at the top. We will sometimes drop system labels, when it is clear from context. Processes with no input are known as \emph{states}, those with no outputs as \emph{effects}, and those with neither inputs nor outputs as \emph{scalars}. The key feature of a process theory is that this collection of processes $\mathcal{P}$ is closed under wiring processes together to form \emph{diagrams}; for example,
\begin{equation}
%
 \in \mathcal{P}. \end{equation}
Wirings of processes must connect outputs to inputs such that systems match and no cycles are created.

We will commonly draw `clamp'-shaped higher-order processes \cite{chiribella2008quantum,
chiribella2008transforming} such as:
\beq
%
\eeq
which we call \emph{testers}. These can be thought of as something which maps a process from $A$ to $B$ to a scalar  via:
\beq
%
.
\eeq
These are not primitive notions within the framework of process theories and instead are always thought of as being built out of particular state, effect and auxiliary system\footnote{The recent work of Ref.~\cite{wilson2021causality} has shown how these, and other, higher order processes can actually be incorporated as primitive notions within a process-theoretic framework.}. In other words, the tester $\tau$ is really just shorthand notation for a triple $(s_\tau,e_\tau,W_\tau)$ where:
\beq\label{eq:tester}
%
.
\eeq

A diagram-preserving map \colorbox{PineGreen!20}{$M : \mathcal{P} \to \mathcal{P}'$} from one process theory $\mathcal{P}$ to another, $\mathcal{P}'$, is defined as a map taking
systems in $\mathcal{P}$ to systems in $\mathcal{P}'$, denoted as
\beq
S \to M(S),
\eeq
and processes in $\mathcal{P}$ to processes in $\mathcal{P}'$. Taking inspiration from \cite{fritz2018bimonoidal,mellies2006functorial} this will be depicted diagrammatically as
\beq\label{eq:functorApp}
M ::%
 \mapsto \input{Diagrams/3_DPM3-1.tikz}
,
\eeq
such that wiring together processes before or after applying the map is equivalent, that is:
\begin{equation}\label{eq:DPM}
%
, \end{equation}
 and such that wirings are mapped to wirings, e.g.
\begin{equation}
%
.
\end{equation} 
In particular, this means that $M$ maps the identity processes in $\mathcal{P}$ to identity processes in $\mathcal{P}'$.
\begin{remark} If we interpret these process theories as symmetric monoidal categories, then any strict monoidal functor $\bar{M}$ defines a diagram preserving map $M$, simply by taking $M(A)=\bar{M}(A)$ and $M(f)=\bar{M}(f)$. Note that this latter equation is not obviously well-typed, as, according to Eq.~\eqref{eq:functorApp}, $M(f):M(A_1)\otimes\cdots \otimes M(A_n)\to M(B_1)\otimes\cdots M(B_m)$ while $\bar{M}(f):M(A_1\otimes\cdots \otimes A_n)\to M(B_1\otimes\cdots\otimes B_m)$. However, in the case of strict monoidal functors, 
we have that $\bar{M}(A_1\otimes\cdots \otimes A_n)=\bar{M}(A_1)\otimes\cdots\otimes \bar{M}(A_n)$, and so this is not actually a problem.  If, on the other hand, one instead has a strong monoidal functor, $\bar{M}$, in which this equality is relaxed to the natural isomorphism~\cite{mac2013categories} $\mu$, $\bar{M}(A_1\otimes\cdots \otimes A_n)\stackrel{\mu}{\cong}\bar{M}(A_1)\otimes\cdots\otimes \bar{M}(A_n)$, one can still use this to define a diagram preserving map. The difference is that now (following Ref.~\cite{mellies2006functorial}) we need to use the natural isomorphisms $\mu$ in order to define $M(f)$, that is, we define $M(f):= \mu^{-1}\circ \bar{M}(f) \circ \mu$. In this paper we will always be considering strong monoidal functors where $\bar{M}(I)=I$, but if one merely had that $\bar{M}(I)\stackrel{\epsilon}{\cong} I$, then one must also incorporate this natural isomorphism when defining the action of $M$ on states, effects, and scalars.
\end{remark}

We will also use the concept of a sub-process theory, where the intuitive idea is that  $\mathcal{P}'$ is a sub-process theory of $\mathcal{P}$, denoted $\mathcal{P}' \subseteq \mathcal{P}$, if the processes in $\mathcal{P}'$ are a subset of the processes in $\mathcal{P}$ that are themselves closed under forming diagrams. Formally, we do not require that a sub-process theory \emph{is} such a theory, but only that it is \emph{equivalent} to such a theory; that is, we say that $\mathcal{P'}\subseteq \mathcal{P}$ if and only if there exists a faithful strong monoidal functor from $\mathcal{P'}$ into $\mathcal{P}$.  

The key process theory which underpins this work is $\RL$, defined as follows:

\begin{example}[$\RL$]
Systems are labeled by finite dimensional 
real vector spaces $V$ where the composition of systems $V$ and $W$ is given by the tensor product $V\otimes W$. Processes
are defined as linear maps from the input vector space to the output vector space. Composing two processes in sequence corresponds to composing the linear maps, while composing them in parallel corresponds to the tensor product of the maps. If a process lacks an input and/or an output then we view them as linear maps to or from the one-dimensional vector space $\mathds{R}$. Hence, processes with no input correspond to vectors in $V$ and processes with no output to covectors, i.e., elements of $V^*$. This implies that {\em scalars}---processes with neither inputs nor outputs---correspond to real numbers. $\RL$ is equivalent to the process theory of real-valued matrices. However, representing the former in terms of the latter requires artificially choosing a preferred basis for the vector spaces.
\end{example}

The first is the process theory of (sub)stochastic processes. Here, systems are labeled by finite sets
$\Lambda$ which compose via the Cartesian product. Processes with input $\Lambda$ and output $\Lambda'$ correspond to (sub)stochastic maps, and can be thought of as functions
\beq
f:\Lambda\times\Lambda' \to [0,1] :: (\lambda,\lambda')\mapsto f(\lambda'|\lambda)
\eeq
where for all $\lambda\in\Lambda$ we have $\sum_{\lambda'\in\Lambda'} f(\lambda'|\lambda) \leq 1$. When this inequality is an equality, they are said to be {\em stochastic} (rather than substochastic).
For any pair of functions $f: \Lambda \times \Lambda'\to[0,1]$ and $g: \Lambda' \times \Lambda''\to[0,1]$ (where the output type of $f$ matches the input type of $g$), sequential composition is given by $g \circ f:\Lambda\times\Lambda''\to[0,1]$ via the following rule for composing the functions:
\beq
g\circ f (\lambda''|\lambda) := \sum_{\lambda'\in\Lambda'} g(\lambda''|\lambda')f(\lambda'|\lambda).
\eeq
For any pair of functions $f: \Lambda \times \Lambda'\to[0,1]$ and $g: \Lambda'' \times \Lambda'''\to[0,1]$, parallel composition is given by $g \otimes f:(\Lambda'\times\Lambda) \times (\Lambda'''\times\Lambda'')\to[0,1]$ via:
\beq
g\otimes f ((\lambda''',\lambda'')|(\lambda',\lambda)):= g(\lambda'''|\lambda')f(\lambda''|\lambda).
\eeq

It is sometimes more convenient or natural to take an alternative (but equivalent) point of view on this process theory (e.g., this view makes it more clear that this is a sub-process theory of $\RL$). In this alternative view, the systems are not simply given by finite sets $\Lambda$, but rather are taken to be the vector space of functions from $\Lambda$ to $\mathds{R}$, denoted $\mathds{R}^\Lambda$. Then, rather than taking the processes to be functions $f:\Lambda\times \Lambda'\to[0,1]$, one takes them to be linear maps from $\mathds{R}^\Lambda$ to $\mathds{R}^{\Lambda'}$, denoted by:
\beq
\mathbf{f}:\mathds{R}^\Lambda \to \mathds{R}^{\Lambda'}:: \mathbf{v}\mapsto \mathbf{f}(\mathbf{v})
\eeq
where for all $\lambda'\in\Lambda'$, we define
$\mathbf{f(v)}(\lambda') := \sum_{\lambda\in\Lambda} f(\lambda'|\lambda)\mathbf{v}(\lambda)$. It is then straightforward to show that sequential composition of the stochastic processes corresponds to composition of the associated linear maps and that parallel composition of the stochastic processes corresponds to the tensor product of the associated linear maps. For example, for sequential composition we have that for all $\mathbf{v}\in\mathds{R}^\Lambda$,
\begin{align}
\mathbf{(g\circ f)(v)}(\lambda'') &= \sum_{\lambda\in\Lambda}g\circ f(\lambda''|\lambda)\mathbf{v}(\lambda) \\
&= \sum_{\lambda\in\Lambda}\sum_{\lambda'\in\Lambda'}g(\lambda''|\lambda')f(\lambda'|\lambda)\mathbf{v}(\lambda)\\
&= \sum_{\lambda'\in\Lambda'}g(\lambda''|\lambda')\left(\sum_{\lambda\in\Lambda}f(\lambda'|\lambda)\mathbf{v}(\lambda)\right)\\
&= \sum_{\lambda'\in\Lambda}g(\lambda''|\lambda')\mathbf{f(v)}(\lambda)\\
&= \mathbf{g(f(v))}(\lambda)
\end{align}
Moreover, consider processes with no input---that is, linear maps $\mathbf{p}:\mathds{R}\to\mathds{R}^\Lambda$. By using the trivial isomorphism $\mathds{R}\cong \mathds{R}^\star$ where $\star$ is the singleton set $\star:=\{*\}$, one can see that these correspond to functions $p:\star\times\Lambda\to[0,1]$ such that $\sum_{\lambda\in\Lambda}p(\lambda|*)\leq 1$; following standard conventions, we can denote this as $p(\lambda):= p(\lambda|*)$. So $\mathbf{p}$ just corresponds to a subnormalised probability distribution, as expected. Similarly, processes with no output, such as $\mathbf{r}:\mathds{R}^\Lambda\to\mathds{R}$,  correspond to \emph{response functions}, that is, functions $r:\Lambda\times\star\to[0,1]$ such that for all $\lambda$ one has $r(\lambda):=r(*|\lambda) \in [0,1]$. Finally, it follows that processes with neither inputs nor outputs, $\mathbf{s}:\mathds{R}\to\mathds{R}$, correspond to elements of $[0,1]$, i.e., to probabilities.

Summarizing the above, we have:
\begin{example}[$\SubS$]
We define $\SubS$ as a subtheory of $\RL$ where systems are restricted to vector spaces of the form $\mathds{R}^\Lambda$ and processes are restricted to those that correspond to (sub)stochastic maps.
\end{example}

The second subtheory of $\RL$ is $\QSS$, which is the same as the process theory of (sub)stochastic processes, but where the constraint of positivity is dropped. The systems can be taken to be finite sets $\Lambda$, and the processes with input $\Lambda$ and output $\Lambda'$ can be taken to be
functions
\beq
f:\Lambda\times \Lambda' \to \mathds{R} :: (\lambda,\lambda')\mapsto f(\lambda'|\lambda).
\eeq
These are said to be quasistochastic (as opposed to quasisubstochastic) if they moreover satisfy  $\sum_{\lambda\in\Lambda'}f(\lambda'|\lambda) = 1$ for all $\lambda\in\Lambda$. The way that these compose and are represented in $\RL$ is exactly the same as in the case of substochastic maps.

Summarizing, we have
\begin{example}[$\QSS$]
We define $\QSS$ as the subtheory of $\RL$ where systems are restricted to vector spaces of the form $\mathds{R}^\Lambda$ and processes are those that correspond to quasi(sub)stochastic maps.
\end{example}

By construction, $\SubS \subset \QSS \subset \RL$; however, in contrast to $\RL$, $\SubS$ and $\QSS$ {\em do} come equipped with a preferred basis for each system.  It is known that quantum theory as a GPT ($\mathbf{QT}$) can be represented as a subtheory of $\QSS$ (see, for example, \cite{van2017quantum}).

\subsection{Operational theories}

We now introduce a process-theoretic presentation of the framework of operational theories as defined in  Ref.~\cite{Spekkens2005},
resulting in a framework that is essentially that of (unquotiented) operational probabilistic theories~\cite{chiribella2016quantum}.
An operational theory $\Op$ is given by a process theory specifying a set of physical systems and the processes which act on them (where processes are viewed as lists of lab instructions), together with a rule for assigning probabilities to any closed process.
A generic laboratory procedure has an associated set of inputs and outputs, and will be denoted diagrammatically as:
\tikeq[.]{transf}
Of special interest are processes with no inputs and processes with no outputs, depicted respectively as
\begin{equation}
\begin{tikzpicture}
	\begin{pgfonlayer}{nodelayer}
		\node [style=point] (0) at (0, -0.5) {$P$};
		\node [style=none] (1) at (0, 0.5) {};
	\end{pgfonlayer}
	\begin{pgfonlayer}{edgelayer}
		\draw [style=qWire] (0) to (1.center);
	\end{pgfonlayer}
\end{tikzpicture}\ ,\
\begin{tikzpicture}
	\begin{pgfonlayer}{nodelayer}
		\node [style=copoint] (0) at (0, 0.5) {$E$};
		\node [style=none] (1) at (0, -0.5) {};
	\end{pgfonlayer}
	\begin{pgfonlayer}{edgelayer}
		\draw [style=qWire] (0) to (1.center);
	\end{pgfonlayer}
\end{tikzpicture}\
\label{prep}.
\end{equation}
The former is viewed as a \emph{preparation procedure} and the latter is viewed as an effect, corresponding to some outcome of some measurement.
We depict the probability rule by a map $p$, as
\beq
\input{Diagrams/probrule.tikz} = \Pr(E,P) \in [0,1]
\label{probrule}.
\eeq
That is, the application of $p$ on any closed diagram yields a real number between $0$ and $1$. Note that this is not a diagram-preserving map as it can only be applied to processes with no input and no output. (Nonetheless, we will see shortly how it has a diagram-preserving extension to arbitrary processes---namely, the quotienting map).

This probability rule must be
compatible with certain relations that hold between procedures~\cite{chiribella2008quantum,
hardy2011reformulating}.
First, it must factorise over separated diagrams, for example,
\begin{align}
\begin{tikzpicture}
	\begin{pgfonlayer}{nodelayer}
		\node [style=copoint] (0) at (0, 0.5) {$E_1$};
		\node [style=point] (1) at (0, -0.5) {$P_1$};
		\node [style=none] (2) at (-1.25, 1.5) {};
		\node [style=none] (3) at (3, 1.5) {};
		\node [style=none] (4) at (3, -1.5) {};
		\node [style=none] (5) at (-1.25, -1.5) {};
		\node [style=none] (6) at (2.75, -1.25) {\tiny $p$};
		\node [style=point] (7) at (1.75, -0.5) {$P_2$};
		\node [style=copoint] (8) at (1.75, 0.5) {$E_2$};
	\end{pgfonlayer}
	\begin{pgfonlayer}{edgelayer}
		\filldraw[fill=blue!20,draw=blue!40] (2.center) to (5.center) to (4.center) to (3.center) to cycle;
		\draw [style=qWire] (0) to (1);
		\draw [style=qWire] (8) to (7);
	\end{pgfonlayer}
\end{tikzpicture}
\quad &= \quad
\begin{tikzpicture}
	\begin{pgfonlayer}{nodelayer}
		\node [style=copoint] (0) at (0, 0.5) {$E_1$};
		\node [style=point] (1) at (0, -0.5) {$P_1$};
		\node [style=none] (2) at (-1.25, 1.5) {};
		\node [style=none] (3) at (1.25, 1.5) {};
		\node [style=none] (4) at (1.25, -1.5) {};
		\node [style=none] (5) at (-1.25, -1.5) {};
		\node [style=none] (6) at (1, -1.25) {\tiny $p$};
	\end{pgfonlayer}
	\begin{pgfonlayer}{edgelayer}
		\filldraw[fill=blue!20,draw=blue!40] (2.center) to (5.center) to (4.center) to (3.center) to cycle;
				\draw [style=qWire] (0) to (1);
	\end{pgfonlayer}
\end{tikzpicture}
\
\begin{tikzpicture}
	\begin{pgfonlayer}{nodelayer}
		\node [style=copoint] (0) at (0, 0.5) {$E_2$};
		\node [style=point] (1) at (0, -0.5) {$P_2$};
		\node [style=none] (2) at (-1.25, 1.5) {};
		\node [style=none] (3) at (1.25, 1.5) {};
		\node [style=none] (4) at (1.25, -1.5) {};
		\node [style=none] (5) at (-1.25, -1.5) {};
		\node [style=none] (6) at (1, -1.25) {\tiny $p$};
	\end{pgfonlayer}
	\begin{pgfonlayer}{edgelayer}
		\filldraw[fill=blue!20,draw=blue!40] (2.center) to (5.center) to (4.center) to (3.center) to cycle;
				\draw [style=qWire] (0) to (1);
	\end{pgfonlayer}
\end{tikzpicture}\\
&=\quad\Pr(E_1,P_1)\Pr(E_2,P_2).
\end{align}
Moreover, if $T_1$ is a procedure
that is a mixture of $T_2$ and $T_3$ with weights $\omega$ and $1-\omega$ respectively\footnote{That is, either $T_2$ or $T_3$ is implemented, as determined by the outcome of a weighted coin-flip.}, then it must hold that for any tester $\tau$,
we have
\beq
\begin{tikzpicture}
	\begin{pgfonlayer}{nodelayer}
		\node [style=small box] (0) at (2.25, 0) {$T_1$};
		\node [style=none] (1) at (2.25, 0.75) {};
		\node [style=none] (2) at (2.25, -0.75) {};
		\node [style=none] (3) at (1.75, 0.75) {};
		\node [style=none] (4) at (1.75, 1.25) {};
		\node [style=none] (5) at (3.5, 1.25) {};
		\node [style=none] (6) at (3.5, -1.25) {};
		\node [style=none] (7) at (1.75, -1.25) {};
		\node [style=none] (8) at (1.75, -0.75) {};
		\node [style=none] (9) at (3, 0.75) {};
		\node [style=none] (10) at (3, -0.75) {};
		\node [style=none] (11) at (3.25, 0) {$\tau$};
		\node [style=none] (12) at (1, 2) {};
		\node [style=none] (13) at (4.25, 2) {};
		\node [style=none] (14) at (4.25, -2) {};
		\node [style=none] (15) at (1, -2) {};
		\node [style=none] (16) at (4, -1.75) {\tiny $p$};
	\end{pgfonlayer}
	\begin{pgfonlayer}{edgelayer}
		\filldraw[fill=blue!20,draw=blue!40] (12.center) to (15.center) to (14.center) to (13.center) to cycle;
	\filldraw[fill=white,draw=black] (7.center) to (6.center) to (5.center) to (4.center) to (3.center) to (9.center) to (10.center) to (8.center) to cycle;	
		\draw [style=qWire] (0) to (1.center);
		\draw [style=qWire] (0) to (2.center);
	\end{pgfonlayer}
\end{tikzpicture}
\ =\omega\
\begin{tikzpicture}
	\begin{pgfonlayer}{nodelayer}
		\node [style=small box] (0) at (2.25, 0) {$T_2$};
		\node [style=none] (1) at (2.25, 0.75) {};
		\node [style=none] (2) at (2.25, -0.75) {};
		\node [style=none] (3) at (1.75, 0.75) {};
		\node [style=none] (4) at (1.75, 1.25) {};
		\node [style=none] (5) at (3.5, 1.25) {};
		\node [style=none] (6) at (3.5, -1.25) {};
		\node [style=none] (7) at (1.75, -1.25) {};
		\node [style=none] (8) at (1.75, -0.75) {};
		\node [style=none] (9) at (3, 0.75) {};
		\node [style=none] (10) at (3, -0.75) {};
		\node [style=none] (11) at (3.25, 0) {$\tau$};
		\node [style=none] (12) at (1, 2) {};
		\node [style=none] (13) at (4.25, 2) {};
		\node [style=none] (14) at (4.25, -2) {};
		\node [style=none] (15) at (1, -2) {};
		\node [style=none] (16) at (4, -1.75) {\tiny $p$};
	\end{pgfonlayer}
	\begin{pgfonlayer}{edgelayer}
		\filldraw[fill=blue!20,draw=blue!40] (12.center) to (15.center) to (14.center) to (13.center) to cycle;
	\filldraw[fill=white,draw=black] (7.center) to (6.center) to (5.center) to (4.center) to (3.center) to (9.center) to (10.center) to (8.center) to cycle;	
		\draw [style=qWire] (0) to (1.center);
		\draw [style=qWire] (0) to (2.center);
	\end{pgfonlayer}
\end{tikzpicture}
\ +(1-\omega)\
\begin{tikzpicture}
	\begin{pgfonlayer}{nodelayer}
		\node [style=small box] (0) at (2.25, 0) {$T_3$};
		\node [style=none] (1) at (2.25, 0.75) {};
		\node [style=none] (2) at (2.25, -0.75) {};
		\node [style=none] (3) at (1.75, 0.75) {};
		\node [style=none] (4) at (1.75, 1.25) {};
		\node [style=none] (5) at (3.5, 1.25) {};
		\node [style=none] (6) at (3.5, -1.25) {};
		\node [style=none] (7) at (1.75, -1.25) {};
		\node [style=none] (8) at (1.75, -0.75) {};
		\node [style=none] (9) at (3, 0.75) {};
		\node [style=none] (10) at (3, -0.75) {};
		\node [style=none] (11) at (3.25, 0) {$\tau$};
		\node [style=none] (12) at (1, 2) {};
		\node [style=none] (13) at (4.25, 2) {};
		\node [style=none] (14) at (4.25, -2) {};
		\node [style=none] (15) at (1, -2) {};
		\node [style=none] (16) at (4, -1.75) {\tiny $p$};
	\end{pgfonlayer}
	\begin{pgfonlayer}{edgelayer}
		\filldraw[fill=blue!20,draw=blue!40] (12.center) to (15.center) to (14.center) to (13.center) to cycle;
	\filldraw[fill=white,draw=black] (7.center) to (6.center) to (5.center) to (4.center) to (3.center) to (9.center) to (10.center) to (8.center) to cycle;	
		\draw [style=qWire] (0) to (1.center);
		\draw [style=qWire] (0) to (2.center);
	\end{pgfonlayer}
\end{tikzpicture}\label{mixprobconstr}
\eeq
Additionally,
if one operational effect $E_1$ is the coarse-graining of
two others, $E_2$ and $E_3$, then $\Pr(E_1,P)$ must
be the sum of $\Pr(E_2,P)$ and $\Pr(E_3,P)$ for all $P$.

Our main result holds only for operational theories satisfying the following property:
\begin{align} \label{tomlocforopthry}
\forall E_1,E_2,P_1,P_2\quad \input{Diagrams/7b_TomLoc1b.tikz} = \input{Diagrams/8b_TomLoc2b.tikz} \nonumber\\  \rotatebox{90}{$\iff$}\hspace{.22\textwidth} \nonumber \\
\forall E,P\quad \input{Diagrams/9b_TomLoc3b.tikz} = \input{Diagrams/10b_TomLoc4b.tikz}.
\end{align}
In other words, any two processes $T$ and $T'$ that give the same statistics for all local preparations on their inputs and all local measurements on their outputs {\em also} give the same statistics in arbitrary circuits.
 Such operational theories are alternatively characterized by the fact that the GPT defined by quotienting them satisfies tomographic locality, as we show below.

Two processes with the same input systems and output systems are said to be {\em operationally equivalent}~\cite{Spekkens2005} if they give rise to the same probabilities no matter what closed diagram they are embedded in.  The testers from Eq.~\eqref{eq:tester} facilitate a convenient diagrammatic representation of this condition.
That is, two processes are operationally equivalent, denoted by
\beq
\begin{tikzpicture}
	\begin{pgfonlayer}{nodelayer}
		\node [style={small box}] (0) at (-4, -0) {$T$};
		\node [style=none] (1) at (-4, -1.25) {};
		\node [style=none] (2) at (-4, 1.25) {};
		\node [style={small box}] (3) at (-1, -0) {$T'$};
		\node [style=none] (4) at (-1, -1.25) {};
		\node [style=none] (5) at (-1, 1.25) {};
		\node [style=none] (6) at (-2.5, 0) {$\simeq$};
		\node [style={right label}] (7) at (-4, -0.9999999) {$A$};
		\node [style={right label}] (8) at (-1, -0.9999999) {$A$};
		\node [style={right label}] (9) at (-4, 0.9999999) {$B$};
		\node [style={right label}] (10) at (-1, 0.9999999) {$B$};
	\end{pgfonlayer}
	\begin{pgfonlayer}{edgelayer}
		\draw [style=qWire] (0) to (2.center);
		\draw [style=qWire] (0) to (1.center);
		\draw [style=qWire] (3) to (5.center);
		\draw [style=qWire] (3) to (4.center);
	\end{pgfonlayer}
\end{tikzpicture}
\eeq
if they assign equal probabilities to every tester\footnote{The recognition that (for operational theories corresponding to GPTs that are not tomographically local) one must consider testers that allow for side channels can be found in Refs.~\cite{hardy2011reformulating,chiribella2014dilation,chiribella2016quantum} .}, so that
\beq
\forall \tau  \ \
\begin{tikzpicture}
	\begin{pgfonlayer}{nodelayer}
		\node [style=small box] (0) at (2.25, -0) {$T$};
		\node [style=none] (1) at (2.25, 0.75) {};
		\node [style=none] (2) at (2.25, -0.75) {};
		\node [style=none] (3) at (1.75, 0.75) {};
		\node [style=none] (4) at (1.75, 1.25) {};
		\node [style=none] (5) at (3.5, 1.25) {};
		\node [style=none] (6) at (3.5, -1.25) {};
		\node [style=none] (7) at (1.75, -1.25) {};
		\node [style=none] (8) at (1.75, -0.75) {};
		\node [style=none] (9) at (3, 0.75) {};
		\node [style=none] (10) at (3, -0.75) {};
		\node [style=none] (11) at (3.25, -0) {$\tau$};
		\node [style=none] (12) at (1, 2) {};
		\node [style=none] (13) at (4.25, 2) {};
		\node [style=none] (14) at (4.25, -2) {};
		\node [style=none] (15) at (1, -2) {};
		\node [style=none] (16) at (4, -1.75) {\tiny $p$};
		\node [style=none] (17) at (6.5, 0.75) {};
		\node [style=none] (18) at (9, 2) {};
		\node [style=none] (19) at (7, 0.75) {};
		\node [style=none] (20) at (8.25, 1.25) {};
		\node [style=none] (21) at (7.75, 0.75) {};
		\node [style=small box] (22) at (7, -0) {$T'$};
		\node [style=none] (23) at (6.5, -0.75) {};
		\node [style=none] (24) at (9, -2) {};
		\node [style=none] (25) at (7, -0.75) {};
		\node [style=none] (26) at (8.25, -1.25) {};
		\node [style=none] (27) at (6.5, 1.25) {};
		\node [style=none] (28) at (5.75, -2) {};
		\node [style=none] (29) at (6.5, -1.25) {};
		\node [style=none] (30) at (7.75, -0.75) {};
		\node [style=none] (31) at (8, -0) {$\tau$};
		\node [style=none] (32) at (5.75, 2) {};
		\node [style=none] (33) at (8.75, -1.75) {\tiny $p$};
		\node [style=none] (34) at (5, -0) {$=$};
	\end{pgfonlayer}
	\begin{pgfonlayer}{edgelayer}
	\filldraw[fill=blue!20,draw=blue!40]  (32.center) to (28.center) to (24.center) to (18.center) to cycle;
	\filldraw[fill=blue!20,draw=blue!40] (12.center) to (15.center) to (14.center) to (13.center) to cycle;
		\draw [style=qWire] (0) to (1.center);
		\draw [style=qWire] (0) to (2.center);
	\filldraw[fill=white,draw=black] (7.center) to (6.center) to (5.center) to (4.center) to (3.center) to (9.center) to (10.center) to (8.center) to cycle;	
		\draw [style=qWire] (22) to (19.center);
		\draw [style=qWire] (22) to (25.center);
		\filldraw[fill=white,draw=black] (29.center) to (26.center) to (20.center) to (27.center) to (17.center) to (21.center) to (30.center) to (23.center) to cycle;
	\end{pgfonlayer}
\end{tikzpicture}\ .
\eeq

It is easy to see that operational equivalence defines an equivalence relation. Hence, we can divide the space of processes into equivalence classes, and each process $T$ in the operational theory can be specified by its equivalence class $\widetilde{T}$, together with a label $c_T$ of the context of $T$, specifying which element of the equivalence class it is. For a given $T$, $c_T$ provides all the information which defines that process which is not relevant to its equivalence class. Hence, each procedure is specified by a tuple, $T :=(\widetilde{T},c_T)$, and we will denote it as such when convenient. In the case of closed diagrams, the equivalence class can be uniquely specified by the probability given by the map $p$, and so any information beyond this forms the context of the closed diagram.

Next, we define a quotienting map $\sim$ which maps procedures into their equivalence class (exactly as is done to construct quotiented operational probabilistic theories in Ref.~\cite{chiribella2016quantum}). Given a characterization of each procedure as a tuple of equivalence class and context, the quotienting map picks out the first element of this tuple, taking $(\widetilde{T},c_T) \to \widetilde{T}$. \footnote{ This notation is very convenient in practice, but it is worth noting that it has the awkward feature that two operationally equivalent procedures $S$ and $T$ would map to $\widetilde{S}$ and $\widetilde{T}$ (respectively), which constitute two distinct labels for the same equivalence class (since $\widetilde{S}=\widetilde{T}$).  } Diagrammatically, we have
\[\input{Diagrams/15b_Quot1b.tikz}\ =\ \input{Diagrams/16b_Quot2b.tikz}.\]
We prove that it is diagram-preserving in Appendix~\ref{app:quotDP}. 
For processes which are closed diagrams, one can always choose the representative of the equivalence class to be the real number specified by the probability rule. 
Hence, the map \colorbox{black!50!blue!30!}{$\sim:\Op\to\widetilde{\Op}$} can be viewed as a diagram-preserving extension of the probability rule $p$. 
This implies that the quotiented operational theory reproduces the predictions of the operational theory, since
\beq\label{eq:GPTProbs}
\begin{tikzpicture}
	\begin{pgfonlayer}{nodelayer}
		\node [style=point] (0) at (0, -.75) {$\widetilde{P}$};
		\node [style=copoint] (1) at (0, .75) {$\widetilde{E}$};
	\end{pgfonlayer}
	\begin{pgfonlayer}{edgelayer}
		\draw[qWire] (1) to (0);
	\end{pgfonlayer}
\end{tikzpicture}
\ =\
\begin{tikzpicture}
	\begin{pgfonlayer}{nodelayer}
		\node [style=point] (0) at (0, -1) {${P}$};
		\node [style=copoint] (1) at (0, 1) {${E}$};
		\node [style=none] (2) at (-1, 2) {};
		\node [style=none] (3) at (-1, 0) {};
		\node [style=none] (4) at (1, 0) {};
		\node [style=none] (5) at (1, 2) {};
		\node [style=none] (6) at (-1, -0.25) {};
		\node [style=none] (7) at (-1, -2) {};
		\node [style=none] (8) at (1, -2) {};
		\node [style=none] (9) at (1, -0.25) {};
		\node [style=none] (10) at (0.75, -1.75) {\tiny $\sim$};
		\node [style=none] (11) at (0.75, 0.25) {\tiny $\sim$};
	\end{pgfonlayer}
	\begin{pgfonlayer}{edgelayer}
		\filldraw [ fill=black!50!blue!30!,draw=black!40!blue!40] (6.center) to (7.center) to (8.center) to (9.center) to cycle;
				\filldraw [ fill=black!50!blue!30!,draw=black!40!blue!40] (5.center) to (4.center) to (3.center) to (2.center) to cycle;
		\draw[qWire] (1) to (0);
	\end{pgfonlayer}
\end{tikzpicture}
\ =\
\begin{tikzpicture}
	\begin{pgfonlayer}{nodelayer}
		\node [style=point] (0) at (0, -.75) {${P}$};
		\node [style=copoint] (1) at (0, .75) {${E}$};
		\node [style=none] (2) at (-1, 1.75) {};
		\node [style=none] (3) at (-1, -1.75) {};
		\node [style=none] (4) at (1, -1.75) {};
		\node [style=none] (5) at (1, 1.75) {};
		\node [style=none] (11) at (0.75, -1.5) {\tiny $\sim$};
	\end{pgfonlayer}
	\begin{pgfonlayer}{edgelayer}
					\filldraw [ fill=black!50!blue!30!,draw=black!40!blue!40] (5.center) to (4.center) to (3.center) to (2.center) to cycle;
		\draw[qWire] (1) to (0);
	\end{pgfonlayer}
\end{tikzpicture}
\ =\
\begin{tikzpicture}
	\begin{pgfonlayer}{nodelayer}
		\node [style=point] (0) at (0, -.75) {${P}$};
		\node [style=copoint] (1) at (0, .75) {${E}$};
		\node [style=none] (2) at (-1, 1.75) {};
		\node [style=none] (3) at (-1, -1.75) {};
		\node [style=none] (4) at (1, -1.75) {};
		\node [style=none] (5) at (1, 1.75) {};
		\node [style=none] (11) at (0.75, -1.5) {\tiny $p$};
	\end{pgfonlayer}
	\begin{pgfonlayer}{edgelayer}
					\filldraw [ fill=blue!20,draw=blue!40] (5.center) to (4.center) to (3.center) to (2.center) to cycle;
		\draw[qWire] (1) to (0);
	\end{pgfonlayer}
\end{tikzpicture}
\ =\
\Pr(E,P).
\eeq
 It is worth noting that in the quotiented operational theory, a closed diagram is equal to a real number (the probability associated to it), while in the operational theory these are not equal until the map $p$ is applied to the closed diagram.

We will assume that every deterministic effect for a given system, $A$,  in the operational theory (corresponding to implementing a measurement on the system and marginalizing over its outcome) is operationally equivalent. We denote these deterministic effects as:
\beq
\begin{tikzpicture}
	\begin{pgfonlayer}{nodelayer}
		\node [style=none] (0) at (0.5, -0.75) {};
		\node [style=none] (1) at (0.5, 0.25) {};
		\node [style=upground] (2) at (0.5, 0.5) {};
		\node [style=none] (3) at (1, 0.75) {$c$};
		\node [style=right label] (4) at (0.5, -0.5) {$A$};
	\end{pgfonlayer}
	\begin{pgfonlayer}{edgelayer}
		\draw [qWire] (1.center) to (0.center);
	\end{pgfonlayer}
\end{tikzpicture}
\eeq
where $c$ labels the context. 

\subsubsection{The GPT associated to an operational theory} \label{Edefn}

It is well known \cite{chiribella2010probabilistic,
chiribella2011informational,
hardy2011reformulating} that a quotiented operational theory, $\widetilde{\Op}$, is nothing but a generalized probabilistic theory~\cite{hardy2001quantum,barrett2007}, and in fact for this paper we view this as the definition of a GPT.  We will now demonstrate this by showing that $\widetilde{\Op}$ is tomographic (a notion that will be defined momentarily),
 is representable in  real vector spaces,  is convex, and has a unique deterministic effect. This is analogous to how quotiented OPTs arise from unquotiented OPTs in \cite{chiribella2010probabilistic,
chiribella2011informational}.

Firstly, note that the quotiented operational theory is \emph{tomographic}. For a generic process theory, $\mathcal{P}$, being tomographic means that processes are characterized by scalars. That is, given any two distinct processes $f,g:A \to B \in \mathcal{P}$,
\beq
\begin{tikzpicture}
	\begin{pgfonlayer}{nodelayer}
		\node [style={small box}] (0) at (-3.5, 0) {$f$};
		\node [style=none] (1) at (-3.5, -1.25) {};
		\node [style=none] (2) at (-3.5, 1.25) {};
	\end{pgfonlayer}
	\begin{pgfonlayer}{edgelayer}
		\draw [style=qWire] (0) to (2.center);
		\draw [style=qWire] (0) to (1.center);
	\end{pgfonlayer}
\end{tikzpicture}
\ \neq\
\begin{tikzpicture}
	\begin{pgfonlayer}{nodelayer}
		\node [style=none] (0) at (-1.5, 1.25) {};
		\node [style=none] (1) at (-1.5, -1.25) {};
		\node [style={small box}] (2) at (-1.5, 0) {$g$};
	\end{pgfonlayer}
	\begin{pgfonlayer}{edgelayer}
		\draw [style=qWire] (2) to (0.center);
		\draw [style=qWire] (2) to (1.center);
	\end{pgfonlayer}
\end{tikzpicture}, \label{TLLHS}
\eeq
there must exist a tester $h \in\mathcal{P}$ that turns each of these processes into a closed diagram, i.e., a scalar, such that the scalars are distinct:
\beq \exists h \in\mathcal{P} : \
\begin{tikzpicture}
	\begin{pgfonlayer}{nodelayer}
		\node [style=none] (0) at (1.75, -1.25) {};
		\node [style=none] (1) at (1.75, -0.75) {};
		\node [style=none] (2) at (1.75, 1.25) {};
		\node [style=none] (3) at (2.25, 0.75) {};
		\node [style={small box}] (4) at (2.25, 0) {$f$};
		\node [style=none] (5) at (3, 0.75) {};
		\node [style=none] (6) at (2.25, -0.75) {};
		\node [style=none] (7) at (3.5, -1.25) {};
		\node [style=none] (8) at (3.5, 1.25) {};
		\node [style=none] (9) at (3, -0.75) {};
		\node [style=none] (10) at (1.75, 0.75) {};
		\node [style=none] (11) at (3.25, 0) {$h$};
	\end{pgfonlayer}
	\begin{pgfonlayer}{edgelayer}
		\draw [style=qWire] (4) to (3.center);
		\draw [style=qWire] (4) to (6.center);
		\draw (2.center) to (8.center);
		\draw (8.center) to (7.center);
		\draw (7.center) to (0.center);
		\draw (0.center) to (1.center);
		\draw (1.center) to (9.center);
		\draw (9.center) to (5.center);
		\draw (5.center) to (10.center);
		\draw (10.center) to (2.center);
	\end{pgfonlayer}
\end{tikzpicture}
\ \neq\
\begin{tikzpicture}
	\begin{pgfonlayer}{nodelayer}
		\node [style=none] (0) at (1.75, -1.25) {};
		\node [style=none] (1) at (1.75, -0.75) {};
		\node [style=none] (2) at (1.75, 1.25) {};
		\node [style=none] (3) at (2.25, 0.75) {};
		\node [style={small box}] (4) at (2.25, 0) {$g$};
		\node [style=none] (5) at (3, 0.75) {};
		\node [style=none] (6) at (2.25, -0.75) {};
		\node [style=none] (7) at (3.5, -1.25) {};
		\node [style=none] (8) at (3.5, 1.25) {};
		\node [style=none] (9) at (3, -0.75) {};
		\node [style=none] (10) at (1.75, 0.75) {};
		\node [style=none] (11) at (3.25, 0) {$h$};
	\end{pgfonlayer}
	\begin{pgfonlayer}{edgelayer}
		\draw [style=qWire] (4) to (3.center);
		\draw [style=qWire] (4) to (6.center);
		\draw (2.center) to (8.center);
		\draw (8.center) to (7.center);
		\draw (7.center) to (0.center);
		\draw (0.center) to (1.center);
		\draw (1.center) to (9.center);
		\draw (9.center) to (5.center);
		\draw (5.center) to (10.center);
		\draw (10.center) to (2.center);
	\end{pgfonlayer}
\end{tikzpicture}. \label{TLRHS}
\eeq

That Eq.~\eqref{TLRHS} implies Eq.~\eqref{TLLHS}  for processes in a quotiented operational theory is trivial; we now give the proof that Eq.~\eqref{TLLHS} implies Eq.~\eqref{TLRHS}. Consider two distinct processes, $\widetilde{T}$ and $\widetilde{T}'$,
in the quotiented operational theory---the images under the quotienting map of process $T$ and $T'$ in the operational theory---such that:
\beq
\begin{tikzpicture}
	\begin{pgfonlayer}{nodelayer}
		\node [style={small box}] (0) at (-4, -0) {$\widetilde{T}$};
		\node [style=none] (1) at (-4, -1.25) {};
		\node [style=none] (2) at (-4, 1.25) {};
		\node [style={small box}] (3) at (-1, -0) {$\widetilde{T}'$};
		\node [style=none] (4) at (-1, -1.25) {};
		\node [style=none] (5) at (-1, 1.25) {};
		\node [style=none] (6) at (-2.5, 0) {$\neq$};
		\node [style={right label}] (7) at (-4, -1.25) {$A$};
		\node [style={right label}] (8) at (-1, -1.25) {$A$};
		\node [style={right label}] (9) at (-4, 0.9999999) {$B$};
		\node [style={right label}] (10) at (-1, 0.9999999) {$B$};
	\end{pgfonlayer}
	\begin{pgfonlayer}{edgelayer}
		\draw [style=qWire] (0) to (2.center);
		\draw [style=qWire] (0) to (1.center);
		\draw [style=qWire] (3) to (5.center);
		\draw [style=qWire] (3) to (4.center);
	\end{pgfonlayer}
\end{tikzpicture}.
\eeq
By definition, we know that $\widetilde{T} \neq \widetilde{T}'$ implies that $T\not\simeq T'$, and hence there exists some tester $\tau$ such that
\beq
\begin{tikzpicture}
	\begin{pgfonlayer}{nodelayer}
		\node [style=small box] (0) at (2.25, 0) {$T$};
		\node [style=none] (1) at (2.25, 0.75) {};
		\node [style=none] (2) at (2.25, -0.75) {};
		\node [style=none] (3) at (1.75, 0.75) {};
		\node [style=none] (4) at (1.75, 1.25) {};
		\node [style=none] (5) at (3.5, 1.25) {};
		\node [style=none] (6) at (3.5, -1.25) {};
		\node [style=none] (7) at (1.75, -1.25) {};
		\node [style=none] (8) at (1.75, -0.75) {};
		\node [style=none] (9) at (3, 0.75) {};
		\node [style=none] (10) at (3, -0.75) {};
		\node [style=none] (11) at (3.25, 0) {$\tau$};
		\node [style=none] (12) at (1, 2) {};
		\node [style=none] (13) at (4.25, 2) {};
		\node [style=none] (14) at (4.25, -2) {};
		\node [style=none] (15) at (1, -2) {};
		\node [style=none] (16) at (4, -1.75) {\tiny $p$};
		\node [style=none] (17) at (6.5, 0.75) {};
		\node [style=none] (18) at (9, 2) {};
		\node [style=none] (19) at (7, 0.75) {};
		\node [style=none] (20) at (8.25, 1.25) {};
		\node [style=none] (21) at (7.75, 0.75) {};
		\node [style=small box] (22) at (7, 0) {$T'$};
		\node [style=none] (23) at (6.5, -0.75) {};
		\node [style=none] (24) at (9, -2) {};
		\node [style=none] (25) at (7, -0.75) {};
		\node [style=none] (26) at (8.25, -1.25) {};
		\node [style=none] (27) at (6.5, 1.25) {};
		\node [style=none] (28) at (5.75, -2) {};
		\node [style=none] (29) at (6.5, -1.25) {};
		\node [style=none] (30) at (7.75, -0.75) {};
		\node [style=none] (31) at (8, 0) {$\tau$};
		\node [style=none] (32) at (5.75, 2) {};
		\node [style=none] (33) at (8.75, -1.75) {\tiny $p$};
		\node [style=none] (34) at (5, 0) {$\neq$};
	\end{pgfonlayer}
	\begin{pgfonlayer}{edgelayer}
		\filldraw[fill=blue!20,draw=blue!40]  (32.center) to (28.center) to (24.center) to (18.center) to cycle;
	\filldraw[fill=blue!20,draw=blue!40] (12.center) to (15.center) to (14.center) to (13.center) to cycle;
	\filldraw[fill=white,draw=black] (7.center) to (6.center) to (5.center) to (4.center) to (3.center) to (9.center) to (10.center) to (8.center) to cycle;	
\filldraw[fill=white,draw=black] (29.center) to (26.center) to (20.center) to (27.center) to (17.center) to (21.center) to (30.center) to (23.center) to cycle;
		\draw [style=qWire] (0) to (1.center);
		\draw [style=qWire] (0) to (2.center);
		\draw [style=qWire] (22) to (19.center);
		\draw [style=qWire] (22) to (25.center);
	\end{pgfonlayer}
\end{tikzpicture}.
\eeq
Since the action of $p$ is identical to that of $\sim$ on closed diagrams, this implies that
\beq
\begin{tikzpicture}
	\begin{pgfonlayer}{nodelayer}
		\node [style=small box] (0) at (2.25, 0) {$T$};
		\node [style=none] (1) at (2.25, 0.75) {};
		\node [style=none] (2) at (2.25, -0.75) {};
		\node [style=none] (3) at (1.75, 0.75) {};
		\node [style=none] (4) at (1.75, 1.25) {};
		\node [style=none] (5) at (3.5, 1.25) {};
		\node [style=none] (6) at (3.5, -1.25) {};
		\node [style=none] (7) at (1.75, -1.25) {};
		\node [style=none] (8) at (1.75, -0.75) {};
		\node [style=none] (9) at (3, 0.75) {};
		\node [style=none] (10) at (3, -0.75) {};
		\node [style=none] (11) at (3.25, 0) {$\tau$};
		\node [style=none] (12) at (1, 2) {};
		\node [style=none] (13) at (4.25, 2) {};
		\node [style=none] (14) at (4.25, -2) {};
		\node [style=none] (15) at (1, -2) {};
		\node [style=none] (16) at (4, -1.75) {\tiny $\sim$};
		\node [style=none] (17) at (6.5, 0.75) {};
		\node [style=none] (18) at (9, 2) {};
		\node [style=none] (19) at (7, 0.75) {};
		\node [style=none] (20) at (8.25, 1.25) {};
		\node [style=none] (21) at (7.75, 0.75) {};
		\node [style=small box] (22) at (7, 0) {$T'$};
		\node [style=none] (23) at (6.5, -0.75) {};
		\node [style=none] (24) at (9, -2) {};
		\node [style=none] (25) at (7, -0.75) {};
		\node [style=none] (26) at (8.25, -1.25) {};
		\node [style=none] (27) at (6.5, 1.25) {};
		\node [style=none] (28) at (5.75, -2) {};
		\node [style=none] (29) at (6.5, -1.25) {};
		\node [style=none] (30) at (7.75, -0.75) {};
		\node [style=none] (31) at (8, 0) {$\tau$};
		\node [style=none] (32) at (5.75, 2) {};
		\node [style=none] (33) at (8.75, -1.75) {\tiny $\sim$};
		\node [style=none] (34) at (5, 0) {$\neq$};
	\end{pgfonlayer}
	\begin{pgfonlayer}{edgelayer}
		\filldraw[ fill=black!50!blue!30!,draw=black!40!blue!40]  (32.center) to (28.center) to (24.center) to (18.center) to cycle;
	\filldraw[ fill=black!50!blue!30!,draw=black!40!blue!40] (12.center) to (15.center) to (14.center) to (13.center) to cycle;
	\filldraw[fill=white,draw=black] (7.center) to (6.center) to (5.center) to (4.center) to (3.center) to (9.center) to (10.center) to (8.center) to cycle;	
\filldraw[fill=white,draw=black] (29.center) to (26.center) to (20.center) to (27.center) to (17.center) to (21.center) to (30.center) to (23.center) to cycle;
		\draw [style=qWire] (0) to (1.center);
		\draw [style=qWire] (0) to (2.center);
		\draw [style=qWire] (22) to (19.center);
		\draw [style=qWire] (22) to (25.center);
	\end{pgfonlayer}
\end{tikzpicture}.
\eeq
Finally, we can use the fact that the quotienting map is diagram-preserving to write that there exists $\widetilde{\tau}$ such that
\beq \label{eq:tomog}
\begin{tikzpicture}
	\begin{pgfonlayer}{nodelayer}
		\node [style=small box] (0) at (2.25, 0) {$\widetilde{T}$};
		\node [style=none] (1) at (2.25, 0.75) {};
		\node [style=none] (2) at (2.25, -0.75) {};
		\node [style=none] (3) at (1.75, 0.75) {};
		\node [style=none] (4) at (1.75, 1.25) {};
		\node [style=none] (5) at (3.5, 1.25) {};
		\node [style=none] (6) at (3.5, -1.25) {};
		\node [style=none] (7) at (1.75, -1.25) {};
		\node [style=none] (8) at (1.75, -0.75) {};
		\node [style=none] (9) at (3, 0.75) {};
		\node [style=none] (10) at (3, -0.75) {};
		\node [style=none] (11) at (3.25, 0) {$\widetilde{\tau}$};
		\node [style=none] (17) at (6.5, 0.75) {};
		\node [style=none] (19) at (7, 0.75) {};
		\node [style=none] (20) at (8.25, 1.25) {};
		\node [style=none] (21) at (7.75, 0.75) {};
		\node [style=small box] (22) at (7, 0) {$\widetilde{T}'$};
		\node [style=none] (23) at (6.5, -0.75) {};
		\node [style=none] (25) at (7, -0.75) {};
		\node [style=none] (26) at (8.25, -1.25) {};
		\node [style=none] (27) at (6.5, 1.25) {};
		\node [style=none] (29) at (6.5, -1.25) {};
		\node [style=none] (30) at (7.75, -0.75) {};
		\node [style=none] (31) at (8, 0) {$\widetilde{\tau}$};
		\node [style=none] (34) at (5, 0) {$\neq$};
	\end{pgfonlayer}
	\begin{pgfonlayer}{edgelayer}
	\filldraw[fill=white,draw=black] (7.center) to (6.center) to (5.center) to (4.center) to (3.center) to (9.center) to (10.center) to (8.center) to cycle;	
\filldraw[fill=white,draw=black] (29.center) to (26.center) to (20.center) to (27.center) to (17.center) to (21.center) to (30.center) to (23.center) to cycle;
		\draw [style=qWire] (0) to (1.center);
		\draw [style=qWire] (0) to (2.center);
		\draw [style=qWire] (22) to (19.center);
		\draw [style=qWire] (22) to (25.center);
	\end{pgfonlayer}
\end{tikzpicture}.
\eeq
This establishes that Eq.~\eqref{TLLHS} implies Eq.~\eqref{TLRHS},
and so the quotiented operational theory is tomographic.

This means that we can identify an operational equivalence class of processes, $\widetilde{T}$, with a real vector, $\mathbf{K}_{\widetilde{T}}$, living in $\mathds{R}^{\mathcal{T}^{A\to B}}$, where $\mathcal{T}^{A\to B}$ denotes the set of testers for processes with input $A$ and output $B$. Concretely, we define these vectors component-wise via
\beq
\left[\mathbf{K}_{\widetilde{T}}\right]_\tau : = \begin{tikzpicture}
	\begin{pgfonlayer}{nodelayer}
		\node [style=small box] (0) at (0, 0) {$\widetilde{T}$};
		\node [style=none] (1) at (0, 0.75) {};
		\node [style=none] (2) at (0, -0.75) {};
		\node [style=none] (3) at (-0.5, 0.75) {};
		\node [style=none] (4) at (-0.5, 1.25) {};
		\node [style=none] (5) at (1.25, 1.25) {};
		\node [style=none] (6) at (1.25, -1.25) {};
		\node [style=none] (7) at (-0.5, -1.25) {};
		\node [style=none] (8) at (-0.5, -0.75) {};
		\node [style=none] (9) at (0.75, 0.75) {};
		\node [style=none] (10) at (0.75, -0.75) {};
		\node [style=none] (11) at (1, 0) {$\widetilde{\tau}$};
	\end{pgfonlayer}
	\begin{pgfonlayer}{edgelayer}
		\draw [fill=white, draw=black] (7.center)
			 to (6.center)
			 to (5.center)
			 to (4.center)
			 to (3.center)
			 to (9.center)
			 to (10.center)
			 to (8.center)
			 to cycle;
		\draw [style=qWire] (0) to (1.center);
		\draw [style=qWire] (0) to (2.center);
	\end{pgfonlayer}
\end{tikzpicture}.
\eeq
Clearly, following on from the discussion around Eq.~\eqref{eq:tomog}, we have that $\mathbf{K}_{\widetilde{T}} = \mathbf{K}_{\widetilde{T'}}$ if and only if $\widetilde{T} = \widetilde{T'}$.
 This vector space representation, however, is generally infinite dimensional, and gives a highly inefficient characterisation of processes. We can instead focus on some minimal subset of \emph{fiducial} testers $\mathcal{F}^{A\to B} \subset \mathcal{T}^{A\to B}$ which, for notational convenience, we index as $\mathcal{F}^{A\to B} := \{\widetilde{\tau}_\alpha\}_{\alpha = 1,2,..., m^{{A}\to {B}}}$.

 The term \emph{fiducial} means that this subset of testers satisfies two key properties. The first is that they must also suffice  for tomography, i.e.,  
 \beq
\begin{tikzpicture}
	\begin{pgfonlayer}{nodelayer}
		\node [style=small box] (0) at (5.5, 0) {$\widetilde{T}$};
		\node [style=none] (1) at (5.5, 0.75) {};
		\node [style=none] (2) at (5.5, -0.75) {};
		\node [style=none] (3) at (5, 0.75) {};
		\node [style=none] (4) at (5, 1.25) {};
		\node [style=none] (5) at (7.25, 1.25) {};
		\node [style=none] (6) at (7.25, -1.25) {};
		\node [style=none] (7) at (5, -1.25) {};
		\node [style=none] (8) at (5, -0.75) {};
		\node [style=none] (9) at (6.25, 0.75) {};
		\node [style=none] (10) at (6.25, -0.75) {};
		\node [style=none] (11) at (6.75, 0) {$\widetilde{\tau}_\alpha$};
		\node [style=none] (17) at (8.75, 0.75) {};
		\node [style=none] (19) at (9.25, 0.75) {};
		\node [style=none] (20) at (11, 1.25) {};
		\node [style=none] (21) at (10, 0.75) {};
		\node [style=small box] (22) at (9.25, 0) {$\widetilde{T}'$};
		\node [style=none] (23) at (8.75, -0.75) {};
		\node [style=none] (25) at (9.25, -0.75) {};
		\node [style=none] (26) at (11, -1.25) {};
		\node [style=none] (27) at (8.75, 1.25) {};
		\node [style=none] (29) at (8.75, -1.25) {};
		\node [style=none] (30) at (10, -0.75) {};
		\node [style=none] (31) at (10.5, 0) {$\widetilde{\tau}_\alpha$};
		\node [style=none] (34) at (8, 0) {$\neq$};
		\node [style=label] (35) at (2, 0) {$\exists \alpha \in \{1,...,m^{A\to B}\}$  : \ };
		\node [style=none] (36) at (-1.75, 0) {$\iff$};
		\node [style=small box] (37) at (-5.25, 0) {$\widetilde{T}$};
		\node [style=none] (38) at (-5.25, -1.25) {};
		\node [style=none] (39) at (-5.25, 1.25) {};
		\node [style=small box] (40) at (-3.25, 0) {$\widetilde{T}'$};
		\node [style=none] (41) at (-3.25, -1.25) {};
		\node [style=none] (42) at (-3.25, 1.25) {};
		\node [style=none] (43) at (-4.25, 0) {$\neq$};
		\node [style=right label] (44) at (-5.25, -1.25) {$A$};
		\node [style=right label] (45) at (-3.25, -1.25) {$A$};
		\node [style=right label] (46) at (-5.25, 1.25) {$B$};
		\node [style=right label] (47) at (-3.25, 1.25) {$B$};
	\end{pgfonlayer}
	\begin{pgfonlayer}{edgelayer}
		\draw [style=qWire] (0) to (1.center);
		\draw [style=qWire] (0) to (2.center);
		\draw [style=qWire] (22) to (19.center);
		\draw [style=qWire] (22) to (25.center);
		\draw [style=qWire] (37) to (39.center);
		\draw [style=qWire] (37) to (38.center);
		\draw [style=qWire] (40) to (42.center);
		\draw [style=qWire] (40) to (41.center);
		\draw (4.center) to (3.center);
		\draw (3.center) to (9.center);
		\draw (9.center) to (10.center);
		\draw (10.center) to (8.center);
		\draw (8.center) to (7.center);
		\draw (7.center) to (6.center);
		\draw (6.center) to (5.center);
		\draw (5.center) to (4.center);
		\draw (27.center) to (17.center);
		\draw (17.center) to (21.center);
		\draw (21.center) to (30.center);
		\draw (30.center) to (23.center);
		\draw (23.center) to (29.center);
		\draw (29.center) to (26.center);
		\draw (26.center) to (20.center);
		\draw (20.center) to (27.center);
	\end{pgfonlayer}
\end{tikzpicture}
.
\eeq 
  We can therefore use these fiducial testers to define a finite dimensional vector representation $\mathbf{R}_{\widetilde{T}}$ of a process $\widetilde{T}$, defined componentwise via
 \beq
[\mathbf{R}_{\widetilde{T}}]_\alpha :=
\begin{tikzpicture}
	\begin{pgfonlayer}{nodelayer}
		\node [style=small box] (0) at (0, 0) {$\widetilde{T}$};
		\node [style=none] (1) at (0, 0.75) {};
		\node [style=none] (2) at (0, -0.75) {};
		\node [style=none] (3) at (-0.5, 0.75) {};
		\node [style=none] (4) at (-0.5, 1.25) {};
		\node [style=none] (5) at (1.75, 1.25) {};
		\node [style=none] (6) at (1.75, -1.25) {};
		\node [style=none] (7) at (-0.5, -1.25) {};
		\node [style=none] (8) at (-0.5, -0.75) {};
		\node [style=none] (9) at (0.75, 0.75) {};
		\node [style=none] (10) at (0.75, -0.75) {};
		\node [style=none] (11) at (1.25, 0) {$\widetilde{\tau}_\alpha$};
	\end{pgfonlayer}
	\begin{pgfonlayer}{edgelayer}
		\draw [style=qWire] (0) to (1.center);
		\draw [style=qWire] (0) to (2.center);
		\draw (4.center) to (3.center);
		\draw (3.center) to (9.center);
		\draw (9.center) to (10.center);
		\draw (10.center) to (8.center);
		\draw (8.center) to (7.center);
		\draw (7.center) to (6.center);
		\draw (6.center) to (5.center);
		\draw (5.center) to (4.center);
	\end{pgfonlayer}
\end{tikzpicture} \label{defmatrep}
\eeq
for $\alpha = 1,2,...,m^{{A}\to{B}}$.
This new representation, $\mathbf{R}_{\widetilde{T}}$, has a straightforward relation to the original vector representation, $\mathbf{K}_{\widetilde{T}}$; all one must do to go from the original to the new representation is to restrict the $\mathbf{K}_{\widetilde{T}}$ vectors to the relevant subset of their components. We can think of this as a linear \emph{restriction} map
$\big|_{\mathcal{F}^{A\to B}}:\mathds{R}^{\mathcal{T}^{A\to B}}\to \mathds{R}^{\mathcal{F}^{A\to B}}:: \mathbf{K} \mapsto \mathbf{K}|_{\mathcal{F}^{A\to B}}$. This allows us now to relate these two representations via the observation that $\mathbf{K}_{\widetilde{T}}|_{\mathcal{F}^{A\to B}} = \mathbf{R}_{\widetilde{T}}$ for all processes $\widetilde{T}:A\to B$.
 
The second key property of fiducial sets of testers is that they define a \emph{linear compression} of the $\mathbf{K}_{\widetilde{T}}$ vectors.  
  Formally, what we mean by this is that there is a linear map $E^{A\to B} : \mathds{R}^{\mathcal{F}^{A\to B}} \to \mathds{R}^{\mathcal{T}^{A\to B}}$ which is the inverse to the restriction map $\big|_{\mathcal{F}^{A\to B}}$ on vectors $\mathbf{K}_{\widetilde{T}}$, that is, for all $\widetilde{T}$ we have that $E^{A\to B}(\mathbf{R}_{\widetilde{T}})=E^{A\to B}(\mathbf{K}_{\widetilde{T}}|_{\mathcal{F}^{A\to B}})=\mathbf{K}_{\widetilde{T}}$. 
  
  We reiterate that $\mathcal{F}^{A\to B}$ is taken to be a \emph{minimal} fiducial set of testers, which means that it is a minimal cardinality set of testers satisfying these two properties. Note that  minimal fiducial sets are typically not unique.

Consider now how the sequential composition of processes is represented. Given representations of a pair of processes $(\mathbf{R}_{\widetilde{T}},\mathbf{R}_{\widetilde{T}'})$ we know (as $\mathbf{R}$ is injective) that we can determine $\widetilde{T}$ and $\widetilde{T}'$, compute their composition $\widetilde{T}'\circ \widetilde{T}$, and via Eq.~\eqref{defmatrep} obtain $\mathbf{R}_{\widetilde{T}'\circ\widetilde{T}}$. We denote the sequential composition map on the vector representation as $\mathbf{R}_{\widetilde{T}'} \smallsquare \mathbf{R}_{\widetilde{T}} := \mathbf{R}_{\widetilde{T}'\circ \widetilde{T}}$. Similarly, for parallel composition we can define $\mathbf{R}_{\widetilde{T}'}\smallboxtimes\mathbf{R}_{\widetilde{T}} := \mathbf{R}_{\widetilde{T}'\otimes \widetilde{T}}$.
As we demonstrate in Appendix~\ref{sec:bilinearity}, both $\smallsquare$ and $\smallboxtimes$  can be uniquely extended to  bilinear maps on the relevant vector spaces. Specifically, we have:
\begin{lemma} \label{bilinearityprf}
The operation $\smallsquare$  can be uniquely extended to  a bilinear map
\beq
\smallsquare:\left(\mathds{R}^{m^{B\to C}},\mathds{R}^{m^{A\to B}}\right)\to \mathds{R}^{m^{A\to C}},
\eeq
and the operation $\smallboxtimes$  can be uniquely extended to  a bilinear map
\beq
\smallboxtimes : \left(\mathds{R}^{m^{A\to B}}, \mathds{R}^{m^{C\to D}}\right)\to \mathds{R}^{m^{{AC}\to{BD}}}.
\eeq
\end{lemma}
\noindent This implies that transformations act linearly on the state space, and also that the summation operation distributes over diagrams, i.e.:
\beq \label{sumdistributes}
\begin{tikzpicture}
	\begin{pgfonlayer}{nodelayer}
		\node [style=box] (0) at (-1, 0.75) {$\mathbf{R}_{\widetilde{T}}$};
		\node [style=none] (1) at (-1, 1.75) {};
		\node [style=box] (2) at (-1, -0.75) {$\sum_i r_i \mathbf{R}_{\widetilde{T}_i}$};
		\node [style=none] (3) at (-1, -1.75) {};
		\node [style=none] (4) at (1.75, -1.75) {};
		\node [style=none] (5) at (1.75, 1.75) {};
		\node [style=box] (6) at (1.75, -0.75) {$\mathbf{R}_{\widetilde{T}'}$};
	\end{pgfonlayer}
	\begin{pgfonlayer}{edgelayer}
		\draw [style=qWire] (2) to (0);
		\draw [style=qWire] (0) to (1.center);
		\draw [style=qWire] (3.center) to (2);
		\draw [style=qWire] (4.center) to (6);
		\draw [style=qWire] (6) to (5.center);
	\end{pgfonlayer}
\end{tikzpicture}
\ \ = \sum_i r_i \ \
\begin{tikzpicture}
	\begin{pgfonlayer}{nodelayer}
		\node [style=box] (0) at (-1, 0.75) {$\mathbf{R}_{\widetilde{T}}$};
		\node [style=none] (1) at (-1, 1.75) {};
		\node [style=box] (2) at (-1, -0.75) {$\mathbf{R}_{\widetilde{T}_i}$};
		\node [style=none] (3) at (-1, -1.75) {};
		\node [style=none] (4) at (1, -1.75) {};
		\node [style=none] (5) at (1, 1.75) {};
		\node [style=box] (6) at (1, -0.75) {$\mathbf{R}_{\widetilde{T}'}$};
	\end{pgfonlayer}
	\begin{pgfonlayer}{edgelayer}
		\draw [style=qWire] (2) to (0);
		\draw [style=qWire] (0) to (1.center);
		\draw [style=qWire] (3.center) to (2);
		\draw [style=qWire] (4.center) to (6);
		\draw [style=qWire] (6) to (5.center);
	\end{pgfonlayer}
\end{tikzpicture}.
\eeq

It is generally easier to work with this vector representation of processes rather than directly with the abstract process theory of operational equivalence classes of procedures. We will do so
when convenient, abusing notation by dropping the explicit symbol $\mathbf{R}$, and simply denoting the vector representation of the equivalence classes in the same way as the equivalence classes themselves. That is, we will denote $\mathbf{R}_{\widetilde{T}}$ by $\widetilde{T}$.
For example, we will write
 Eq.~\eqref{sumdistributes} as
\beq
\begin{tikzpicture}
	\begin{pgfonlayer}{nodelayer}
		\node [style=box] (0) at (-1, 0.75) {${\widetilde{T}}$};
		\node [style=none] (1) at (-1, 1.75) {};
		\node [style=box] (2) at (-1, -0.75) {$\sum_i r_i {\widetilde{T}_i}$};
		\node [style=none] (3) at (-1, -1.75) {};
		\node [style=none] (4) at (1.25, -1.75) {};
		\node [style=none] (5) at (1.25, 1.75) {};
		\node [style=box] (6) at (1.25, -0.75) {${\widetilde{T}'}$};
	\end{pgfonlayer}
	\begin{pgfonlayer}{edgelayer}
		\draw [style=qWire] (2) to (0);
		\draw [style=qWire] (0) to (1.center);
		\draw [style=qWire] (3.center) to (2);
		\draw [style=qWire] (4.center) to (6);
		\draw [style=qWire] (6) to (5.center);
	\end{pgfonlayer}
\end{tikzpicture}
\ \ = \sum_i r_i \ \
\begin{tikzpicture}
	\begin{pgfonlayer}{nodelayer}
		\node [style=box] (0) at (-1, 0.75) {${\widetilde{T}}$};
		\node [style=none] (1) at (-1, 1.75) {};
		\node [style=box] (2) at (-1, -0.75) {${\widetilde{T}_i}$};
		\node [style=none] (3) at (-1, -1.75) {};
		\node [style=none] (4) at (.5, -1.75) {};
		\node [style=none] (5) at (.5, 1.75) {};
		\node [style=box] (6) at (.5, -0.75) {${\widetilde{T}'}$};
	\end{pgfonlayer}
	\begin{pgfonlayer}{edgelayer}
		\draw [style=qWire] (2) to (0);
		\draw [style=qWire] (0) to (1.center);
		\draw [style=qWire] (3.center) to (2);
		\draw [style=qWire] (4.center) to (6);
		\draw [style=qWire] (6) to (5.center);
	\end{pgfonlayer}
\end{tikzpicture}.
\eeq

Note that generic linear combinations such as $\sum_i r_i \widetilde{T}_i$ need not correspond to any
process in the operational theory.
However, some linear combinations correspond to mixtures and coarse-grainings, and these {\em will} correspond to other processes in the operational theory. Namely,  if $T_1$ is a procedure that is a mixture of $T_2$ and $T_3$ with weights $\omega$ and $1-\omega$, then by Eq.~\eqref{mixprobconstr} it follows that
\beq
\begin{tikzpicture}
	\begin{pgfonlayer}{nodelayer}
		\node [style=small box] (0) at (2.25, 0) {$T_1$};
		\node [style=none] (1) at (2.25, 0.75) {};
		\node [style=none] (2) at (2.25, -0.75) {};
		\node [style=none] (3) at (1.75, 0.75) {};
		\node [style=none] (4) at (1.75, 1.25) {};
		\node [style=none] (5) at (3.5, 1.25) {};
		\node [style=none] (6) at (3.5, -1.25) {};
		\node [style=none] (7) at (1.75, -1.25) {};
		\node [style=none] (8) at (1.75, -0.75) {};
		\node [style=none] (9) at (3, 0.75) {};
		\node [style=none] (10) at (3, -0.75) {};
		\node [style=none] (11) at (3.25, 0) {$\tau$};
		\node [style=none] (12) at (1, 2) {};
		\node [style=none] (13) at (4.25, 2) {};
		\node [style=none] (14) at (4.25, -2) {};
		\node [style=none] (15) at (1, -2) {};
		\node [style=none] (16) at (4, -1.75) {\tiny $p$};
	\end{pgfonlayer}
	\begin{pgfonlayer}{edgelayer}
		\filldraw[fill=blue!20,draw=blue!40] (12.center) to (15.center) to (14.center) to (13.center) to cycle;
	\filldraw[fill=white,draw=black] (7.center) to (6.center) to (5.center) to (4.center) to (3.center) to (9.center) to (10.center) to (8.center) to cycle;	
		\draw [style=qWire] (0) to (1.center);
		\draw [style=qWire] (0) to (2.center);
	\end{pgfonlayer}
\end{tikzpicture}
\ =\omega\
\begin{tikzpicture}
	\begin{pgfonlayer}{nodelayer}
		\node [style=small box] (0) at (2.25, 0) {$T_2$};
		\node [style=none] (1) at (2.25, 0.75) {};
		\node [style=none] (2) at (2.25, -0.75) {};
		\node [style=none] (3) at (1.75, 0.75) {};
		\node [style=none] (4) at (1.75, 1.25) {};
		\node [style=none] (5) at (3.5, 1.25) {};
		\node [style=none] (6) at (3.5, -1.25) {};
		\node [style=none] (7) at (1.75, -1.25) {};
		\node [style=none] (8) at (1.75, -0.75) {};
		\node [style=none] (9) at (3, 0.75) {};
		\node [style=none] (10) at (3, -0.75) {};
		\node [style=none] (11) at (3.25, 0) {$\tau$};
		\node [style=none] (12) at (1, 2) {};
		\node [style=none] (13) at (4.25, 2) {};
		\node [style=none] (14) at (4.25, -2) {};
		\node [style=none] (15) at (1, -2) {};
		\node [style=none] (16) at (4, -1.75) {\tiny $p$};
	\end{pgfonlayer}
	\begin{pgfonlayer}{edgelayer}
		\filldraw[fill=blue!20,draw=blue!40] (12.center) to (15.center) to (14.center) to (13.center) to cycle;
	\filldraw[fill=white,draw=black] (7.center) to (6.center) to (5.center) to (4.center) to (3.center) to (9.center) to (10.center) to (8.center) to cycle;	
		\draw [style=qWire] (0) to (1.center);
		\draw [style=qWire] (0) to (2.center);
	\end{pgfonlayer}
\end{tikzpicture}
\ +(1-\omega)\
\begin{tikzpicture}
	\begin{pgfonlayer}{nodelayer}
		\node [style=small box] (0) at (2.25, 0) {$T_3$};
		\node [style=none] (1) at (2.25, 0.75) {};
		\node [style=none] (2) at (2.25, -0.75) {};
		\node [style=none] (3) at (1.75, 0.75) {};
		\node [style=none] (4) at (1.75, 1.25) {};
		\node [style=none] (5) at (3.5, 1.25) {};
		\node [style=none] (6) at (3.5, -1.25) {};
		\node [style=none] (7) at (1.75, -1.25) {};
		\node [style=none] (8) at (1.75, -0.75) {};
		\node [style=none] (9) at (3, 0.75) {};
		\node [style=none] (10) at (3, -0.75) {};
		\node [style=none] (11) at (3.25, 0) {$\tau$};
		\node [style=none] (12) at (1, 2) {};
		\node [style=none] (13) at (4.25, 2) {};
		\node [style=none] (14) at (4.25, -2) {};
		\node [style=none] (15) at (1, -2) {};
		\node [style=none] (16) at (4, -1.75) {\tiny $p$};
	\end{pgfonlayer}
	\begin{pgfonlayer}{edgelayer}
		\filldraw[fill=blue!20,draw=blue!40] (12.center) to (15.center) to (14.center) to (13.center) to cycle;
	\filldraw[fill=white,draw=black] (7.center) to (6.center) to (5.center) to (4.center) to (3.center) to (9.center) to (10.center) to (8.center) to cycle;	
		\draw [style=qWire] (0) to (1.center);
		\draw [style=qWire] (0) to (2.center);
	\end{pgfonlayer}
\end{tikzpicture}
\eeq
for all $\tau$, which in turn implies that
\beq
\begin{tikzpicture}
	\begin{pgfonlayer}{nodelayer}
		\node [style=small box] (0) at (0, 0) {$\widetilde{T}_1$};
		\node [style=none] (1) at (0, 0.75) {};
		\node [style=none] (2) at (0, -0.75) {};
		\node [style=none] (3) at (-0.5, 0.75) {};
		\node [style=none] (4) at (-0.5, 1.25) {};
		\node [style=none] (5) at (1.25, 1.25) {};
		\node [style=none] (6) at (1.25, -1.25) {};
		\node [style=none] (7) at (-0.5, -1.25) {};
		\node [style=none] (8) at (-0.5, -0.75) {};
		\node [style=none] (9) at (0.75, 0.75) {};
		\node [style=none] (10) at (0.75, -0.75) {};
		\node [style=none] (11) at (1, 0) {$\widetilde{\tau}$};
	\end{pgfonlayer}
	\begin{pgfonlayer}{edgelayer}
		\draw [style=qWire] (0) to (1.center);
		\draw [style=qWire] (0) to (2.center);
		\draw (4.center) to (3.center);
		\draw (3.center) to (9.center);
		\draw (9.center) to (10.center);
		\draw (10.center) to (8.center);
		\draw (8.center) to (7.center);
		\draw (7.center) to (6.center);
		\draw (6.center) to (5.center);
		\draw (5.center) to (4.center);
	\end{pgfonlayer}
\end{tikzpicture}
\ =\omega\
\begin{tikzpicture}
	\begin{pgfonlayer}{nodelayer}
		\node [style=small box] (0) at (0, 0) {$\widetilde{T}_2$};
		\node [style=none] (1) at (0, 0.75) {};
		\node [style=none] (2) at (0, -0.75) {};
		\node [style=none] (3) at (-0.5, 0.75) {};
		\node [style=none] (4) at (-0.5, 1.25) {};
		\node [style=none] (5) at (1.25, 1.25) {};
		\node [style=none] (6) at (1.25, -1.25) {};
		\node [style=none] (7) at (-0.5, -1.25) {};
		\node [style=none] (8) at (-0.5, -0.75) {};
		\node [style=none] (9) at (0.75, 0.75) {};
		\node [style=none] (10) at (0.75, -0.75) {};
		\node [style=none] (11) at (1, 0) {$\widetilde{\tau}$};
	\end{pgfonlayer}
	\begin{pgfonlayer}{edgelayer}
		\draw [style=qWire] (0) to (1.center);
		\draw [style=qWire] (0) to (2.center);
		\draw (4.center) to (3.center);
		\draw (3.center) to (9.center);
		\draw (9.center) to (10.center);
		\draw (10.center) to (8.center);
		\draw (8.center) to (7.center);
		\draw (7.center) to (6.center);
		\draw (6.center) to (5.center);
		\draw (5.center) to (4.center);
	\end{pgfonlayer}
\end{tikzpicture}
\ +(1-\omega)\
\begin{tikzpicture}
	\begin{pgfonlayer}{nodelayer}
		\node [style=small box] (0) at (0, 0) {$\widetilde{T}_3$};
		\node [style=none] (1) at (0, 0.75) {};
		\node [style=none] (2) at (0, -0.75) {};
		\node [style=none] (3) at (-0.5, 0.75) {};
		\node [style=none] (4) at (-0.5, 1.25) {};
		\node [style=none] (5) at (1.25, 1.25) {};
		\node [style=none] (6) at (1.25, -1.25) {};
		\node [style=none] (7) at (-0.5, -1.25) {};
		\node [style=none] (8) at (-0.5, -0.75) {};
		\node [style=none] (9) at (0.75, 0.75) {};
		\node [style=none] (10) at (0.75, -0.75) {};
		\node [style=none] (11) at (1, 0) {$\widetilde{\tau}$};
	\end{pgfonlayer}
	\begin{pgfonlayer}{edgelayer}
		\draw [style=qWire] (0) to (1.center);
		\draw [style=qWire] (0) to (2.center);
		\draw (4.center) to (3.center);
		\draw (3.center) to (9.center);
		\draw (9.center) to (10.center);
		\draw (10.center) to (8.center);
		\draw (8.center) to (7.center);
		\draw (7.center) to (6.center);
		\draw (6.center) to (5.center);
		\draw (5.center) to (4.center);
	\end{pgfonlayer}
\end{tikzpicture}, \label{mixprocrel}
\eeq
and so by the fact that quotiented operational theories are tomographic,
\beq
\widetilde{T}_1 = \omega \widetilde{T}_2 + (1-\omega) \widetilde{T}_3.
\eeq
Hence, the mixing relations between preparation procedures in the operational theory are captured by a convex structure in this representation. More generally we find that 
\beq \label{linearityconst}
\forall \tau \
\begin{tikzpicture}
	\begin{pgfonlayer}{nodelayer}
		\node [style=small box] (0) at (2.25, 0) {$T$};
		\node [style=none] (1) at (2.25, 0.75) {};
		\node [style=none] (2) at (2.25, -0.75) {};
		\node [style=none] (3) at (1.75, 0.75) {};
		\node [style=none] (4) at (1.75, 1.25) {};
		\node [style=none] (5) at (3.5, 1.25) {};
		\node [style=none] (6) at (3.5, -1.25) {};
		\node [style=none] (7) at (1.75, -1.25) {};
		\node [style=none] (8) at (1.75, -0.75) {};
		\node [style=none] (9) at (3, 0.75) {};
		\node [style=none] (10) at (3, -0.75) {};
		\node [style=none] (11) at (3.25, 0) {$\tau$};
		\node [style=none] (12) at (1, 2) {};
		\node [style=none] (13) at (4.25, 2) {};
		\node [style=none] (14) at (4.25, -2) {};
		\node [style=none] (15) at (1, -2) {};
		\node [style=none] (16) at (4, -1.75) {\tiny $p$};
	\end{pgfonlayer}
	\begin{pgfonlayer}{edgelayer}
		\filldraw[fill=blue!20,draw=blue!40] (12.center) to (15.center) to (14.center) to (13.center) to cycle;
	\filldraw[fill=white,draw=black] (7.center) to (6.center) to (5.center) to (4.center) to (3.center) to (9.center) to (10.center) to (8.center) to cycle;	
		\draw [style=qWire] (0) to (1.center);
		\draw [style=qWire] (0) to (2.center);
	\end{pgfonlayer}
\end{tikzpicture}
=\sum_ir_i
\begin{tikzpicture}
	\begin{pgfonlayer}{nodelayer}
		\node [style=small box] (0) at (2.25, 0) {$T_i$};
		\node [style=none] (1) at (2.25, 0.75) {};
		\node [style=none] (2) at (2.25, -0.75) {};
		\node [style=none] (3) at (1.75, 0.75) {};
		\node [style=none] (4) at (1.75, 1.25) {};
		\node [style=none] (5) at (3.5, 1.25) {};
		\node [style=none] (6) at (3.5, -1.25) {};
		\node [style=none] (7) at (1.75, -1.25) {};
		\node [style=none] (8) at (1.75, -0.75) {};
		\node [style=none] (9) at (3, 0.75) {};
		\node [style=none] (10) at (3, -0.75) {};
		\node [style=none] (11) at (3.25, 0) {$\tau$};
		\node [style=none] (12) at (1, 2) {};
		\node [style=none] (13) at (4.25, 2) {};
		\node [style=none] (14) at (4.25, -2) {};
		\node [style=none] (15) at (1, -2) {};
		\node [style=none] (16) at (4, -1.75) {\tiny $p$};
	\end{pgfonlayer}
	\begin{pgfonlayer}{edgelayer}
		\filldraw[fill=blue!20,draw=blue!40] (12.center) to (15.center) to (14.center) to (13.center) to cycle;
	\filldraw[fill=white,draw=black] (7.center) to (6.center) to (5.center) to (4.center) to (3.center) to (9.center) to (10.center) to (8.center) to cycle;	
		\draw [style=qWire] (0) to (1.center);
		\draw [style=qWire] (0) to (2.center);
	\end{pgfonlayer}
\end{tikzpicture}
\iff
\widetilde{T} = \sum_i r_i \widetilde{T}_i
\eeq
for arbitrary coefficients $r_i$. Hence, for example, the coarse-graining relations that hold among operational effects are captured by the linear structure in this representation of the quotiented operational theory.

If one makes the standard assumption that every possible mixture of processes in the operational theory is another process in the operational theory, then it follows that the quotiented operational theory is convex.

Finally, the fact that we assumed that each system $A$ had a unique equivalence class of deterministic effects means that the quotiented operational theory will have a {\em unique} deterministic effect~\cite{chiribella2010probabilistic}  for each system: 
\beq
\begin{tikzpicture}
	\begin{pgfonlayer}{nodelayer}
		\node [style=none] (0) at (0.5, -0.75) {};
		\node [style=none] (1) at (0.5, 0.25) {};
		\node [style=upground] (2) at (0.5, 0.5) {};
		\node [style=right label] (4) at (0.5, -0.5) {$A$};
	\end{pgfonlayer}
	\begin{pgfonlayer}{edgelayer}
		\draw [qWire] (1.center) to (0.center);
	\end{pgfonlayer}
\end{tikzpicture}
\quad :=\quad
\begin{tikzpicture}
	\begin{pgfonlayer}{nodelayer}
		\node [style=right label] (4) at (0, -0.5) {$A$};
		\node [style=upground] (5) at (0, 0.5) {};
		\node [style=none] (6) at (0, 0.25) {};
		\node [style=none] (7) at (0, -1) {};
		\node [style=right label] (8) at (0, -1.5) {};
		\node [style=right label] (9) at (0, -0.5) {};
		\node [style=none] (10) at (1.25, -0.85) {\tiny $\sim$};
		\node [style=none] (11) at (1.5, -1) {};
		\node [style=none] (12) at (1.5, 1.5) {};
		\node [style=none] (13) at (-1.5, 1.5) {};
		\node [style=none] (14) at (-1.5, -1) {};
		\node [style=none] (15) at (0, -1.5) {};
		\node [style=none] (16) at (0.5, 0.5) {$c$};
	\end{pgfonlayer}
	\begin{pgfonlayer}{edgelayer}
		\draw [fill={black!50!blue!30!}, draw={black!40!blue!40}] (11.center)
			 to (12.center)
			 to (13.center)
			 to (14.center)
			 to cycle;
		\draw [qWire] (7.center) to (6.center);
		\draw [qWire] (7.center) to (15.center);
	\end{pgfonlayer}
\end{tikzpicture}
.
\eeq

In summary, a quotiented operational theory satisfies the key properties of a GPT: being tomographic,
 representability of each system  in $\mathds{R}^d$ (for some $d$), convexity, and uniqueness of the deterministic effect  for each system.  Henceforth, we will refer to the quotiented operational theory as {\em the GPT associated to the operational theory},  and we presume that every GPT can be achieved in this way. 

For example, quantum theory qua operational theory is the process theory whose processes are laboratory procedures (including contexts), while quantum theory qua GPT is the process theory whose processes are completely positive~\cite{NielsenAndChuang,Schmidcausal} trace-nonincreasing maps, whose states are density operators, and so on. When one quotients quantum theory qua operational theory, one obtains quantum theory qua GPT.

 It is worth noting that a quotiented operational theory should not be viewed as
an instance of an operational theory\footnote{ Although this is not captured by the formal definitions in this paper, it is captured by the more refined formalization in Ref.~\cite{schmid2020unscrambling}.}. In an operational theory, it is not merely that contexts are {\em permitted}, rather they are {\em required} in the definition of the procedures.
The primitives in an operational theory are laboratory procedures, and these necessarily involve a complete description of the context of a procedure. For example, ``prepare the maximally mixed state'' specifies a
preparation when viewing quantum theory as a quotiented operational theory, but not when viewing it as an operational theory. In the latter case, one must additionally specify how this mixture is achieved, e.g., which ensemble of pure states was prepared or which entangled state was followed by tracing of a subsystem. 

\subsubsection{Tomographic locality of a GPT}

Tomographic locality is a common assumption in the GPT literature---indeed, in some early work on GPTs it was considered such a basic principle that it was taken as part of the framework itself~\cite{barrett2007}. Intuitively, it states that processes can be characterized by local state preparations and local measurements. In this section, we will show that all tomographically local GPTs can be represented as subtheories $\widetilde{\Op} \subset \RL$, using arguments similar to those in, e.g., the duotensor formalism in Ref.~\cite{hardy2011reformulating}.

A GPT is said to satisfy {\em tomographic locality} if one can determine the identity of any process by local operations on each of its inputs and outputs:
\begin{align} \label{loctom}
 \forall \widetilde{E},\widetilde{E}',\widetilde{P},\widetilde{P}' \input{Diagrams/17b_TomLocGPT1b.tikz} &= \input{Diagrams/18b_TomLocGPT2b.tikz}\nonumber \\ &\rotatebox{90}{$\iff$} \nonumber  \\
 \input{Diagrams/19b_TomLocGPT3b.tikz}\quad &=\quad \input{Diagrams/20b_TomLocGPT4b.tikz}.
\end{align}
\noindent 
One can immediately verify that an operational theory satisfies Eq.~\eqref{tomlocforopthry} if and only if the GPT obtained by quotienting relative to operational equivalences is tomographically local. 

There are many equivalent characterizations of tomographic locality for a GPT. The most useful for us is the following condition, first introduced in Sections~6.8 and 9.3 of \cite{hardy2011reformulating}, which allows us to show that tomographically local GPTs can be represented as subtheories of $\RL$.

Consider an arbitrary set of linearly independent and spanning states $\{\widetilde{P}^A_{j} \}_{j = 1,2,...,m^A}$ on system $A$ and an arbitrary set of linearly independent and spanning effects $\{\widetilde{E}^B_{i} \}_{i= 1,2,...,m^B}$ on system $B$. (If the systems are composite, these should moreover be chosen as product states and product effects respectively.) Define the `transition matrix' in this basis for a process $\widetilde{T}$ from system $A$ to system $B$ as
\beq
\begin{tikzpicture}
	\begin{pgfonlayer}{nodelayer}
		\node [style=none] (10) at (0, -3) {};
		\node [style=none] (11) at (0, 3) {};
		\node [style=none] (12) at (0, 0) {$[\mathbf{N}_{\widetilde{T}}]_j^k$};
	\end{pgfonlayer}
\end{tikzpicture}
\quad := \quad
 \begin{tikzpicture}
	\begin{pgfonlayer}{nodelayer}
		\node [style=copoint] (1) at (0, 1.75) {$\widetilde{E}^B_k$};
		\node [style=small box] (2) at (0, 0) {$\widetilde{T}$};
		\node [style=right label] (3) at (0, 0.875) {$B$};
		\node [style=point] (4) at (0, -1.75) {$\widetilde{P}^A_j$};
		\node [style=right label] (5) at (0, -1) {$A$};
	\end{pgfonlayer}
	\begin{pgfonlayer}{edgelayer}
		\draw [qWire] (1) to (2);
		\draw [qWire] (2) to (4);
	\end{pgfonlayer}
\end{tikzpicture}.
\eeq

\begin{lemma} \label{lemmadecomp}
A GPT is tomographically local if and only if one can decompose the identity process, denoted $\widetilde{\mathds{1}}_A$, for every system $A$ as
\begin{equation}
\begin{tikzpicture}
	\begin{pgfonlayer}{nodelayer}
		\node [style=none] (0) at (0, 1) {};
		\node [style=none] (1) at (0, -1) {};
		\node [style=right label] (2) at (0, -0.75) {$A$};
	\end{pgfonlayer}
	\begin{pgfonlayer}{edgelayer}
		\draw [qWire] (0.center) to (1.center);
	\end{pgfonlayer}
\end{tikzpicture}

\quad = \quad \sum_{ij} [\mathbf{M}_{\widetilde{\mathds{1}}_A}]_i^j
\input{Diagrams/measprepA.tikz},
\label{eq:identitydecomp1}
\end{equation}
where $\mathbf{M}_{{\widetilde{\mathds{1}}_A}}$ is the matrix inverse of the transition matrix $\mathbf{N}_{{\widetilde{\mathds{1}}_A}}$ of the identity process, that is,
\beq \label{lem:identitydecomp}
\mathbf{M}_{{\widetilde{\mathds{1}}}_A}:= \mathbf{N}_{\widetilde{\mathds{1}}_A}^{-1}.
\eeq
\end{lemma}

This was originally shown in \cite[Sec. 6.8]{hardy2011reformulating}, and we reprove this in our notation in Appendix~\ref{proofTL}.

The vector space spanned by the set  $\{\widetilde{P}_{j}^A \}_{j = 1,2,...,m^A}$ of states is $\mathds{R}^{m^A}$, and the vector space spanned by the set $\{\widetilde{E}_{i}^B \}_{i= 1,2,...,m^B}$ of effects is $\mathds{R}^{m^B}$. Note that in general, $\widetilde{E}_j^A\circ\widetilde{P}_i^A \ne \delta_{ij}$\footnote{This fails to be an equality whenever the spanning effects $E_i$ do not perfectly distinguish the spanning states $P_j$. Note that this will be the generic case, as it is only in a simplicial GPT that there is a spanning set of states which can be perfectly distinguished.
}, which generically implies that $\mathbf{M}_{{\widetilde{\mathds{1}}}_A}$ is not the identity matrix (nor is it equal to $\mathbf{N}_{\widetilde{\mathds{1}}_A}$), counter to intuitions one might have from working with orthonormal bases.
The following corollary then makes explicit some extra structure which was implicit in the vector representation $\mathbf{R}_{\widetilde{T}}$ of the previous section. In particular, it
 shows that the vector space $\mathds{R}^{m^{A\to B}}$ of transformations from $A$ to $B$ is isomorphic to the vector space of linear maps from $\mathds{R}^{m^A}\to \mathds{R}^{m^B}$, where a process $\widetilde{T}$ is represented as a vector $\mathbf{R}_{\widetilde{T}}$ in the former and a matrix $\mathbf{M}_{\widetilde{T}}$ in the latter.

\begin{corollary}\label{corolTL}
A GPT is tomographically local if and only if one can decompose any process $\widetilde{T}$ as
\begin{equation}
\input{Diagrams/processAB.tikz}
\quad = \quad \sum_{ij} [\mathbf{M}_{\widetilde{T}}]_{i}^{j}
\input{Diagrams/measprepAB.tikz},
\label{eq:Tdecomp}
\end{equation}
where $\mathbf{M}_{\widetilde{T}} = \mathbf{M}_{\widetilde{\mathds{1}}_B}\circ \mathbf{N}_{\widetilde{T}}\circ\mathbf{M}_{\widetilde{\mathds{1}}_A}$.
\label{cor:Tdecomp}
\end{corollary}
\begin{proof}
To prove that a GPT is tomographically local if one can decompose any process as in Eq.~\eqref{eq:Tdecomp}, simply note that for the special case of $\widetilde{T} = \widetilde{\mathds{1}}_A$, Eq.~\eqref{eq:Tdecomp} implies Eq.~\eqref{eq:identitydecomp1}, and hence, by Lemma~\ref{lemmadecomp}, implies tomographic locality.

To prove the converse, we assume tomographic locality and apply Lemma~\ref{lemmadecomp} to decompose the input and output systems of an arbitrary process $\widetilde{T}$ as
\beq
\begin{tikzpicture}
	\begin{pgfonlayer}{nodelayer}
		\node [style=none] (0) at (0, 1.5) {};
		\node [style=none] (1) at (0, 1.25) {};
		\node [style={small box}] (2) at (0, -0) {$\widetilde{T}$};
		\node [style={right label}] (3) at (0, 1) {$B$};
		\node [style=none] (4) at (0, -1.25) {};
		\node [style={right label}] (5) at (0, -1.25) {$A$};
	\end{pgfonlayer}
	\begin{pgfonlayer}{edgelayer}
		\draw [qWire] (1.center) to (2);
		\draw [qWire] (2) to (4.center);
	\end{pgfonlayer}
\end{tikzpicture}
\quad =\quad
 \begin{tikzpicture}
	\begin{pgfonlayer}{nodelayer}
		\node [style=copoint] (0) at (0, 1.25) {$\widetilde{E}_k^B$};
		\node [style={small box}] (1) at (0, 0) {$\widetilde{T}$};
		\node [style=point] (2) at (0, -1.25) {$\widetilde{P}_j^A$};
		\node [style=none] (3) at (-2.5, -2.25) {$\displaystyle \sum_{ij}[\mathbf{M}_{\widetilde{\mathds{1}}_A}]_i^j$};
		\node [style=none] (4) at (-2.5, 1.5) {$\displaystyle\sum_{kl}[\mathbf{M}_{\widetilde{\mathds{1}}_B}]_k^l$};
		\node [style=point] (5) at (0, 3.5) {$\widetilde{P}_l^B$};
		\node [style=copoint] (6) at (0, -3.5) {$\widetilde{E}_i^A$};
		\node [style=none] (7) at (0, -4.5) {};
		\node [style=none] (8) at (0, 4.5) {};
	\end{pgfonlayer}
	\begin{pgfonlayer}{edgelayer}
		\draw [qWire] (0) to (1);
		\draw [qWire] (1) to (2);
		\draw [qWire] (6) to (7.center);
		\draw [qWire] (5) to (8.center);
	\end{pgfonlayer}
\end{tikzpicture}
\quad = \quad
\begin{tikzpicture}
	\begin{pgfonlayer}{nodelayer}
		\node [style=none] (0) at (-2, -2) {$\displaystyle \sum_{ij}[\mathbf{M}_{\widetilde{\mathds{1}}_A}]_i^j$};
		\node [style=none] (1) at (-2, 1.75) {$\displaystyle\sum_{kl}[\mathbf{M}_{\widetilde{\mathds{1}}_B}]_k^l$};
		\node [style=point] (2) at (0, 2.75) {$\widetilde{P}_l^B$};
		\node [style=copoint] (3) at (0, -2.75) {$\widetilde{E}_i^A$};
		\node [style=none] (4) at (0, -3.75) {};
		\node [style=none] (5) at (0, 3.75) {};
		\node [style=none] (6) at (0, 0) {$[\mathbf{N}_{\widetilde{T}}]_j^k$};
	\end{pgfonlayer}
	\begin{pgfonlayer}{edgelayer}
		\draw [qWire] (3) to (4.center);
		\draw [qWire] (2) to (5.center);
	\end{pgfonlayer}
\end{tikzpicture}
\eeq
which can be rewritten as:
\begin{align}
\begin{tikzpicture}
	\begin{pgfonlayer}{nodelayer}
		\node [style=none] (0) at (0, 1.5) {};
		\node [style=none] (1) at (0, 1.25) {};
		\node [style={small box}] (2) at (0, -0) {$\widetilde{T}$};
		\node [style={right label}] (3) at (0, 1) {$B$};
		\node [style=none] (4) at (0, -1.25) {};
		\node [style={right label}] (5) at (0, -1.25) {$A$};
	\end{pgfonlayer}
	\begin{pgfonlayer}{edgelayer}
		\draw [qWire] (1.center) to (2);
		\draw [qWire] (2) to (4.center);
	\end{pgfonlayer}
\end{tikzpicture}
\ &= \ \sum_{il}\left(\sum_{jk}
[\mathbf{M}_{\widetilde{\mathds{1}}_B}]_k^l
[\mathbf{N}_{\widetilde{T}}]_j^k
[\mathbf{M}_{\widetilde{\mathds{1}}_A}]_i^j
\right)
 \begin{tikzpicture}
	\begin{pgfonlayer}{nodelayer}
		\node [style=point] (8) at (0, 1.125) {$\widetilde{P}_l^B$};
		\node [style=copoint] (9) at (0, -1.125) {$\widetilde{E}_i^A$};
		\node [style=none] (10) at (0, -2) {};
		\node [style=none] (11) at (0, 2) {};
	\end{pgfonlayer}
	\begin{pgfonlayer}{edgelayer}
		\draw [qWire] (9) to (10.center);
		\draw [qWire] (8) to (11.center);
	\end{pgfonlayer}
\end{tikzpicture}
\\
&= \
\sum_{il}
[\mathbf{M}_{\widetilde{\mathds{1}}_B}\circ
\mathbf{N}_{\widetilde{T}}\circ\mathbf{M}_{\widetilde{\mathds{1}}_A}]_i^l
 \begin{tikzpicture}
	\begin{pgfonlayer}{nodelayer}
		\node [style=point] (8) at (0, 1.125) {$\widetilde{P}_l^B$};
		\node [style=copoint] (9) at (0, -1.125) {$\widetilde{E}_i^A$};
		\node [style=none] (10) at (0, -2) {};
		\node [style=none] (11) at (0, 2) {};
	\end{pgfonlayer}
	\begin{pgfonlayer}{edgelayer}
		\draw [qWire] (9) to (10.center);
		\draw [qWire] (8) to (11.center);
	\end{pgfonlayer}
\end{tikzpicture}
\end{align}
at which point we can simply identify $\mathbf{M}_{\widetilde{T}} = \mathbf{M}_{\widetilde{\mathds{1}}_B}\circ \mathbf{N}_{\widetilde{T}}\circ\mathbf{M}_{\widetilde{\mathds{1}}_A}$, giving the desired result.
\end{proof}

Given the vector representation of two equivalence classes, $\mathbf{R}_{\widetilde{T}}$ and $\mathbf{R}_{\widetilde{T}'}$, we showed how to compute the representation of either the sequential or parallel composition of these (via $\smallsquare$ and $\smallboxtimes$, respectively). However, if we represent equivalence classes by matrices $\mathbf{M}_{\widetilde{T}}$ instead, then how must we represent the parallel and sequential composition of processes?
It turns out that parallel composition is represented by
\beq \label{parallelcompr}
\mathbf{M}_{\widetilde{T}\otimes \widetilde{T}'} = \mathbf{M}_{\widetilde{T}}\otimes \mathbf{M}_{\widetilde{T}'}
\eeq
where the $\otimes$ on the left represents the parallel composition of equivalence classes, while on the right it represents the tensor product of the two matrices.
Meanwhile, the sequential composition of this matrix representation is given by
\beq \label{nastycomp}
\mathbf{M}_{\widetilde{T}'\circ \widetilde{T}'} = \mathbf{M}_{\widetilde{T}'}\circ \mathbf{N}_{\widetilde{\mathds{1}}_B}\circ \mathbf{M}_{\widetilde{T}}
\eeq
where on the left-hand side $\circ$ represents the sequential composition of the equivalence classes, while on the right-hand side it represents matrix multiplication. These two facts are proven in Appendix~\ref{comprepn}.

The fact that Eq.~\eqref{nastycomp} is not simply the sequential composition rule for $\RL$, namely the matrix product of $\mathbf{M}_{\widetilde{T}'}$ and $\mathbf{M}_{\widetilde{T}}$, implies that this matrix representation is not a subtheory of $\RL$, nor even some other diagram-preserving representation of the GPT.
This form of composition has, however, appeared numerous times in the literature, for example in Refs.~\cite{ferrie2008frame,hardy2011reformulating,appleby2017introducing,van2017quantum}. There is, moreover, a well known trick to turn this representation into a diagram-preserving representation in $\RL$: one simply defines a new matrix representation by replacing $\mathbf{M}_{\widetilde{T}}$ with $\mathbf{M}_{\widetilde{T}}\circ \mathbf{N}_{\widetilde{\mathds{1}}_A}$\footnote{In the language of duotensors \cite{hardy2011reformulating} this means that we put the duotensors into standard form.}.  It is then easy to verify that these do indeed compose using the standard composition rules (tensor products for parallel composition and matrix multiplication for sequential composition), and to verify that the identity process is represented by the identity matrix.

Putting all of this together we arrive at the following.
\begin{theorem} \label{OpinRL}
Any tomographically local GPT has a diagram-preserving representation in $\RL$ given by
the map
\[A \mapsto \mathds{R}^{m^A}\]
on systems and the map
\[\widetilde{T} \mapsto \mathbf{M}_{\widetilde{T}}\circ \mathbf{N}_{\widetilde{\mathds{1}}_A}\]
 on processes, where
\beq [\mathbf{N}_{\widetilde{T}}]_j^k:=
\begin{tikzpicture}
	\begin{pgfonlayer}{nodelayer}
		\node [style=copoint] (1) at (0, 1.75) {$\widetilde{E}_k^B$};
		\node [style=small box] (2) at (0, 0) {$\widetilde{T}$};
		\node [style=right label] (3) at (0, 0.75) {$B$};
		\node [style=point] (4) at (0, -1.75) {$\widetilde{P}_j^A$};
		\node [style=right label] (5) at (0, -.875) {$A$};
	\end{pgfonlayer}
	\begin{pgfonlayer}{edgelayer}
		\draw [qWire] (1) to (2);
		\draw [qWire] (2) to (4);
	\end{pgfonlayer}
\end{tikzpicture}
\eeq
 for some basis of states $\{\widetilde{P}_j^A\}$ and effects $\{\widetilde{E}_k ^B\}$,
and where 
\beq
\mathbf{M}_{\widetilde{T}} := \mathbf{N}_{\widetilde{\mathds{1}}_B}^{-1}\circ \mathbf{N}_{\widetilde{T}} \circ \mathbf{N}_{\widetilde{\mathds{1}}_A}^{-1}.
\eeq
\end{theorem}
This result is implicit in the work of Refs.~\cite{hardy2011reformulating,chiribella2016quantum} and more explicit in the quantum case of \cite{van2017quantum}. 

Effectively this means that we can view any tomographically local GPT simply as a suitably defined subtheory $\widetilde{\Op}\subset \RL$.  For the remainder of this paper we restrict our attention to tomographically local GPTs, and we will moreover abuse notation and simply denote the linear maps in this representation by $\widetilde{T}$ rather than by $\mathbf{M}_{\widetilde{T}}\circ \mathbf{N}_{\widetilde{\mathds{1}}_A}$, and similarly, the vector spaces as $A$ rather than by $\mathds{R}^{m^A}$. That is, we will neglect to make the distinction between the quotiented operational theory and its representation as a subtheory of $\RL$, as preserving the distinction is unwieldy and typically unhelpful.

 Quantum theory is an example of an operational theory, and it is well known that the GPT representation of quantum theory is tomographically local. The latter is a subtheory of $\RL$, as $\mathcal{B(H)}$ is a real vector space and completely positive trace non-increasing maps are just a particular class of linear maps between these vector spaces. Classical theory, the Spekkens toy model~\cite{spekkens2007evidence}, and the stabilizer subtheory~\cite{Gottesman:1998hu} for arbitrary dimensions are also tomographically local.  Examples of GPTs which are not tomographically local are real quantum theory \cite{hardy2012limited} and the real stabilizer subtheory.

\subsection{Representations of operational theories and GPTs}

One often wishes to find alternative representations of an operational theory or a GPT, e.g., as an ontological model or a quasiprobabilistic model (to be defined shortly). A key motivation for studying ontological models is the attempt to find an explanation for the statistics in terms of some underlying properties of the relevant systems, especially if this explanation can be said to be {\em classical} in some well-motivated sense. In this section, we introduce the definition of ontological models and quasiprobabilistic models, and in the next section we discuss under what conditions one can say that such representations provide a classical explanation of the operational theory or GPT which they describe.

\subsubsection{Ontological models}\label{subsec:ont}

An ontological model is a map associating to every system $S$ a set $\Lambda_S$ of ontic states, and associating to every process a stochastic map from the ontic state space associated to the input systems to the ontic state space associated with the output systems.

It is important to distinguish between ontological models of operational theories and ontological models of GPTs, as was shown in Ref.~\cite{schmid2019characterization}. In particular, the former allows for context-dependence while the latter does not. See App.~\ref{contextsingpts} for a detailed discussion of this point.

\begin{definition}[Ontological models of operational theories] \label{defnontop}
An ontological model~\cite{harrigan2007} of an operational theory $\Op$ is a diagram-preserving map \colorbox{BurntOrange!20}{$\xi:\Op\to \SubS$}, depicted as
\[\xi::\ \input{Diagrams/21b_Ont1b.tikz}\ \mapsto\ \input{Diagrams/21_Ont1.tikz},\]
from the operational theory to the process theory $\SubS$,
where the map satisfies three properties:
\begin{enumerate}
\item It represents all deterministic effects in the operational theory appropriately:
\[\begin{tikzpicture}
	\begin{pgfonlayer}{nodelayer}
		\node [style=upground] (0) at (0, 0) {};
		\node [style=none] (7) at (0, -0.25) {};
		\node [style=none] (8) at (0, -1.5) {};
		\node [style=right label] (9) at (0, -2) {$\mathds{R}^{\Lambda_A}$};
		\node [style=right label] (10) at (0, -1) {$A$};
		\node [style=none] (13) at (1.25, -1.25) {\tiny ${\xi}$};
		\node [style=none] (14) at (1.5, -1.5) {};
		\node [style=none] (15) at (1.5, 1) {};
		\node [style=none] (16) at (-1.5, 1) {};
		\node [style=none] (17) at (-1.5, -1.5) {};
		\node [style=none] (18) at (0, -2) {};
		\node[style=none] (19) at (0.5,0.25) {c};
	\end{pgfonlayer}
	\begin{pgfonlayer}{edgelayer}
			\filldraw[fill=BurntOrange!20,draw=BurntOrange!40](14.center) to (15.center) to (16.center) to (17.center) to cycle;
		\draw [qWire] (8.center) to (7.center);
		\draw (8.center) to (18.center);
	\end{pgfonlayer}
\end{tikzpicture}
\ =\
\begin{tikzpicture}
	\begin{pgfonlayer}{nodelayer}
		\node [style=upground] (0) at (0, -0) {};
		\node [style=none] (7) at (0, -0.25) {};
		\node [style=none] (8) at (0, -1.5) {};
		\node [style=right label] (9) at (0, -1.5) {$\mathds{R}^{\Lambda_A}$};
	\end{pgfonlayer}
	\begin{pgfonlayer}{edgelayer}
		\draw (8.center) to (7.center);
	\end{pgfonlayer}
\end{tikzpicture}
 \ =\
\mathbf{1}.\]
\item It reproduces the operational predictions of the operational theory (i.e., is empirically adequate). That is, for all closed diagrams:
\[\input{Diagrams/24_Ont4.tikz}\ =\ \input{Diagrams/25_Ont5.tikz}\ =\mathrm{Pr}(E,P).\]
\item It preserves the convex and coarse-graining relations between operational procedures. E.g., if $T_1$ is a procedure
that is a mixture of $T_2$ and $T_3$ with weights $\omega$ and $1-\omega$, respectively, then it must hold that
\beq
\begin{tikzpicture}
	\begin{pgfonlayer}{nodelayer}
		\node [style=none] (0) at (0, 0) {$
{T_1}$};
		\node [style=none] (1) at (-0.5, 0.5) {};
		\node [style=none] (2) at (0.5, 0.5) {};
		\node [style=none] (3) at (0.5, -0.5) {};
		\node [style=none] (4) at (-0.5, -0.5) {};
		\node [style=none] (5) at (0, 0.5) {};
		\node [style=none] (6) at (0, 1.5) {};
		\node [style=none] (7) at (0, -0.5) {};
		\node [style=none] (8) at (0, -1.5) {};
		\node [style=none] (13) at (1.25, -1.25) {\tiny ${\xi}$};
		\node [style=none] (14) at (1.5, -1.5) {};
		\node [style=none] (15) at (1.5, 1.5) {};
		\node [style=none] (16) at (-1.5, 1.5) {};
		\node [style=none] (17) at (-1.5, -1.5) {};
		\node [style=none] (18) at (0, 2) {};
		\node [style=none] (19) at (0, -2) {};
	\end{pgfonlayer}
	\begin{pgfonlayer}{edgelayer}
			\filldraw[fill=BurntOrange!20,draw=BurntOrange!40](14.center) to (15.center) to (16.center) to (17.center) to cycle;
		\filldraw[fill=white,draw=black] (1.center) to (2.center) to (3.center) to (4.center) to cycle;
		\draw [qWire] (5.center) to (6.center);
		\draw [qWire] (8.center) to (7.center);
		\draw (18.center) to (6.center);
		\draw (8.center) to (19.center);
	\end{pgfonlayer}
\end{tikzpicture}
\ = \omega\
\begin{tikzpicture}
	\begin{pgfonlayer}{nodelayer}
		\node [style=none] (0) at (0, 0) {$
{T_2}$};
		\node [style=none] (1) at (-0.5, 0.5) {};
		\node [style=none] (2) at (0.5, 0.5) {};
		\node [style=none] (3) at (0.5, -0.5) {};
		\node [style=none] (4) at (-0.5, -0.5) {};
		\node [style=none] (5) at (0, 0.5) {};
		\node [style=none] (6) at (0, 1.5) {};
		\node [style=none] (7) at (0, -0.5) {};
		\node [style=none] (8) at (0, -1.5) {};
		\node [style=none] (13) at (1.25, -1.25) {\tiny ${\xi}$};
		\node [style=none] (14) at (1.5, -1.5) {};
		\node [style=none] (15) at (1.5, 1.5) {};
		\node [style=none] (16) at (-1.5, 1.5) {};
		\node [style=none] (17) at (-1.5, -1.5) {};
		\node [style=none] (18) at (0, 2) {};
		\node [style=none] (19) at (0, -2) {};
	\end{pgfonlayer}
	\begin{pgfonlayer}{edgelayer}
			\filldraw[fill=BurntOrange!20,draw=BurntOrange!40](14.center) to (15.center) to (16.center) to (17.center) to cycle;
		\filldraw[fill=white,draw=black] (1.center) to (2.center) to (3.center) to (4.center) to cycle;
		\draw [qWire] (5.center) to (6.center);
		\draw [qWire] (8.center) to (7.center);
		\draw (18.center) to (6.center);
		\draw (8.center) to (19.center);
	\end{pgfonlayer}
\end{tikzpicture}
\ +(1-\omega)\
\begin{tikzpicture}
	\begin{pgfonlayer}{nodelayer}
		\node [style=none] (0) at (0, 0) {${T_3}$};
		\node [style=none] (1) at (-0.5, 0.5) {};
		\node [style=none] (2) at (0.5, 0.5) {};
		\node [style=none] (3) at (0.5, -0.5) {};
		\node [style=none] (4) at (-0.5, -0.5) {};
		\node [style=none] (5) at (0, 0.5) {};
		\node [style=none] (6) at (0, 1.5) {};
		\node [style=none] (7) at (0, -0.5) {};
		\node [style=none] (8) at (0, -1.5) {};
		\node [style=none] (13) at (1.25, -1.25) {\tiny ${\xi}$};
		\node [style=none] (14) at (1.5, -1.5) {};
		\node [style=none] (15) at (1.5, 1.5) {};
		\node [style=none] (16) at (-1.5, 1.5) {};
		\node [style=none] (17) at (-1.5, -1.5) {};
		\node [style=none] (18) at (0, 2) {};
		\node [style=none] (19) at (0, -2) {};
	\end{pgfonlayer}
	\begin{pgfonlayer}{edgelayer}
			\filldraw[fill=BurntOrange!20,draw=BurntOrange!40](14.center) to (15.center) to (16.center) to (17.center) to cycle;
		\filldraw[fill=white,draw=black] (1.center) to (2.center) to (3.center) to (4.center) to cycle;
		\draw [qWire] (5.center) to (6.center);
		\draw [qWire] (8.center) to (7.center);
		\draw (18.center) to (6.center);
		\draw (8.center) to (19.center);
	\end{pgfonlayer}
\end{tikzpicture}.
\eeq
\end{enumerate}
\end{definition}

This diagrammatic definition of an ontological model reproduces the usual notions~\cite{Spekkens2005} of ontological representations of preparation procedures and of operational effects. In particular, an operational preparation procedure is an operational process with a trivial input, and by diagram preservation of $\xi$, this is mapped to a process in $\SubS$ with a trivial input, that is, to a probability distribution over the ontic states:
$P \mapsto \xi(P) : \Lambda \to \mathds{R}^+ \text{ s.t. } \sum_\lambda \xi(P)(\lambda)=1$. Similarly, an operational effect is an operational process with a trivial output, and by diagram preservation of $\xi$ is mapped to a substochastic process with a trivial output, that is, to a response function over the ontic states: $E \mapsto \xi(E) : \Lambda \to [0,1] \text{ s.t. }  \xi(E)(\lambda) \leq 1$.

\begin{definition}[Ontological models of GPTs] \label{defnontgpt}
An ontological model $\widetilde{\xi}$ of a GPT $\widetilde{\Op}$ is a diagram-preserving map \colorbox{Red!20}{$\widetilde{\xi}: \widetilde{\Op} \to \SubS$}, depicted as
\[\widetilde{\xi}:: \ \input{Diagrams/26b_GPT1b.tikz}\ \mapsto \ \input{Diagrams/26_GPT1.tikz},\]
from the GPT to the process theory $\SubS$,
where the map satisfies three properties:
\begin{enumerate}
\item It represents the deterministic effect  for each system  appropriately: \[
\begin{tikzpicture}
	\begin{pgfonlayer}{nodelayer}
		\node [style=upground] (0) at (0, 0) {};
		\node [style=none] (7) at (0, -0.25) {};
		\node [style=none] (8) at (0, -1.5) {};
		\node [style=right label] (9) at (0, -2) {$\mathds{R}^{\Lambda_A}$};
		\node [style=right label] (10) at (0, -1) {$A$};
		\node [style=none] (13) at (1.25, -1.25) {\tiny $\tilde{\xi}$};
		\node [style=none] (14) at (1.5, -1.5) {};
		\node [style=none] (15) at (1.5, 1) {};
		\node [style=none] (16) at (-1.5, 1) {};
		\node [style=none] (17) at (-1.5, -1.5) {};
		\node [style=none] (18) at (0, -2) {};
	\end{pgfonlayer}
	\begin{pgfonlayer}{edgelayer}
			\filldraw[fill=Red!20,draw=Red!40](14.center) to (15.center) to (16.center) to (17.center) to cycle;
		\draw [qWire] (8.center) to (7.center);
		\draw (8.center) to (18.center);
	\end{pgfonlayer}
\end{tikzpicture}
\ =\ \begin{tikzpicture}
	\begin{pgfonlayer}{nodelayer}
		\node [style=upground] (0) at (0, -0) {};
		\node [style=none] (7) at (0, -0.25) {};
		\node [style=none] (8) at (0, -1.5) {};
		\node [style=right label] (9) at (0, -1.5) {$\mathds{R}^{\Lambda_A}$};
	\end{pgfonlayer}
	\begin{pgfonlayer}{edgelayer}
		\draw (8.center) to (7.center);
	\end{pgfonlayer}
\end{tikzpicture}
 =\ \mathbf{1}.\]
\item It reproduces the operational predictions of the GPT (i.e., is empirically adequate), so that for all closed diagrams,
\beq\label{eq:EmpAdeq}\input{Diagrams/24_empadG.tikz}\ =\ \input{Diagrams/25_empadG.tikz}\ =\mathrm{Pr}(\widetilde{E},\widetilde{P}).\eeq
\item It preserves the convex and coarse-graining relations between operational procedures. E.g., if
\beq
\begin{tikzpicture}
	\begin{pgfonlayer}{nodelayer}
		\node [style=none] (0) at (0, 0) {$
\widetilde{T}_1$};
		\node [style=none] (1) at (-0.5, 0.5) {};
		\node [style=none] (2) at (0.5, 0.5) {};
		\node [style=none] (3) at (0.5, -0.5) {};
		\node [style=none] (4) at (-0.5, -0.5) {};
		\node [style=none] (5) at (0, 0.5) {};
		\node [style=none] (6) at (0, 1.5) {};
		\node [style=none] (7) at (0, -0.5) {};
		\node [style=none] (8) at (0, -1.5) {};
	\end{pgfonlayer}
	\begin{pgfonlayer}{edgelayer}
		\draw [qWire] (5.center) to (6.center);
		\draw [qWire] (8.center) to (7.center);
		\draw (1.center) to (2.center);
		\draw (2.center) to (3.center);
		\draw (3.center) to (4.center);
		\draw (4.center) to (1.center);
	\end{pgfonlayer}
\end{tikzpicture}
\ = \omega\
\begin{tikzpicture}
	\begin{pgfonlayer}{nodelayer}
		\node [style=none] (0) at (0, 0) {$
\widetilde{T}_2$};
		\node [style=none] (1) at (-0.5, 0.5) {};
		\node [style=none] (2) at (0.5, 0.5) {};
		\node [style=none] (3) at (0.5, -0.5) {};
		\node [style=none] (4) at (-0.5, -0.5) {};
		\node [style=none] (5) at (0, 0.5) {};
		\node [style=none] (6) at (0, 1.5) {};
		\node [style=none] (7) at (0, -0.5) {};
		\node [style=none] (8) at (0, -1.5) {};
	\end{pgfonlayer}
	\begin{pgfonlayer}{edgelayer}
		\draw [qWire] (5.center) to (6.center);
		\draw [qWire] (8.center) to (7.center);
		\draw (1.center) to (2.center);
		\draw (2.center) to (3.center);
		\draw (3.center) to (4.center);
		\draw (4.center) to (1.center);
	\end{pgfonlayer}
\end{tikzpicture}
\ + (1-\omega)\
 \begin{tikzpicture}
	\begin{pgfonlayer}{nodelayer}
		\node [style=none] (0) at (0, 0) {$
\widetilde{T}_3$};
		\node [style=none] (1) at (-0.5, 0.5) {};
		\node [style=none] (2) at (0.5, 0.5) {};
		\node [style=none] (3) at (0.5, -0.5) {};
		\node [style=none] (4) at (-0.5, -0.5) {};
		\node [style=none] (5) at (0, 0.5) {};
		\node [style=none] (6) at (0, 1.5) {};
		\node [style=none] (7) at (0, -0.5) {};
		\node [style=none] (8) at (0, -1.5) {};
	\end{pgfonlayer}
	\begin{pgfonlayer}{edgelayer}
		\draw [qWire] (5.center) to (6.center);
		\draw [qWire] (8.center) to (7.center);
		\draw (1.center) to (2.center);
		\draw (2.center) to (3.center);
		\draw (3.center) to (4.center);
		\draw (4.center) to (1.center);
	\end{pgfonlayer}
\end{tikzpicture}
\eeq
then it must hold that
 \beq
\begin{tikzpicture}
	\begin{pgfonlayer}{nodelayer}
		\node [style=none] (0) at (0, 0) {$
\widetilde{T}_1$};
		\node [style=none] (1) at (-0.5, 0.5) {};
		\node [style=none] (2) at (0.5, 0.5) {};
		\node [style=none] (3) at (0.5, -0.5) {};
		\node [style=none] (4) at (-0.5, -0.5) {};
		\node [style=none] (5) at (0, 0.5) {};
		\node [style=none] (6) at (0, 1.5) {};
		\node [style=none] (7) at (0, -0.5) {};
		\node [style=none] (8) at (0, -1.5) {};
		\node [style=none] (13) at (1.25, -1.25) {\tiny $\tilde{\xi}$};
		\node [style=none] (14) at (1.5, -1.5) {};
		\node [style=none] (15) at (1.5, 1.5) {};
		\node [style=none] (16) at (-1.5, 1.5) {};
		\node [style=none] (17) at (-1.5, -1.5) {};
		\node [style=none] (18) at (0, 2) {};
		\node [style=none] (19) at (0, -2) {};
	\end{pgfonlayer}
	\begin{pgfonlayer}{edgelayer}
			\filldraw[fill=Red!20,draw=Red!40](14.center) to (15.center) to (16.center) to (17.center) to cycle;
		\filldraw[fill=white,draw=black] (1.center) to (2.center) to (3.center) to (4.center) to cycle;
		\draw [qWire] (5.center) to (6.center);
		\draw [qWire] (8.center) to (7.center);
		\draw (18.center) to (6.center);
		\draw (8.center) to (19.center);
	\end{pgfonlayer}
\end{tikzpicture}
\ = \omega\
\begin{tikzpicture}
	\begin{pgfonlayer}{nodelayer}
		\node [style=none] (0) at (0, 0) {$
\widetilde{T}_2$};
		\node [style=none] (1) at (-0.5, 0.5) {};
		\node [style=none] (2) at (0.5, 0.5) {};
		\node [style=none] (3) at (0.5, -0.5) {};
		\node [style=none] (4) at (-0.5, -0.5) {};
		\node [style=none] (5) at (0, 0.5) {};
		\node [style=none] (6) at (0, 1.5) {};
		\node [style=none] (7) at (0, -0.5) {};
		\node [style=none] (8) at (0, -1.5) {};
		\node [style=none] (13) at (1.25, -1.25) {\tiny $\tilde{\xi}$};
		\node [style=none] (14) at (1.5, -1.5) {};
		\node [style=none] (15) at (1.5, 1.5) {};
		\node [style=none] (16) at (-1.5, 1.5) {};
		\node [style=none] (17) at (-1.5, -1.5) {};
		\node [style=none] (18) at (0, 2) {};
		\node [style=none] (19) at (0, -2) {};
	\end{pgfonlayer}
	\begin{pgfonlayer}{edgelayer}
			\filldraw[fill=Red!20,draw=Red!40](14.center) to (15.center) to (16.center) to (17.center) to cycle;
		\filldraw[fill=white,draw=black] (1.center) to (2.center) to (3.center) to (4.center) to cycle;
		\draw [qWire] (5.center) to (6.center);
		\draw [qWire] (8.center) to (7.center);
		\draw (18.center) to (6.center);
		\draw (8.center) to (19.center);
	\end{pgfonlayer}
\end{tikzpicture}
\ +(1-\omega)\
\begin{tikzpicture}
	\begin{pgfonlayer}{nodelayer}
		\node [style=none] (0) at (0, 0) {$
\widetilde{T}_3$};
		\node [style=none] (1) at (-0.5, 0.5) {};
		\node [style=none] (2) at (0.5, 0.5) {};
		\node [style=none] (3) at (0.5, -0.5) {};
		\node [style=none] (4) at (-0.5, -0.5) {};
		\node [style=none] (5) at (0, 0.5) {};
		\node [style=none] (6) at (0, 1.5) {};
		\node [style=none] (7) at (0, -0.5) {};
		\node [style=none] (8) at (0, -1.5) {};
		\node [style=none] (13) at (1.25, -1.25) {\tiny $\tilde{\xi}$};
		\node [style=none] (14) at (1.5, -1.5) {};
		\node [style=none] (15) at (1.5, 1.5) {};
		\node [style=none] (16) at (-1.5, 1.5) {};
		\node [style=none] (17) at (-1.5, -1.5) {};
		\node [style=none] (18) at (0, 2) {};
		\node [style=none] (19) at (0, -2) {};
	\end{pgfonlayer}
	\begin{pgfonlayer}{edgelayer}
			\filldraw[fill=Red!20,draw=Red!40](14.center) to (15.center) to (16.center) to (17.center) to cycle;
		\filldraw[fill=white,draw=black] (1.center) to (2.center) to (3.center) to (4.center) to cycle;
		\draw [qWire] (5.center) to (6.center);
		\draw [qWire] (8.center) to (7.center);
		\draw (18.center) to (6.center);
		\draw (8.center) to (19.center);
	\end{pgfonlayer}
\end{tikzpicture}.
\eeq
\end{enumerate}
\end{definition}

In analogy with the discussion above, one has that normalized GPT states on some system are represented in an ontological model by probability distributions over the ontic state space associated with that system, while GPT effects are represented by response functions.

The state spaces in $\SubS$ form simplices, and so we will sometimes refer to an ontological model of a GPT as a simplex embedding. This terminology is a natural extension of the definition of simplex embedding in \cite{schmid2019characterization}.
\subsubsection{Quasiprobabilistic models}

We now introduce quasiprobabilistic models of a GPT. One could analogously define quasiprobabilistic models of an operational theory (as diagram-preserving maps from $\Op$ to $\QSS$).  However, given that the expressive freedom offered by the possibility of context-dependence is sufficient to ensure that every operational theory admits of an ontological model, and hence a {\em positive} quasiprobabilistic model, there is no need to make use of the additional expressive freedom offered by allowing negative quasiprobabilities, and hence no motivation to introduce such models.
On the other hand, in the case of GPTs, there does not always exist an ontological model, hence quasiprobabilistic models are a useful conceptual and mathematical tool for assessing the classicality of a GPT.

\begin{definition}[Quasiprobabilistic models of GPTs]\label{def:quasi}
A quasiprobabilistic model of a GPT $\widetilde{\Op}$, is a diagram-preserving map \colorbox{purple!20}{$\hat{\xi}: \widetilde{\Op} \to \QSS$}, depicted as
\[\hat{\xi}::\ \input{Diagrams/26b_GPT1b.tikz}\ \mapsto \ \input{Diagrams/26c_GPT_hat.tikz},\]
where the map satisfies three properties:
\begin{enumerate}
\item It represents the deterministic effect for each system  appropriately:
\begin{equation}
\input{Diagrams/22b_Ont2_hat.tikz}\ =\ \begin{tikzpicture}
	\begin{pgfonlayer}{nodelayer}
		\node [style=upground] (0) at (0, -0) {};
		\node [style=none] (7) at (0, -0.25) {};
		\node [style=none] (8) at (0, -1.5) {};
		\node [style=right label] (9) at (0, -1.5) {$\mathds{R}^{\Lambda_A}$};
	\end{pgfonlayer}
	\begin{pgfonlayer}{edgelayer}
		\draw (8.center) to (7.center);
	\end{pgfonlayer}
\end{tikzpicture}
 =\
\mathbf{1}. \label{eq:quasinorm}
\end{equation}
\item It reproduces the operational predictions of the GPT (i.e., is empirically adequate), so that for all closed diagrams,
\[\input{Diagrams/24_empadQrepn.tikz}\ = \ \input{Diagrams/25_empadQrepn.tikz}\ =\mathrm{Pr}(\widetilde{E},\widetilde{P}).\]
\item It preserves the convex and coarse-graining relations between operational procedures. E.g., if
\beq
\begin{tikzpicture}
	\begin{pgfonlayer}{nodelayer}
		\node [style=none] (0) at (0, 0) {$
\widetilde{T}_1$};
		\node [style=none] (1) at (-0.5, 0.5) {};
		\node [style=none] (2) at (0.5, 0.5) {};
		\node [style=none] (3) at (0.5, -0.5) {};
		\node [style=none] (4) at (-0.5, -0.5) {};
		\node [style=none] (5) at (0, 0.5) {};
		\node [style=none] (6) at (0, 1.5) {};
		\node [style=none] (7) at (0, -0.5) {};
		\node [style=none] (8) at (0, -1.5) {};
	\end{pgfonlayer}
	\begin{pgfonlayer}{edgelayer}
		\draw [qWire] (5.center) to (6.center);
		\draw [qWire] (8.center) to (7.center);
		\draw (1.center) to (2.center);
		\draw (2.center) to (3.center);
		\draw (3.center) to (4.center);
		\draw (4.center) to (1.center);
	\end{pgfonlayer}
\end{tikzpicture}
\ = \omega\
\begin{tikzpicture}
	\begin{pgfonlayer}{nodelayer}
		\node [style=none] (0) at (0, 0) {$
\widetilde{T}_2$};
		\node [style=none] (1) at (-0.5, 0.5) {};
		\node [style=none] (2) at (0.5, 0.5) {};
		\node [style=none] (3) at (0.5, -0.5) {};
		\node [style=none] (4) at (-0.5, -0.5) {};
		\node [style=none] (5) at (0, 0.5) {};
		\node [style=none] (6) at (0, 1.5) {};
		\node [style=none] (7) at (0, -0.5) {};
		\node [style=none] (8) at (0, -1.5) {};
	\end{pgfonlayer}
	\begin{pgfonlayer}{edgelayer}
		\draw [qWire] (5.center) to (6.center);
		\draw [qWire] (8.center) to (7.center);
		\draw (1.center) to (2.center);
		\draw (2.center) to (3.center);
		\draw (3.center) to (4.center);
		\draw (4.center) to (1.center);
	\end{pgfonlayer}
\end{tikzpicture}
\ + (1-\omega)\
 \begin{tikzpicture}
	\begin{pgfonlayer}{nodelayer}
		\node [style=none] (0) at (0, 0) {$
\widetilde{T}_3$};
		\node [style=none] (1) at (-0.5, 0.5) {};
		\node [style=none] (2) at (0.5, 0.5) {};
		\node [style=none] (3) at (0.5, -0.5) {};
		\node [style=none] (4) at (-0.5, -0.5) {};
		\node [style=none] (5) at (0, 0.5) {};
		\node [style=none] (6) at (0, 1.5) {};
		\node [style=none] (7) at (0, -0.5) {};
		\node [style=none] (8) at (0, -1.5) {};
	\end{pgfonlayer}
	\begin{pgfonlayer}{edgelayer}
		\draw [qWire] (5.center) to (6.center);
		\draw [qWire] (8.center) to (7.center);
		\draw (1.center) to (2.center);
		\draw (2.center) to (3.center);
		\draw (3.center) to (4.center);
		\draw (4.center) to (1.center);
	\end{pgfonlayer}
\end{tikzpicture}
\eeq
 then it must hold that
 \beq
\begin{tikzpicture}
	\begin{pgfonlayer}{nodelayer}
		\node [style=none] (0) at (0, 0) {$
\widetilde{T}_1$};
		\node [style=none] (1) at (-0.5, 0.5) {};
		\node [style=none] (2) at (0.5, 0.5) {};
		\node [style=none] (3) at (0.5, -0.5) {};
		\node [style=none] (4) at (-0.5, -0.5) {};
		\node [style=none] (5) at (0, 0.5) {};
		\node [style=none] (6) at (0, 1.5) {};
		\node [style=none] (7) at (0, -0.5) {};
		\node [style=none] (8) at (0, -1.5) {};
		\node [style=none] (13) at (1.25, -1.25) {\tiny $\hat{\xi}$};
		\node [style=none] (14) at (1.5, -1.5) {};
		\node [style=none] (15) at (1.5, 1.5) {};
		\node [style=none] (16) at (-1.5, 1.5) {};
		\node [style=none] (17) at (-1.5, -1.5) {};
		\node [style=none] (18) at (0, 2) {};
		\node [style=none] (19) at (0, -2) {};
	\end{pgfonlayer}
	\begin{pgfonlayer}{edgelayer}
			\filldraw[fill=purple!20,draw=purple!40](14.center) to (15.center) to (16.center) to (17.center) to cycle;
		\filldraw[fill=white,draw=black] (1.center) to (2.center) to (3.center) to (4.center) to cycle;
		\draw [qWire] (5.center) to (6.center);
		\draw [qWire] (8.center) to (7.center);
		\draw (18.center) to (6.center);
		\draw (8.center) to (19.center);
	\end{pgfonlayer}
\end{tikzpicture}
\ = \omega\
\begin{tikzpicture}
	\begin{pgfonlayer}{nodelayer}
		\node [style=none] (0) at (0, 0) {$
\widetilde{T}_2$};
		\node [style=none] (1) at (-0.5, 0.5) {};
		\node [style=none] (2) at (0.5, 0.5) {};
		\node [style=none] (3) at (0.5, -0.5) {};
		\node [style=none] (4) at (-0.5, -0.5) {};
		\node [style=none] (5) at (0, 0.5) {};
		\node [style=none] (6) at (0, 1.5) {};
		\node [style=none] (7) at (0, -0.5) {};
		\node [style=none] (8) at (0, -1.5) {};
		\node [style=none] (13) at (1.25, -1.25) {\tiny $\hat{\xi}$};
		\node [style=none] (14) at (1.5, -1.5) {};
		\node [style=none] (15) at (1.5, 1.5) {};
		\node [style=none] (16) at (-1.5, 1.5) {};
		\node [style=none] (17) at (-1.5, -1.5) {};
		\node [style=none] (18) at (0, 2) {};
		\node [style=none] (19) at (0, -2) {};
	\end{pgfonlayer}
	\begin{pgfonlayer}{edgelayer}
			\filldraw[fill=purple!20,draw=purple!40](14.center) to (15.center) to (16.center) to (17.center) to cycle;
		\filldraw[fill=white,draw=black] (1.center) to (2.center) to (3.center) to (4.center) to cycle;
		\draw [qWire] (5.center) to (6.center);
		\draw [qWire] (8.center) to (7.center);
		\draw (18.center) to (6.center);
		\draw (8.center) to (19.center);
	\end{pgfonlayer}
\end{tikzpicture}
\ +(1-\omega)\
\begin{tikzpicture}
	\begin{pgfonlayer}{nodelayer}
		\node [style=none] (0) at (0, 0) {$
\widetilde{T}_3$};
		\node [style=none] (1) at (-0.5, 0.5) {};
		\node [style=none] (2) at (0.5, 0.5) {};
		\node [style=none] (3) at (0.5, -0.5) {};
		\node [style=none] (4) at (-0.5, -0.5) {};
		\node [style=none] (5) at (0, 0.5) {};
		\node [style=none] (6) at (0, 1.5) {};
		\node [style=none] (7) at (0, -0.5) {};
		\node [style=none] (8) at (0, -1.5) {};
		\node [style=none] (13) at (1.25, -1.25) {\tiny $\hat{\xi}$};
		\node [style=none] (14) at (1.5, -1.5) {};
		\node [style=none] (15) at (1.5, 1.5) {};
		\node [style=none] (16) at (-1.5, 1.5) {};
		\node [style=none] (17) at (-1.5, -1.5) {};
		\node [style=none] (18) at (0, 2) {};
		\node [style=none] (19) at (0, -2) {};
	\end{pgfonlayer}
	\begin{pgfonlayer}{edgelayer}
			\filldraw[fill=purple!20,draw=purple!40](14.center) to (15.center) to (16.center) to (17.center) to cycle;
		\filldraw[fill=white,draw=black] (1.center) to (2.center) to (3.center) to (4.center) to cycle;
		\draw [qWire] (5.center) to (6.center);
		\draw [qWire] (8.center) to (7.center);
		\draw (18.center) to (6.center);
		\draw (8.center) to (19.center);
	\end{pgfonlayer}
\end{tikzpicture}.
\eeq
\end{enumerate}
\end{definition}

One can see that the only technical distinction between an ontological model of a GPT and a quasiprobabilistic model of a GPT is that in the latter, the probabilities are replaced by quasiprobabilities, which are allowed to go negative.

In analogy with the discussion at the end of Section~\ref{subsec:ont}, one has that GPT states on some system are represented in a quasiprobabilistic model by quasidistributions over the sample space associated with that system, that is, functions on $\Lambda$ normalised to $1$ but where the values can be negative, while GPT effects are represented by arbitrary real-valued functions over the sample space.

\section{Three equivalent notions of classicality} \label{threenotions}

The only ontological models that constitute good classical explanations are those that satisfy additional assumptions.
One such principle is that of (generalized) noncontextuality~\cite{Spekkens2005}. It was argued in Refs.~\cite{Spekkens2005,Pirsa_nc,Spekkens2008,schmid2019characterization}
that an ontological model of an operational theory should be deemed a good classical explanation only if it is noncontextual.
We now provide the definition of a noncontextual ontological model in the framework we have introduced here.

\begin{definition}[A noncontextual ontological model of an operational theory]
An ontological model of an operational theory \colorbox{black!30!BurntOrange!30}{$\xiNC:\Op\to\SubS$} satisfies the principle of {\em generalized noncontextuality} if and only if every two operationally equivalent procedures in the operational theory are mapped to the same substochastic map in the ontological model.
That is, if
\begin{equation} \begin{tikzpicture}
	\begin{pgfonlayer}{nodelayer}
		\node [style=small box] (0) at (0, -0) {$T$};
		\node [style=none] (1) at (0, 1) {};
		\node [style=none] (2) at (0, -1) {};
	\end{pgfonlayer}
	\begin{pgfonlayer}{edgelayer}
		\draw [qWire] (2.center) to (0);
		\draw [qWire] (0) to (1.center);
	\end{pgfonlayer}
\end{tikzpicture}
 \simeq \begin{tikzpicture}
	\begin{pgfonlayer}{nodelayer}
		\node [style=small box] (0) at (0, -0) {$T'$};
		\node [style=none] (1) at (0, 1) {};
		\node [style=none] (2) at (0, -1) {};
	\end{pgfonlayer}
	\begin{pgfonlayer}{edgelayer}
		\draw [qWire] (2.center) to (0);
		\draw [qWire] (0) to (1.center);
	\end{pgfonlayer}
\end{tikzpicture}
 \implies \input{Diagrams/34_NC3.tikz} = \input{Diagrams/35_NC4.tikz}. \label{eq:nonctx}
\end{equation}
\end{definition}
Another way of stating this condition is that the map $\xiNC$ does not depend functionally on the context of any processes in the operational theory, so that for all $T := (\widetilde{T},c_{T})$ one has $\xiNC(T)=\xiNC(\widetilde{T})$.

Ontological models of GPTs (as we have defined them) cannot be said to be either generalized-contextual or generalized-noncontextual (in contrast to ontological models of operational theories, which can). This is because the domain of our notion of an ontological model of a GPT has no notion of a context on which the ontological representation could conceivably depend. (This was first pointed out in Ref.~\cite{schmid2019characterization}, and we explain it further in Appendix~\ref{contextsingpts}.)
However, Ref.~\cite{schmid2019characterization} showed (in the context of prepare-and-measure scenarios) that the principle of noncontextuality nonetheless induces a notion of classicality within the framework of GPTs: namely, the GPT is said to have a classical explanation if and only if it admits of an ontological model. (Not all GPTs admit of an ontological model, even if the operational theory from which they are obtained as a quotiented theory do. This is a consequence of the representational inflexibility resulting from the lack of contexts on which the representation might depend.\footnote{
Accordingly, the Beltrametti-Bujaski model~\cite{Beltrametti_1995} can be viewed as an ontological model of the single qubit subtheory {\em qua operational theory}, but not as an ontological model of the single qubit subtheory {\em qua GPT}.  This can be seen by noting that this model is explicitly contextual while the single qubit subtheory qua GPT has no contexts.  Equivalently, it can be seen by noting that  the single qubit subtheory qua GPT does not admit of {\em any} ontological model.  The same can be said of the $8$-state model of Ref.~\cite{8state} relative to the stabilizer qubit subtheory: it is an ontological model of the stabilizer qubit subtheory {\em qua operational theory} (a {\em contextual} ontological model) but not of the stabilizer qubit subtheory qua GPT.  The latter has no contexts and does not admit of any ontological model.}) We now extend this result (Theorem~$1$ of Ref.~\cite{schmid2019characterization}) from prepare-and-measure scenarios to arbitrary scenarios.

\begin{proposition}\label{thm:NCOMandOM}
There is a one-to-one correspondence between noncontextual ontological models of an operational theory, \colorbox{black!30!BurntOrange!30}{$\xiNC:\Op \to \SubS$}, and ontological models of the associated GPT, \colorbox{Red!20}{$\widetilde{\xi}:\widetilde{\Op}\to\SubS$}.
\end{proposition}
\begin{proof}[Proof sketch]
The idea of the proof is captured by the following diagram:
\[
\begin{tikzpicture}
	\begin{pgfonlayer}{nodelayer}
		\node [style=none] (0) at (-2, 1) {$\Op$};
		\node [style=none] (1) at (2.25, 1) {$\widetilde{\Op}$};
		\node [style=none] (2) at (0, -1.25) {$\SubS$};
		\node [style=none] (3) at (1.5, 1) {};
		\node [style=none] (4) at (-1.5, 1) {};
		\node [style=none] (5) at (-1.75, 0.5) {};
		\node [style=none] (6) at (-0.5, -0.75) {};
		\node [style=none] (7) at (1.75, 0.5) {};
		\node [style=none] (8) at (0.5, -0.75) {};
		\node [style=none] (9) at (0, 1.25) {$\sim$};
		\node [style=none] (10) at (-0.75, 0.25) {$\xiNC$};
		\node [style=none] (11) at (1.5, -0.25) {$\widetilde{\xi}$};
		\node [style=none] (12) at (-1.5, 1.25) {};
		\node [style=none] (13) at (1.5, 1.25) {};
		\node [style=none] (14) at (0, 2.25) {$C$};
	\end{pgfonlayer}
	\begin{pgfonlayer}{edgelayer}
		\draw [arrow plain] (5.center) to (6.center);
		\draw [arrow plain] (4.center) to (3.center);
		\draw [arrow plain] (7.center) to (8.center);
		\draw [arrow dashed, bend left=315] (13.center) to (12.center);
	\end{pgfonlayer}
\end{tikzpicture},
\]
where $C$ is defined as a  map, which is not diagram-preserving (hence the dashed arrow), and  which takes any process $\widetilde{T}$ in the GPT $\widetilde{\Op}$ to some process $T= (\widetilde{T},c_{T})$ in the operational theory. There always exists at least one such map $C$ (in general, there exist many), and all of these satisfy ${\sim} \circ C = {\rm Id}$ (in general, no choice of $C$ will satisfy  $ C \circ {\sim} = {\rm Id}$).

Now, consider an operational theory $\Op$ and the GPT $\widetilde{\Op}$ it defines.

Given an ontological model $\widetilde{\xi}$ of $\widetilde{\Op}$, one can define a noncontextual model $\xiNC$ of $\Op$ via $\xiNC :=   \widetilde{\xi} \circ {\sim}$. The map constructed in this manner cannot depend on the contexts of processes in the operational theory, since these are removed by the quotienting map $\sim$. As such, the map $\xiNC$ necessarily satisfies Eq.~\eqref{eq:nonctx}, and hence is indeed noncontextual.

Given a noncontextual ontological model $\xiNC$ of $\Op$, one can define an ontological model $\widetilde{\xi}$ of $\widetilde{\Op}$ via $\widetilde{\xi}:=\xiNC \circ C$. Because the map $\xiNC$ does not depend on the context, the map constructed in this manner does not depend on the choice of $C$, and is unique.

For completeness, we prove in Appendix~\ref{comptheproof} that $\xiNC :=   \widetilde{\xi} \circ {\sim}$ indeed satisfies the relevant constraints to be an ontological model of an operational theory, and similarly, that $\widetilde{\xi}:=\xiNC \circ C$ satisfies the relevant constraints to be a valid ontological model of a GPT.
\end{proof}

Finally, we note that this notion of classicality of a GPT is closely linked to the positivity of quasiprobabilistic models. This result can be seen as an extension of the equivalence in Ref.~\cite{Spekkens2008} from the prepare-and-measure scenario to arbitrary compositional scenarios.

\begin{definition}[Positive quasiprobabilistic model of a tomographically local GPT] A positive quasiprobabilistic model of a tomographically local GPT $\widetilde{\Op}$ is a quasiprobabilistic model \colorbox{black!30!purple!30}{$\hat{\xi}^+:\widetilde{\Op}\to\QSS$} in which all of the matrix elements of the quasisubstochastic maps in the image of $\hat{\xi}^+$ are positive, that is, if and only if all of the quasisubstochastic maps in the image of $\hat{\xi}^+$ are substochastic.
\end{definition}

Simply by examining the definitions, it is clear that a {\em positive} quasiprobabilistic model \colorbox{black!30!purple!30}{$\hat{\xi}^+:\widetilde{\Op}\to\QSS$} of a GPT is equivalent to an ontological model \colorbox{Red!20}{$\widetilde{\xi}:\widetilde{\Op}\to\SubS$} of that GPT. It follows that:
\begin{proposition}\label{prop:PQRandOM}
There exists a positive quasiprobabilistic model of a GPT $\widetilde{\Op}$ if and only if there exists an ontological model of $\widetilde{\Op}$.
\end{proposition}

 Although Proposition~\ref{prop:PQRandOM} follows immediately from the relevant definitions, we have nonetheless highlighted it here. This is because a {\em generic} quasiprobabilistic model of a GPT has no meaningful conceptual relationship to an ontological model of a GPT, and so it is conceptually important to understand in what special cases the two notions coincide. Furthermore, we hope that highlighting this fact will encourage more dialogue between those researchers studying quasiprobabilistic models and those studying ontological models. 

Putting  Props.~\ref{thm:NCOMandOM} and~\ref{prop:PQRandOM} together, one has that:
\begin{corollary}[Three equivalent notions of classicality]\label{cor:three}
Let $\Op$ be an operational theory and $\widetilde{\Op}$ the GPT obtained from $\Op$ by quotienting. Then, the following are equivalent:\\
(i) There exists a noncontextual ontological model of $\Op$, \colorbox{black!30!BurntOrange!30}{$\xiNC:\Op\to\SubS$}. \\
(ii) There exists an ontological model (a.k.a. simplex embedding) of $\widetilde{\Op}$, \colorbox{Red!20}{$\widetilde{\xi}:\widetilde{\Op}\to\SubS$}.
(iii) There exists a positive quasiprobabilistic model of $\widetilde{\Op}$, \colorbox{black!30!purple!30}{$\hat{\xi}^+:\widetilde{\Op}\to\QSS$}.
\end{corollary}

This generalizes the results of Refs.~\cite{Spekkens2008,schmid2019characterization,shahandeh2019contextuality} from prepare-measure scenarios to arbitrary compositional scenarios.

\section{Structure theorem} \label{mainresults}

With this framework in place, we can prove our main results.
We start with a general theorem, leveraging the fact that $\widetilde{\Op} \subset \RL$, as stated in Theorem~\ref{OpinRL}.
We then specialize to the various physically relevant cases.

\begin{theorem}\label{mainthm}
Any convex-linear, empirically adequate and diagram-preserving map (Eq.~\eqref{eq:DPM}), \colorbox{green!20}{$M:\widetilde{\Op}\to \RL$}, where $\widetilde{\Op}$ is tomographically local can be represented as
\begin{equation}\label{eq:repTrans}
\input{Diagrams/50_Proof2.tikz}\quad = \quad \input{Diagrams/42_QuasiRep2.tikz},
\end{equation}
where for each system $A$, $\chi_A:A\to V_A$ is a  invertible linear map within $\RL$.  Moreover, the $\chi_A$ are uniquely determined by Eq.~\eqref{eq:repTrans}.
\label{thm:structure}
\end{theorem}

Note that we have colored the linear maps $\chi_A$ to make it immediately apparent that they came from the associated diagram-preserving map.

The proof consists of three main arguments, provided explicitly in Appendix~\ref{mainproof} and sketched here.

First, we leverage tomographic locality of the GPT, as well as convex-linearity and diagram preservation of the map, to prove that one can represent the action of the map on a generic process in terms of its action on states and effects
\begin{equation}
	\input{Diagrams/50_Proof2.tikz}\ =\ \input{Diagrams/51_Proof3.tikz}\ = \sum_{ij} r_{ij}\ \input{Diagrams/55_Proof7.tikz}.
\label{step1}
\end{equation}

Second,
using convex-linearity of the map, we prove that one can represent the action of $M$ on states and effects simply as some linear maps within $\RL$; that is,
\begin{equation}
\input{Diagrams/63_Proof15.tikz}\ =\ \input{Diagrams/64_Proof16.tikz}
\qquad\text{and}\qquad
\input{Diagrams/65_Proof17.tikz}\ =\ \input{Diagrams/66_Proof18.tikz},
\label{step2}
\end{equation}
 which relies on the isomorphism between vectors (or covectors) in $V$ and linear maps from $\mathds{R}$ to $V$ (resp. $V$ to $\mathds{R}$).  Note that $\chi_B$ and $\phi_A$ are uniquely fixed by Eq.~\eqref{step2}, which means that there can be no other choice made for the $\chi_A$ appearing in Eq.~\eqref{eq:repTrans}.

Next, we leverage empirical adequacy, that 
\begin{equation}
\input{Diagrams/73b_oneM.tikz} =\input{Diagrams/25_empadG.tikz}
\end{equation}
for all $\widetilde{P}$ and $\widetilde{E}$, together with the fact that they span the vector space and dual, to show that 
\begin{equation}
\begin{tikzpicture}
	\begin{pgfonlayer}{nodelayer}
		\node [style=none] (0) at (0, 1) {};
		\node [style=none] (1) at (0, -1) {};
		\node [style=right label] (2) at (0, -0.75) {$A$};
	\end{pgfonlayer}
	\begin{pgfonlayer}{edgelayer}
		\draw [qWire] (0.center) to (1.center);
	\end{pgfonlayer}
\end{tikzpicture}
\ =\ \ \input{Diagrams/76_Proof28.tikz};
\label{step3a}
\end{equation}
that is, that $\phi_A$ is the left inverse of $\chi_A$.

Finally we consider the representation of the identity as, 
\begin{equation}
\input{Diagrams/70_Proof22.tikz}\ =\ \begin{tikzpicture}
	\begin{pgfonlayer}{nodelayer}
		\node [style=none] (0) at (0, 1) {};
		\node [style=none] (1) at (0, -1) {};
	\end{pgfonlayer}
	\begin{pgfonlayer}{edgelayer}
		\draw (0.center) to (1.center);
	\end{pgfonlayer}
\end{tikzpicture},
\end{equation}
which is a consequence of diagram preservation,
 to prove that,
\begin{equation}
\begin{tikzpicture}
	\begin{pgfonlayer}{nodelayer}
		\node [style=none] (0) at (0, 1) {};
		\node [style=none] (1) at (0, -1) {};
		\node [style=right label] (2) at (0, -0.75) {$V_A$};
	\end{pgfonlayer}
	\begin{pgfonlayer}{edgelayer}
		\draw (0.center) to (1.center);
	\end{pgfonlayer}
\end{tikzpicture}
\ =\ \ \input{Diagrams/78_DimBound2.tikz},
\label{step3b}
\end{equation}
which means that $\phi_A$ is also the right inverse of $\chi_A$, and hence that it is unique such that we can write $\phi_A = \chi_A^{-1}$. 

This shows that the only freedom in the representation is in representation of the states, via the choice of linear maps $\chi_S$, of the theory; after specifying these, one can uniquely extend to the representation of arbitrary processes. It also shows that the representation $M$ is necessarily invertible as we can always define the inverse of $M$ by using the inverses of the $\chi$'s.

One key consequence of this result is the following corollary, whose significance we investigate in Section~\ref{subsec:conseq}.

\begin{corollary}
The dimension of the codomain, $V_A$ of the  map $\chi_A$ is given by the dimension of the GPT vector space $A$.
\label{cor:dimbound}
\end{corollary}
\proof
The linear map $\chi_A$ is invertible, so the dimension of its domain and of its codomain must be equal, and its domain is the GPT vector space.
\endproof

Note that the proof of the structure theorem and this subsequent corollary do not require the full generality of diagram preservation, only the (mathematically) much weaker conditions that:
\begin{equation}
	\input{Diagrams/70_Proof22.tikz}\ =\ \begin{tikzpicture}
	\begin{pgfonlayer}{nodelayer}
		\node [style=none] (0) at (0, 1) {};
		\node [style=none] (1) at (0, -1) {};
	\end{pgfonlayer}
	\begin{pgfonlayer}{edgelayer}
		\draw (0.center) to (1.center);
	\end{pgfonlayer}
\end{tikzpicture}
	\quad \text{,}\quad
	\input{Diagrams/54_Proof6.tikz}\ =\ \begin{tikzpicture}
	\begin{pgfonlayer}{nodelayer}
		\node [style=none] (0) at (0, 2) {};
		\node [style=none] (1) at (0, -2) {};
		\node [style=none] (2) at (0.75, -1.75) {\tiny$M$};
		\node [style=none] (3) at (0, 2.75) {};
		\node [style=none] (4) at (0, -2.75) {};
		\node [style=none] (5) at (-1, -0.25) {};
		\node [style=none] (6) at (1, -0.25) {};
		\node [style=none] (7) at (1, -2) {};
		\node [style=none] (8) at (-1, -2) {};
		\node [style=copoint] (9) at (0, -1.25) {$\widetilde{E}_j$};
		\node [style=point] (10) at (0, 1.25) {$\widetilde{P}_i$};
		\node [style=none] (11) at (0.75, 0.5) {\tiny$M$};
		\node [style=none] (12) at (-1, 2) {};
		\node [style=none] (13) at (-1, 0.25) {};
		\node [style=none] (14) at (1, 0.25) {};
		\node [style=none] (15) at (1, 2) {};
	\end{pgfonlayer}
	\begin{pgfonlayer}{edgelayer}
		\filldraw[fill=green!20,draw=green!40](5.center) to (6.center) to (7.center) to (8.center) to cycle;
		\filldraw[fill=green!20,draw=green!40](12.center) to (13.center) to (14.center) to (15.center) to cycle;
		\draw (3.center) to (0.center);
		\draw (1.center) to (4.center);
		\draw [style=qWire] (9) to (1.center);
		\draw [style=qWire] (10) to (0.center);
	\end{pgfonlayer}
\end{tikzpicture}
,
    \quad \text{and}\quad
    \input{Diagrams/73b_oneM.tikz}\ =\ \input{Diagrams/73_Proof25.tikz}.
\end{equation}
We will give justifications of these (for the case of ontological models and quasiprobabilistic models) in Sec.~\ref{revisassump}, and will
discuss further
consequences of general diagram preservation in Sec.~\ref{onticseparability}.

\subsection{Diagram-preserving quasiprobabilistic models are exact frame representations}

As mentioned in the introduction, $\SubS$ and $\QSS$ are subprocess theories of $\RL$, $\SubS\subset \QSS \subset \RL$. This implies that our main theorem applies to these special cases. The fact that the codomain is restricted can then equivalently be expressed as a constraint on the linear maps $\chi_A$. In the case of quasiprobabilistic representations we obtain:
\begin{proposition} \label{qrepngptstruct}
Any diagram-preserving quasiprobabilistic model of a tomographically local GPT can be written as
\begin{equation}
  \input{Diagrams/41_QuasiRep1.tikz}\quad = \quad \input{Diagrams/42b_QuasiRep2hat.tikz}\label{eq:quasichi}
\end{equation}
for invertible linear maps $\{\chi_S:S\to \mathds{R}^{\Lambda_S}\}$ within $\RL$ for each system, where these satisfy
\begin{equation}\label{eq:quasidiscard}
\input{Diagrams/43b_QuasiRep3hat.tikz}\quad =\quad  \begin{tikzpicture}
	\begin{pgfonlayer}{nodelayer}
		\node [style=none] (0) at (0, -0.75) {};
		\node [style=none] (1) at (0, 0.25) {};
		\node [style=upground] (2) at (0, 0.5) {};
	\end{pgfonlayer}
	\begin{pgfonlayer}{edgelayer}
		\draw [qWire] (0.center) to (1.center);
	\end{pgfonlayer}
\end{tikzpicture}.
\end{equation}
\end{proposition}
\proof
Since $\hat{\xi}$ satisfies the requirements of Theorem~\ref{thm:structure} we immediately obtain Eq.~\eqref{eq:quasichi}. For the particular case of the deterministic effect, we have that
\begin{equation} \input{Diagrams/22c_Ont2_hat_nolbl.tikz}\quad =\quad \input{Diagrams/chi_inv_ground.tikz}. \end{equation}
Recall that, by definition, a quasiprobabilistic model satisfies Eq.~\eqref{eq:quasinorm}:
\begin{equation} \input{Diagrams/22c_Ont2_hat_nolbl.tikz}\quad =\quad \begin{tikzpicture}
	\begin{pgfonlayer}{nodelayer}
		\node [style=upground] (0) at (0, 0.75) {};
		\node [style=none] (7) at (0, 0.5) {};
		\node [style=none] (8) at (0, -0.75) {};
	\end{pgfonlayer}
	\begin{pgfonlayer}{edgelayer}
		\draw (8.center) to (7.center);
	\end{pgfonlayer}
\end{tikzpicture}
. \end{equation}
Combining these gives
\begin{equation}\label{eq:simplexEmbed1} \begin{tikzpicture}
	\begin{pgfonlayer}{nodelayer}
		\node [style=upground] (0) at (0, 0.75) {};
		\node [style=none] (7) at (0, 0.5) {};
		\node [style=none] (8) at (0, -0.75) {};
	\end{pgfonlayer}
	\begin{pgfonlayer}{edgelayer}
		\draw (8.center) to (7.center);
	\end{pgfonlayer}
\end{tikzpicture}
\quad =\quad \input{Diagrams/chi_inv_ground.tikz}.\end{equation}
Composing both sides of this with $\chi_S$ gives Eq.~\eqref{eq:quasidiscard}.
\endproof

The extra constraint of Eq.~\eqref{eq:quasidiscard} is not part of the general structure theorem because an abstract vector space does not have a natural notion of discarding. Such a privileged notion is found within, for example, $\mathbf{SubStoch}$, as the all ones covector, which represents marginalization.\footnote{Even within vector spaces with an associated physical interpretation of vectors as processes, there is not always a unique discarding map; for example, in the vector space of quantum channels, applying the channel to {\em any} fixed input state and tracing the output constitutes a discarding operation.}

Since the $\chi_S$ are just invertible linear maps, this map can be seen as merely transforming from one representation of the GPT to another. Critically, however, one must note that the vector spaces in $\QSS$ are all of the form $\mathds{R}^\Lambda$, and so they come equipped with extra structure---namely, a preferred basis and dual basis. Hence, these representations are effectively singling out this preferred basis for the GPT.

To see this more explicitly, denote the preferred basis and cobasis for $\mathds{R}^\Lambda$ as 
\beq
\left\{\begin{tikzpicture}
	\begin{pgfonlayer}{nodelayer}
		\node [style=point] (0) at (0, -.25) {$\lambda$};
		\node [style=none] (1) at (0, .75) {};
		\node [style=right label] (2) at (0, 0.5) {$\mathds{R}^\Lambda$};
	\end{pgfonlayer}
	\begin{pgfonlayer}{edgelayer}
		\draw (0) to (1.center);
	\end{pgfonlayer}
\end{tikzpicture}\right\}_{\lambda\in\Lambda}
\quad\text{and}\quad
\left\{\begin{tikzpicture}
	\begin{pgfonlayer}{nodelayer}
		\node [style=copoint] (0) at (0, .25) {$\lambda$};
		\node [style=none] (1) at (0, -.75) {};
		\node [style=right label] (2) at (0, -0.75) {$\mathds{R}^\Lambda$};
		\node [style=none] (3) at (0, -0.25) {};
	\end{pgfonlayer}
	\begin{pgfonlayer}{edgelayer}
		\draw (0) to (1.center);
	\end{pgfonlayer}
\end{tikzpicture}\right\}_{\lambda\in\Lambda}\quad\text{respectively,}
\eeq
which, in particular, means we can decompose identities as
\begin{equation}
	\begin{tikzpicture}
	\begin{pgfonlayer}{nodelayer}
		\node [style=none] (0) at (0, 1) {};
		\node [style=none] (1) at (0, -1) {};
	\end{pgfonlayer}
	\begin{pgfonlayer}{edgelayer}
		\draw (0.center) to (1.center);
	\end{pgfonlayer}
\end{tikzpicture}\  = \sum_{\lambda \in \Lambda} \begin{tikzpicture}
	\begin{pgfonlayer}{nodelayer}
		\node [style=none] (0) at (0, -1.5) {};
		\node [style=copoint] (1) at (0, -0.7) {$\lambda$};
		\node [style=point] (2) at (0, 0.7) {$\lambda$};
		\node [style=none] (3) at (0, 1.5) {};
	\end{pgfonlayer}
	\begin{pgfonlayer}{edgelayer}
		\draw (0.center) to (1);
		\draw (3.center) to (2);
	\end{pgfonlayer}
\end{tikzpicture}
.
\end{equation}
This means that, for any system $A$ in the GPT, we can write $\chi_A$ as:
\beq\label{eq:frameDecomp}
\begin{tikzpicture}
	\begin{pgfonlayer}{nodelayer}
		\node [style=none] (0) at (0, -1) {};
		\node [style=small box, fill={purple!20}] (1) at (0, -0) {$\chi_A$};
		\node [style=none] (2) at (0, 1) {};
	\end{pgfonlayer}
	\begin{pgfonlayer}{edgelayer}
		\draw [qWire] (0.center) to (1);
		\draw (1) to (2.center);
	\end{pgfonlayer}
\end{tikzpicture}\ =
\sum_{\lambda\in\Lambda_A}
\begin{tikzpicture}
	\begin{pgfonlayer}{nodelayer}
		\node [style=none] (0) at (0, -1) {};
		\node [style=small box, fill={purple!20}] (1) at (0, 0) {$\chi_A$};
		\node [style=copoint] (2) at (0, 1.25) {$\lambda$};
		\node [style=none] (4) at (0, 3.75) {};
		\node [style=point] (5) at (0, 2.75) {$\lambda$};
	\end{pgfonlayer}
	\begin{pgfonlayer}{edgelayer}
		\draw [qWire] (0.center) to (1);
		\draw (1) to (2);
		\draw (4.center) to (5);
	\end{pgfonlayer}
\end{tikzpicture}
\ =: \sum_{\lambda\in\Lambda_A} \input{Diagrams/62_Proof14.tikz},\eeq
As $\chi_A$ is invertible, then $\{\widetilde{D}_\lambda^A\}_{\lambda\in\Lambda_A}$ is a cobasis for the GPT:
\beq
\left\{\begin{tikzpicture}
	\begin{pgfonlayer}{nodelayer}
		\node [style=none] (0) at (0, -0.75) {};
		\node [style=copoint, fill={purple!20}] (1) at (0, 0.25) {$\widetilde{D}_\lambda^A$};
	\end{pgfonlayer}
	\begin{pgfonlayer}{edgelayer}
		\draw [qWire] (0.center) to (1);
	\end{pgfonlayer}
\end{tikzpicture}\right\}_{\lambda\in\Lambda_A},
\eeq
whereby condition \eqref{eq:quasidiscard} becomes
\beq
\sum_{\lambda \in \Lambda_A} \begin{tikzpicture}
	\begin{pgfonlayer}{nodelayer}
		\node [style=none] (0) at (0, -0.75) {};
		\node [style=copoint, fill={purple!20}] (1) at (0, 0.25) {$\widetilde{D}_\lambda^A$};
	\end{pgfonlayer}
	\begin{pgfonlayer}{edgelayer}
		\draw [qWire] (0.center) to (1);
	\end{pgfonlayer}
    \end{tikzpicture}\ =\ \begin{tikzpicture}
	\begin{pgfonlayer}{nodelayer}
		\node [style=none] (0) at (0, -0.75) {};
		\node [style=none] (1) at (0, 0.25) {};
		\node [style=upground] (2) at (0, 0.5) {};
	\end{pgfonlayer}
	\begin{pgfonlayer}{edgelayer}
		\draw [qWire] (0.center) to (1.center);
	\end{pgfonlayer}
\end{tikzpicture}.\label{eq:cobasisdiscard}
\eeq

Similarly, we could run the same argument using $\chi^{-1}_A$ to single out a basis, $\{\widetilde{F}_\lambda^A\}_{\lambda\in\Lambda_A}$ for the GPT:
\beq 
\left\{
\begin{tikzpicture}
	\begin{pgfonlayer}{nodelayer}
		\node [style=none] (0) at (0, 0.75) {};
		\node [style=point, fill={purple!20}] (1) at (0, -0.25) {$\widetilde{F}_\lambda^A$};
	\end{pgfonlayer}
	\begin{pgfonlayer}{edgelayer}
		\draw [qWire] (0.center) to (1);
	\end{pgfonlayer}
\end{tikzpicture}
\right\}_{\lambda\in\Lambda_A},
\eeq
where (by Eq.~\eqref{eq:simplexEmbed1}) $\forall \lambda$,
\beq \label{eq:cobasisdiscard2}
\begin{tikzpicture}
	\begin{pgfonlayer}{nodelayer}
		\node [style=none] (0) at (0, 0.75) {};
		\node [style=point, fill={purple!20}] (1) at (0, -0.25) {$\widetilde{F}_\lambda^A$};
		\node [style=upground] (2) at (0, 1) {};
	\end{pgfonlayer}
	\begin{pgfonlayer}{edgelayer}
		\draw [qWire] (0.center) to (1);
	\end{pgfonlayer}
\end{tikzpicture} = 1.
\eeq
Moreover, this decomposition, together with reversibility of $\chi_A$ means that
\begin{align}
\sum_{\lambda,\lambda'\in\Lambda_A}
\begin{tikzpicture}
	\begin{pgfonlayer}{nodelayer}
		\node [style=none] (0) at (0, 0) {};
		\node [style=copoint, fill={purple!20}] (1) at (0, .75) {$\widetilde{D}_{\lambda'}^A$};
		\node [style=point, fill={purple!20}] (2) at (0, -.75) {$\widetilde{F}_\lambda^A$};
		\node [style=point] (3) at (2, 0.75) {$\lambda'$};
		\node [style=copoint] (4) at (2, -0.75) {$\lambda$};
		\node [style=none] (5) at (2, -2) {};
		\node [style=none] (6) at (2, 2) {};
	\end{pgfonlayer}
	\begin{pgfonlayer}{edgelayer}
		\draw [qWire] (2) to (0.center);
		\draw [qWire] (0.center) to (1);
		\draw (3) to (6.center);
		\draw (4) to (5.center);
	\end{pgfonlayer}
\end{tikzpicture}
\quad &= \quad
\begin{tikzpicture}
	\begin{pgfonlayer}{nodelayer}
		\node [style=small box, fill={purple!20}] (0) at (0, -1) {$\chi^{-1}_A$};
		\node [style=small box, fill={purple!20}] (1) at (0, 1) {$\chi_A$};
		\node [style=none] (2) at (0, 2) {};
		\node [style=none] (3) at (0, -2) {};
	\end{pgfonlayer}
	\begin{pgfonlayer}{edgelayer}
		\draw [qWire] (1) to (0);
		\draw (3.center) to (0);
		\draw (1) to (2.center);
	\end{pgfonlayer}
\end{tikzpicture} \nonumber \\
&= \quad
\begin{tikzpicture}
	\begin{pgfonlayer}{nodelayer}
		\node [style=none] (0) at (0, 1) {};
		\node [style=none] (1) at (0, -1) {};
		\node [style=none] (7) at (0, -1.5) {};
		\node [style=none] (8) at (0, 1.5) {};
	\end{pgfonlayer}
	\begin{pgfonlayer}{edgelayer}
		\draw (0.center) to (1.center);
	\end{pgfonlayer}
\end{tikzpicture} \nonumber \\
&=
\sum_{\lambda,\lambda'\in\Lambda_A}\delta_{\lambda\lambda'}
\begin{tikzpicture}
	\begin{pgfonlayer}{nodelayer}
		\node [style=point] (3) at (0, 0.75) {$\lambda'$};
		\node [style=copoint] (4) at (0, -0.75) {$\lambda$};
		\node [style=none] (5) at (0, -2) {};
		\node [style=none] (6) at (0, 2) {};
	\end{pgfonlayer}
	\begin{pgfonlayer}{edgelayer}
		\draw (3) to (6.center);
		\draw (4) to (5.center);
	\end{pgfonlayer}
\end{tikzpicture},
\end{align}
and hence
\beq \label{revimplies}
\begin{tikzpicture}
	\begin{pgfonlayer}{nodelayer}
		\node [style=none] (0) at (0, 0.25) {};
		\node [style=point, fill={purple!20}] (1) at (0, -0.75) {$\widetilde{F}_\lambda^A$};
		\node [style=none] (2) at (0, 0) {};
		\node [style=copoint, fill={purple!20}] (3) at (0, 1) {$\widetilde{D}_{\lambda'}^A$};
	\end{pgfonlayer}
	\begin{pgfonlayer}{edgelayer}
		\draw [qWire] (0.center) to (1);
		\draw [qWire] (2.center) to (3);
	\end{pgfonlayer}
\end{tikzpicture}\ =\ \delta_{\lambda\lambda'}.
\eeq
That is, $\{\widetilde{F}_\lambda^A\}_\lambda$ defines a basis of $V_A$ and $\{\widetilde{D}_{\lambda'}^A\}_{\lambda'}$ defines a dual basis of $V_A^*$.

We can then represent the action of $\hat{\xi}$ as:
\beq
\begin{tikzpicture}
	\begin{pgfonlayer}{nodelayer}
		\node [style=small box] (0) at (0, 0) {$\widetilde{T}$};
		\node [style=none] (1) at (0, 1.5) {};
		\node [style=none] (2) at (0, -1.5) {};
		\node [style=none] (3) at (0.75, -1.25) {\tiny$\hat{\xi}$};
		\node [style=none] (4) at (0, 2.5) {};
		\node [style=none] (5) at (0, -2.5) {};
		\node [style=none] (6) at (-1, 1.5) {};
		\node [style=none] (7) at (1, 1.5) {};
		\node [style=none] (8) at (1, -1.5) {};
		\node [style=none] (9) at (-1, -1.5) {};
		\node [style=right label] (10) at (0, -2.25) {$\mathds{R}^{\Lambda_A}$};
		\node [style=right label] (11) at (0, -1) {$A$};
		\node [style=right label] (12) at (0, 0.75) {$B$};
		\node [style=right label] (13) at (0, 2) {$\mathds{R}^{\Lambda_B}$};
	\end{pgfonlayer}
	\begin{pgfonlayer}{edgelayer}
			\filldraw[fill=purple!20,draw=purple!40](9.center) to (6.center) to (7.center) to (8.center) to cycle;
		\draw [style=qWire] (2.center) to (0);
		\draw [style=qWire] (0) to (1.center);
		\draw (4.center) to (1.center);
		\draw (2.center) to (5.center);
	\end{pgfonlayer}
\end{tikzpicture}
= \sum_{\lambda\in \Lambda_A,\lambda'\in \Lambda_B}
\begin{tikzpicture}
	\begin{pgfonlayer}{nodelayer}
		\node [style=small box] (0) at (0, 0) {$\widetilde{T}$};
		\node [style=copoint, fill={purple!20}] (1) at (0, 1.25) {$\widetilde{D}_{\lambda'}^B$};
		\node [style=point, fill={purple!20}] (2) at (0, -1.25) {$\widetilde{F}_\lambda^A$};
		\node [style=point] (3) at (2, 0.75) {$\lambda'$};
		\node [style=copoint] (4) at (2, -0.75) {$\lambda$};
		\node [style=none] (5) at (2, -2.75) {};
		\node [style=none] (6) at (2, 2.75) {};
		\node [style=right label] (7) at (2, -1.75) {$\mathds{R}^{\Lambda_A}$};
		\node [style=right label] (8) at (2, 1.75) {$\mathds{R}^{\Lambda_B}$};
	\end{pgfonlayer}
	\begin{pgfonlayer}{edgelayer}
		\draw [qWire] (2) to (0);
		\draw [qWire] (0) to (1);
		\draw (3) to (6.center);
		\draw (4) to (5.center);
	\end{pgfonlayer}
\end{tikzpicture}
\eeq
which can be viewed as a quasistochastic map defined by the conditional quasiprobability distribution
\beq
\hat{\xi}(\widetilde{T})({\lambda'}|\lambda)= \begin{tikzpicture}
	\begin{pgfonlayer}{nodelayer}
		\node [style=small box] (0) at (0, 0) {$\widetilde{T}$};
		\node [style=copoint, fill={purple!20}] (1) at (0, 1.25) {$\widetilde{D}_{\lambda'}^B$};
		\node [style=point, fill={purple!20}] (2) at (0, -1.25) {$\widetilde{F}_\lambda^A$};
	\end{pgfonlayer}
	\begin{pgfonlayer}{edgelayer}
		\draw [qWire] (2) to (0);
		\draw [qWire] (0) to (1);
	\end{pgfonlayer}
\end{tikzpicture}.
\eeq

 Finally, we note that Proposition~\ref{qrepngptstruct} also implies that any quasiprobability representation constructed using an overcomplete frame necessarily fail to be diagram-preserving. 

\subsubsection{Diagram-preserving quasiprobabilistic models of quantum theory }
We now consider the case of quantum theory as a GPT. The basis $\{\widetilde{F}_{\lambda}\}_{\lambda\in\Lambda}$ is a basis for the real vector space of Hermitian operators for the system while the cobasis, $\{\widetilde{D}_{\lambda}\}_{\lambda\in\Lambda}$ is a basis for the space of linear functionals on the vector space of Hermitian operators.  The Riesz representation theorem~\cite{riesz1914demonstration} guarantees that every element $\widetilde{D}_{\lambda}$ of the cobasis can be represented via the Hilbert-Schmidt inner product with some Hermitian operator, which we will denote as $\widetilde{D}^*_{\lambda}$, such that for all $\rho$:
\beq
\begin{tikzpicture}
	\begin{pgfonlayer}{nodelayer}
		\node [style=none] (0) at (0, 0.25) {};
		\node [style=point, none] (1) at (0, -0.75) {$\rho$};
		\node [style=none] (2) at (0, 0) {};
		\node [style=copoint, fill={purple!20}] (3) at (0, 1) {$\widetilde{D}_{\lambda'}$};
	\end{pgfonlayer}
	\begin{pgfonlayer}{edgelayer}
		\draw [qWire] (0.center) to (1);
		\draw [qWire] (2.center) to (3);
	\end{pgfonlayer}
\end{tikzpicture}
\ = \
\widetilde{D}_{\lambda}\circ \rho = \mathsf{tr}(\widetilde{D}^*_{\lambda} \rho).
\eeq
The condition in Eq.~\eqref{eq:cobasisdiscard} becomes $\sum_\lambda \widetilde{D}^*_\lambda = \mathds{1}$, the condition in Eq.~\eqref{eq:cobasisdiscard2} becomes $\mathsf{tr}(\widetilde{F}_\lambda)=1$, and the condition of Eq.~\eqref{revimplies} becomes
\beq\label{eq:dualFrame}
\mathsf{tr}(\widetilde{D}^*_{\lambda'} \widetilde{F}_\lambda)=\delta_{\lambda\lambda'}.
\eeq
It is clear, therefore, that $\{ \widetilde{F}_\lambda \}_\lambda$ and $\{ \widetilde{D}^*_\lambda \}_\lambda$ constitute a  minimal frame and its dual (in the language of, for example, Refs.~\cite{ferrie2008frame}). Hence, this representation is nothing but an exact frame representation, that is, one which is not overcomplete.
That is, a transformation $\widetilde{T}$, represented by a completely positive trace preserving map $\mathcal{E}_{\widetilde{T}}$, will be represented as a quasistochastic map defined by the conditional quasiprobability distribution:
\beq
\hat{\xi}(\widetilde{T})({\lambda'}|\lambda) = \mathsf{tr}(\widetilde{D}^*_{\lambda'}\mathcal{E}_{\widetilde{T}}(\widetilde{F}_\lambda))
\eeq
It is easy to see that any set of spanning and linearly independent vectors summing to identity will define a suitable dual frame $\{\widetilde{D}^*_\lambda\}$, and then the frame $\{\widetilde{F}_\lambda\}$ itself is uniquely defined by Eq.~\eqref{eq:dualFrame}. (Note in particular that the elements of the frame need not be pairwise orthogonal, nor must those of the dual frame.)

It has previously been shown that all quasiprobabilistic models of quantum theory are frame representations~\cite{ferrie2008frame}. What we learn here is that {\em diagram-preserving} quasiprobabilistic models are necessarily the simplest possible frame representations, namely those that are not overcomplete.

\subsection{Structure theorems for ontological models}

In the case of ontological models of a GPT, we obtain:
\begin{proposition}
Any diagram-preserving ontological model of a tomographically local GPT can be written as
\begin{equation}
 \begin{tikzpicture}
	\begin{pgfonlayer}{nodelayer}
		\node [style=small box] (0) at (0, 0) {$\widetilde{T}$};
		\node [style=none] (1) at (0, 1.5) {};
		\node [style=none] (2) at (0, -1.5) {};
		\node [style=none] (3) at (0.75, -1.25) {\tiny$\tilde{\xi}$};
		\node [style=none] (4) at (0, 2.5) {};
		\node [style=none] (5) at (0, -2.5) {};
		\node [style=none] (6) at (-1, 1.5) {};
		\node [style=none] (7) at (1, 1.5) {};
		\node [style=none] (8) at (1, -1.5) {};
		\node [style=none] (9) at (-1, -1.5) {};
		\node [style=right label] (10) at (0, -2.25) {$\mathds{R}^{\Lambda_A}$};
		\node [style=right label] (11) at (0, -1) {$A$};
		\node [style=right label] (12) at (0, 0.75) {$B$};
		\node [style=right label] (13) at (0, 2) {$\mathds{R}^{\Lambda_B}$};
	\end{pgfonlayer}
	\begin{pgfonlayer}{edgelayer}
			\filldraw[fill=Red!20,draw=Red!40](9.center) to (6.center) to (7.center) to (8.center) to cycle;
		\draw [style=qWire] (2.center) to (0);
		\draw [style=qWire] (0) to (1.center);
		\draw (4.center) to (1.center);
		\draw (2.center) to (5.center);
	\end{pgfonlayer}
\end{tikzpicture}
\quad = \quad \begin{tikzpicture}
	\begin{pgfonlayer}{nodelayer}
		\node [style=small box] (0) at (0, 0) {$\widetilde{T}$};
		\node [style=small box, fill={Red!20}] (1) at (0, 1.75) {$\chi_B$};
		\node [style=small box, fill={Red!20}] (2) at (0, -1.75) {$\chi_A^{-1}$};
		\node [style=none] (3) at (0, 3.25) {};
		\node [style=none] (4) at (0, -3.25) {};
		\node [style=right label] (5) at (0, -3) {$\mathds{R}^{\Lambda_A}$};
		\node [style=right label] (6) at (0, -1) {$A$};
		\node [style=right label] (7) at (0, 0.75) {$B$};
		\node [style=right label] (8) at (0, 2.75) {$\mathds{R}^{\Lambda_B}$};
	\end{pgfonlayer}
	\begin{pgfonlayer}{edgelayer}
		\draw [qWire] (2) to (0);
		\draw [qWire] (0) to (1);
		\draw (4.center) to (2);
		\draw (1) to (3.center);
	\end{pgfonlayer}
\end{tikzpicture}
,
  \label{eq:ontchi}
\end{equation}
where
\begin{equation}\label{eq:ontchidiscard}
\begin{tikzpicture}
	\begin{pgfonlayer}{nodelayer}
		\node [style=none] (0) at (-0.5, -1) {};
		\node [style=small box,fill={Red!20}] (1) at (-0.5, -0) {$\chi_A$};
		\node [style=none] (2) at (-0.5, 1) {};
		\node [style=upground] (3) at (-0.5, 1.25) {};
	\end{pgfonlayer}
	\begin{pgfonlayer}{edgelayer}
		\draw [qWire] (0.center) to (1.center);
		\draw  (1) to (2.center);
	\end{pgfonlayer}
\end{tikzpicture}
\ =\  \begin{tikzpicture}
	\begin{pgfonlayer}{nodelayer}
		\node [style=none] (0) at (0, -0.75) {};
		\node [style=none] (1) at (0, 0.25) {};
		\node [style=upground] (2) at (0, 0.5) {};
	\end{pgfonlayer}
	\begin{pgfonlayer}{edgelayer}
		\draw [qWire] (0.center) to (1.center);
	\end{pgfonlayer}
\end{tikzpicture}
\end{equation}
and where moreover every pair $(\chi_A^{-1},\chi_B)$ defines a {\em positive} map from the cone of transformations from $A$  to $B$ in the GPT $\widetilde{\Op}$ to the cone of substochastic maps from $\Lambda_A$ to $\Lambda_B$ in $\SubS$.
\label{prop:structureOM}
\end{proposition}

The proof is given in Appendix~\ref{propstructureOMproof}. Apart from positivity, the proof follows immediately from Proposition~\ref{qrepngptstruct}.

One can interpret this map from the GPT to the ontological model as an explicit embedding into a simplicial GPT as discussed in \cite{schmid2019characterization}, but generalized to the case in which both the GPT under consideration and the simplicial GPT have arbitrary processes, not just states and effects. This follows from the positivity conditions, empirical adequacy, and the preservation of the deterministic effect.

As in the case of quasiprobabilistic models, we can write out $\chi_A$ as:
\begin{equation}
\begin{tikzpicture}
	\begin{pgfonlayer}{nodelayer}
		\node [style=none] (0) at (0, -1) {};
		\node [style=small box, fill={Red!20}] (1) at (0, -0) {$\chi_A$};
		\node [style=none] (2) at (0, 1) {};
	\end{pgfonlayer}
	\begin{pgfonlayer}{edgelayer}
		\draw [qWire] (0.center) to (1);
		\draw (1) to (2.center);
	\end{pgfonlayer}
\end{tikzpicture}
\ = \sum_{\lambda\in\Lambda_A} \begin{tikzpicture}
	\begin{pgfonlayer}{nodelayer}
		\node [style=none] (0) at (0, -1.75) {};
		\node [style=copoint,fill=Red!20] (1) at (0, -.875) {$\widetilde{D}_\lambda^A$};
		\node [style=point] (2) at (0, 1) {$\lambda$};
		\node [style=none] (3) at (0, 1.75) {};
	\end{pgfonlayer}
	\begin{pgfonlayer}{edgelayer}
		\draw [qWire] (0.center) to (1);
		\draw (3.center) to (2);
	\end{pgfonlayer}
\end{tikzpicture}
,
\label{eq:chidecomp}
\end{equation}
where $\{ \widetilde{D}_\lambda^A\}_\lambda$ forms a basis for the vector space of the GPT defined by the operational theory. The positivity condition for $\chi_A$ discussed above immediately implies that $\widetilde{D}_{\lambda}^A$ is a linear functional which is positive on the state cone, and the normalization condition immediately implies that their sum over $\lambda$ is the deterministic effect. By a similar argument, each $\widetilde{F}_\lambda^A$ is a vector in the vector space of states which is positive on all GPT effects.

In the case of a GPT which satisfies the no-restriction hypothesis \cite{chiribella2010probabilistic} (e.g., quantum theory and the classical probabilistic theory),  this means that the $\widetilde{D}_\lambda^A$  are effects (forming a measurement) and that the $\widetilde{F}_\lambda^A$ are states.
In the quantum case, the notion of positivity that we have expressed here
 reduces to  positivity of the eigenvalues of the Hermitian operators. This provides another immediate proof that quantum theory, as a GPT, does not admit an ontological model---it would require an exact frame and dual frame for the space of Hermitian operators which are all positive, but it is known that such a basis and dual do not exist \cite{ferrie2008frame}.

We have shown (Prop~\ref{thm:NCOMandOM}) that every noncontextual ontological model of an operational theory is equivalent to an ontological model of the GPT defined by the operational theory.
Combining this with proposition~\ref{prop:structureOM}, it immediately follows that:
\begin{corollary}
 For operational theories whose corresponding GPT satisfies tomographic locality, any noncontextual ontological model can be written as
\[\input{Diagrams/45_NCRep1.tikz}\quad =\quad\input{Diagrams/45b_NCRep1b.tikz} \quad =\quad \input{Diagrams/46_NCRep2.tikz}\]
for
\[\input{Diagrams/61_Proof13.tikz}\ = \sum_{\lambda\in\Lambda_A} \begin{tikzpicture}
	\begin{pgfonlayer}{nodelayer}
		\node [style=none] (0) at (0, -2) {};
		\node [style=copoint,fill=Red!20] (1) at (0, -0.875) {$\widetilde{D}_\lambda^A$};
		\node [style=point] (2) at (0, 1) {$\lambda$};
		\node [style=none] (3) at (0, 2) {};
        \node [style=right label] (16) at (0, 1.75) {$\mathds{R}^{\Lambda_A}$};
        \node [style=right label] (18) at (0, -1.75) {$A$};
	\end{pgfonlayer}
	\begin{pgfonlayer}{edgelayer}
		\draw [qWire] (0.center) to (1);
		\draw (3.center) to (2);
	\end{pgfonlayer}
\end{tikzpicture},\]
where $\{ \widetilde{D}_\lambda^A\}_\lambda$ forms a basis for the vector space of the GPT defined by the operational theory.
\end{corollary}

Previously, the notion of a noncontextual ontological representation seemed to be a highly flexible concept, but this corollary demonstrates that it in fact has a very rigid structure. Every noncontextual ontological model can be constructed in two steps: i) quotient to the associated GPT, ii) pick a basis for the GPT such that it is manifestly an ontological model. Furthermore, the only freedom in the representation is in representation of the {\em states} in the theory (via this choice of basis); after specifying this, one can uniquely extend to the representation of arbitrary processes.

\subsection{Consequences of the dimension bound}\label{subsec:conseq}

We can specialize Corollary~\ref{cor:dimbound} to the case of quasiprobabilistic representations or ontological representations of GPTs. 
In this case, it states that the dimension of the GPT vector space for a given system $A$, ${\rm dim}(A)$,
 is equal to the dimension of the codomain of the map $\chi_A$ defining the quasiprobabilistic or ontological representation of $A$, that is, the dimension of $\mathds{R}^{\Lambda_A}$.
In the case of an ontological model, this dimension is simply $|\Lambda_A|$,
  the cardinality of the set $\Lambda_A$ of ontic states of $A$, so the number of ontic states is equal to the dimension of the GPT space. Moreover, by considering Proposition~\ref{thm:NCOMandOM}, this immediately implies that for any operational theory whose corresponding GPT satisfies tomographic locality, if there exists a {\em noncontextual} ontological model thereof, then it must 
  also have a number of ontic states equal to the dimension of the GPT state space.
    In the language of Hardy \cite{Hardy2004}, this exactly means, for each system $A$, that the ``ontological excess baggage factor'', 
\beq
\gamma_A := \frac{|\Lambda_A|}{{\rm dim}(A)},
\eeq
must be exactly $1$. In other words,  {\em demanding  noncontextuality rules out ontological excess baggage}. Since Hardy showed that all ontological models of a qubit must in fact have unbounded excess baggage, his result can immediately be combined with ours to give a new proof that the full statistics of processes on a qubit do not admit a noncontextual model.

In particular, our result implies that a diagram-preserving noncontextual ontological model of a {\em qubit} must have exactly $4$ ontic states. This result extends to 
any subtheory of a qubit whose corresponding GPT is tomographically local, e.g. the stabilizer subtheory. Hence, it constitutes a stringent constraint on ontological models of the qubit stabilizer subtheory qua operational theory. 
For instance, it immediately guarantees that the $8$-state model of Ref.~\cite{8state}---which, as the name suggests, has 8 ontic states---
must be contextual.   Indeed, the $8$-state model was previously shown to be contextual by a different argument which focused on the representation of transformation procedures in prepare-transform-measure scenarios~\cite{lillystone2019single}.

Furthermore, our bound improves an algorithm first proposed in Ref.~\cite{schmid2019characterization}.
In particular,  Ref.~\cite{schmid2019characterization} gave an algorithm for determining if a GPT admits of an ontological model by testing whether or not the GPT embeds in a simplicial GPT of arbitrary dimension. 
The lack of a bound on this dimension means that there is no guarantee that the algorithm will ever terminate.
Ref.~\cite{gitton2020solvable} solves this problem by providing such a bound, namely, the square of the given GPT's dimension. Our result strengthens this bound, reducing it to the given GPT's dimension. In fact, our bound is tight, as there can never be an embedding of the GPT into a lower dimensional space.
These results simplify the algorithm dramatically: rather than testing for embedding in a sequence of simplicial GPTs of increasing dimension, one can simply perform a single test for embedding in a simplicial GPT of the same dimension as the given GPT. 

Yet another application of the dimension bound follows from the results of Ref.~\cite{karanjai2018}. Ref.~\cite{karanjai2018} demonstrates that the number of classical bits required to specify the ontic state in any (necessarily contextual) ontological model of the  qubit stabilizer subtheory is quadratic in the number of qubits. This is contrasted to the case of the {\em qutrit} stabilizer subtheory, wherein there exists a (noncontextual) ontological model  with {\em linear} scaling
 in the number of qutrits.  The quadratic scaling result for the stabilizer qubit subtheory implies that, for a collection of qubits, the number of ontic states is necessarily greater than the dimension of the space of quantum density operators. Together with our dimension bound, this fact  is sufficient to deduce the contextuality of the qubit stabilizer subtheory.  Moreover, the fact that there exists a noncontextual ontological model for the {\em qutrit} stabilizer subtheory~\cite{epistricted}, together with our dimension bound, is sufficient to deduce the linear scaling in this case.

\subsection{Diagram preservation implies ontic separability and more}\label{onticseparability}

Returning to representations of a GPT by some \colorbox{green!20}{$M:\widetilde{\Op}\to \RL$} satisfying the conditions of Theorem~\ref{thm:structure}: 
 if we make use of additional instances of diagram preservation beyond the three instances which we used in proving Theorem~\ref{thm:structure}, then we can derive additional constraints on the representation.
 
One important and immediate consequence of diagram preservation for composite systems is that the composite system $AB$ is represented by the tensor product of the representations of the components: $V_{AB}=V_A\otimes V_B$. In the special cases of quasiprobabilistic representations, this means that $\mathds{R}^{\Lambda_{AB}} = \mathds{R}^{\Lambda_A}\otimes \mathds{R}^{\Lambda_B} = \mathds{R}^{\Lambda_A\times \Lambda_B}$, which in turn means that $\Lambda_{AB}=\Lambda_A\times \Lambda_B$.  That is, the sample space of a composite system is the \emph{Cartesian product} of the sample spaces of the components. 

This constraint has particular significance if we consider the case of ontological models, \colorbox{Red!20}{$\widetilde{\xi}:\widetilde{\Op}\to\SubS$}, as this means that for an ontological model, the ontic state space of a composite system is the Cartesian product of the ontic state spaces of the components.  
We term the latter condition
 {\em ontic separability} (See Refs.~\cite{Spekkens2015,Harrigan}).  It is a species of reductionism, asserting, in effect, that composite systems have no holistic properties.  More precisely, the property ascriptions to composite systems are all and only the property ascriptions to their components. Yet another way of expressing the condition is that the properties of the whole supervene on the properties of the parts.

The assumption of ontic separability for ontological models has been discussed in many prior works \cite{Spekkens2015,Harrigan}, and has been a substantive assumption in certain arguments.  For instance, in Ref.~\cite{Spekkens2014}, ontic separability was used to demonstrate that  in a noncontextual ontological model, all and only projective measurements are represented outcome-deterministically.

It is also worth noting that the assumption of preparation independence in the PBR theorem~\cite{Pusey2012} follows from diagram preservation (e.g. Eq.~\eqref{forPBR}).
This connection between PBR and preservation of compositional structure has been previously explored in Sec.~4 of \cite{gheorghiu2019ontological}, in which they use this connection to derive a categorical version of the PBR theorem.

Moreover, by considering parallel composition we can also obtain the following:
\begin{proposition}
Via diagram preservation, parallel composition implies an additional constraint on the linear maps, $\chi_A$, namely,
\begin{equation}
\input{Diagrams/39_LinRep3.tikz}\quad = \quad \input{Diagrams/40_LinRep4.tikz}.
\label{parconst}
\end{equation} 
\end{proposition}
\proof
To begin, let us define 
\beq
\input{Diagrams/39_LinRep3.tikz}
\eeq
via
\begin{equation} \input{Diagrams/prod_chi_new.tikz}\quad =\quad \input{Diagrams/prod_M_new.tikz}.
\end{equation}
Next, applying this to a product state we have
\begin{equation} \input{Diagrams/prod_chi.tikz}\quad =\quad \input{Diagrams/prod_M.tikz}.
\end{equation}
Since $M$ is diagram-preserving, we have
\begin{equation} \input{Diagrams/prod_M.tikz}\quad =\quad \input{Diagrams/prod_M2.tikz}.
\end{equation}
Recalling the definition of $\chi_S$ again, we conclude that
\begin{equation} \input{Diagrams/prod_chi.tikz}\quad =\quad \input{Diagrams/prod_chi2.tikz}. \label{forPBR}
\end{equation}
In a tomographically local GPT, the product states span the entire state space, and so this implies Eq.~\eqref{parconst}.
\endproof

Now, if we consider the case of quasiprobabilistic representations, \colorbox{purple!20}{$\hat{\xi}: \widetilde{\Op} \to \QSS$}, we obtain that:
\beq
\input{Diagrams/new_product_rep.tikz}
\quad =\quad
\begin{tikzpicture}
	\begin{pgfonlayer}{nodelayer}
		\node [style=none] (0) at (0, -1.5) {};
		\node [style=small box, fill={purple!20}] (1) at (0, -0) {$\chi_{A}$};
		\node [style=none] (2) at (0, 1.5) {};
        \node[style=right label] (3) at (0,1.25) {$\mathds{R}^{\Lambda_{A}}$};
        \node[style=right label] (4) at (0,-1.25) {${A}$};
	\end{pgfonlayer}
	\begin{pgfonlayer}{edgelayer}
		\draw [qWire] (0.center) to (1);
		\draw (1) to (2.center);
	\end{pgfonlayer}
\end{tikzpicture}
\
\begin{tikzpicture}
	\begin{pgfonlayer}{nodelayer}
		\node [style=none] (0) at (0, -1.5) {};
		\node [style=small box, fill={purple!20}] (1) at (0, -0) {$\chi_{B}$};
		\node [style=none] (2) at (0, 1.5) {};
        \node[style=right label] (3) at (0,1.25) {$\mathds{R}^{\Lambda_{B}}$};
        \node[style=right label] (4) at (0,-1.25) {${B}$};
	\end{pgfonlayer}
	\begin{pgfonlayer}{edgelayer}
		\draw [qWire] (0.center) to (1);
		\draw (1) to (2.center);
	\end{pgfonlayer}
\end{tikzpicture}
.
\eeq
 Given Eq.~\eqref{eq:frameDecomp}, this is equivalent to
\beq\label{parallelcomposition}
\sum_{\lambda \in \Lambda_{A}, \lambda'\in \Lambda_B}
\begin{tikzpicture}
	\begin{pgfonlayer}{nodelayer}
		\node [style=none] (0) at (0.5, -2.5) {};
		\node [style=none] (1) at (0.5, -1.25) {};
		\node [style=point] (2) at (0.5, 1.5) {$\lambda$};
		\node [style=none] (3) at (0.5, 2.5) {};
		\node [style=right label] (4) at (0.5, 2.25) {$\mathds{R}^{\Lambda_{A}}$};
		\node [style=right label] (5) at (0.5, -2.25) {${A}$};
		\node [style=point] (6) at (2.5, 1.5) {$\lambda'$};
		\node [style=none] (7) at (2.5, 2.5) {};
		\node [style=right label] (8) at (2.5, 2.25) {$\mathds{R}^{\Lambda_{B}}$};
		\node [style=label] (12) at (1.5, -0.75) {$\widetilde{D}_{(\lambda,\lambda')}^{AB}$};
		\node [style=none] (13) at (2.5, -2.5) {};
		\node [style=none] (14) at (2.5, -1.25) {};
		\node [style=right label] (15) at (2.5, -2.25) {${B}$};
		\node [style=none] (16) at (-0.25, -1.25) {};
		\node [style=none] (17) at (1.5, 0.5) {};
		\node [style=none] (18) at (3.25, -1.25) {};
	\end{pgfonlayer}
	\begin{pgfonlayer}{edgelayer}
		\draw [qWire] (0.center) to (1.center);
		\draw (3.center) to (2);
		\draw (7.center) to (6);
		\draw [qWire] (13.center) to (14.center);
		\draw [fill={purple!20}, draw=black] (17.center)
			 to (18.center)
			 to (16.center)
			 to cycle;
	\end{pgfonlayer}
\end{tikzpicture}
\quad
=
\sum_{\lambda'\in \Lambda_A,\lambda''\in\Lambda_B}
\begin{tikzpicture}
	\begin{pgfonlayer}{nodelayer}
		\node [style=none] (0) at (0, -2.5) {};
		\node [style=copoint,fill=purple!20] (1) at (0, -1) {$\widetilde{D}_{\lambda'}^{A}$};
		\node [style=point] (2) at (0, 1.5) {$\lambda'$};
		\node [style=none] (3) at (0, 2.5) {};
        \node[style=right label] (4) at (0,2.25) {$\mathds{R}^{\Lambda_{A}}$};
        \node[style=right label] (5) at (0,-2.25) {${A}$};
	\end{pgfonlayer}
	\begin{pgfonlayer}{edgelayer}
		\draw [qWire] (0.center) to (1);
		\draw (3.center) to (2);
	\end{pgfonlayer}
\end{tikzpicture}
\
\begin{tikzpicture}
	\begin{pgfonlayer}{nodelayer}
		\node [style=none] (0) at (0, -2.5) {};
		\node [style=copoint,fill=purple!20] (1) at (0, -1) {$\widetilde{D}_{\lambda''}^{B}$};
		\node [style=point] (2) at (0, 1.5) {$\lambda''$};
		\node [style=none] (3) at (0, 2.5) {};
        \node[style=right label] (4) at (0,2.25) {$\mathds{R}^{\Lambda_{B}}$};
        \node[style=right label] (5) at (0,-2.25) {${B}$};
	\end{pgfonlayer}
	\begin{pgfonlayer}{edgelayer}
		\draw [qWire] (0.center) to (1);
		\draw (3.center) to (2);
	\end{pgfonlayer}
\end{tikzpicture}
.
\eeq
That is, diagram preservation implies that the frame representation must factorize across subsystems. In other words, the vector basis defining the frame representation must be a product basis.

\subsection{Converses to structure theorems}

In the above section we showed how all diagram-preserving quasiprobabilistic and ontological representations must have a particularly simple form given by a collection of invertible linear maps $\{\chi_A\}$ satisfying certain constraints. We now prove what is essentially the converse to each of these results.

Consider defining a map from $\widetilde{\Op}$ to $\RL$ by
\beq
\begin{tikzpicture}
	\begin{pgfonlayer}{nodelayer}
		\node [style=small box] (0) at (0, 0) {$\widetilde{T}$};
		\node [style=none] (1) at (0, 1.5) {};
		\node [style=none] (2) at (0, -1.5) {};
		\node [style=right label] (11) at (0, -1) {$A$};
		\node [style=right label] (12) at (0, 0.75) {$B$};
	\end{pgfonlayer}
	\begin{pgfonlayer}{edgelayer}
		\draw [style=qWire] (2.center) to (0);
		\draw [style=qWire] (0) to (1.center);
	\end{pgfonlayer}
\end{tikzpicture}
\quad \mapsto \quad \begin{tikzpicture}
	\begin{pgfonlayer}{nodelayer}
		\node [style=small box] (0) at (0, 0) {$\widetilde{T}$};
		\node [style=small box, fill={black!20}] (1) at (0, 1.75) {$\chi_B$};
		\node [style=small box, fill={black!20}] (2) at (0, -1.75) {$\chi_A^{-1}$};
		\node [style=none] (3) at (0, 3.25) {};
		\node [style=none] (4) at (0, -3.25) {};
		\node [style=right label] (5) at (0, -3) {$V_A$};
		\node [style=right label] (6) at (0, -1) {$A$};
		\node [style=right label] (7) at (0, 0.75) {$B$};
		\node [style=right label] (8) at (0, 2.75) {$V_B$};
	\end{pgfonlayer}
	\begin{pgfonlayer}{edgelayer}
		\draw [qWire] (2) to (0);
		\draw [qWire] (0) to (1);
		\draw (4.center) to (2);
		\draw (1) to (3.center);
	\end{pgfonlayer}
\end{tikzpicture}. \label{mapwconstr}
\eeq
Under what conditions on the set $\{\chi_A\}$ is this map a quasiprobabilistic or ontological representation?

To ensure that Eq.~\eqref{mapwconstr} defines a diagram-preserving map one must simply impose that:
\begin{equation}
\begin{tikzpicture}
	\begin{pgfonlayer}{nodelayer}
		\node [style=none] (0) at (0, 0) {$\chi_{AB}$};
		\node [style=none] (1) at (-1, 0.5) {};
		\node [style=none] (2) at (1, 0.5) {};
		\node [style=none] (3) at (1, -0.5) {};
		\node [style=none] (4) at (-1, -0.5) {};
		\node [style=none] (5) at (-0.5, 0.5) {};
		\node [style=none] (6) at (-0.5, 1.25) {};
		\node [style=none] (7) at (0.5, 1.25) {};
		\node [style=none] (8) at (0.5, 0.5) {};
		\node [style=none] (9) at (-0.5, -0.5) {};
		\node [style=none] (10) at (-0.5, -1.25) {};
		\node [style=none] (11) at (0.5, -0.5) {};
		\node [style=none] (12) at (0.5, -1.25) {};
		\node [style=right label] (13) at (-0.5, -1.25) {$A$};
		\node [style=right label] (14) at (0.5, -1.25) {$B$};
		\node [style=right label] (15) at (-0.5, 1) {$V_A$};
		\node [style=right label] (16) at (0.5, 1) {$V_B$};
	\end{pgfonlayer}
	\begin{pgfonlayer}{edgelayer}
			\filldraw[fill=black!20,draw=black](1.center) to (2.center) to (3.center) to (4.center) to cycle;
		\draw (6.center) to (5.center);
		\draw (7.center) to (8.center);
		\draw [qWire] (9.center) to (10.center);
		\draw [qWire] (11.center) to (12.center);
	\end{pgfonlayer}
\end{tikzpicture}
\quad = \quad \begin{tikzpicture}
	\begin{pgfonlayer}{nodelayer}
		\node [style=none] (0) at (-0.5, -1.25) {};
		\node [style=small box, fill={black!20}] (1) at (-0.5, 0) {$\chi_A$};
		\node [style=none] (2) at (-0.5, 1.25) {};
		\node [style=none] (3) at (1, 1.25) {};
		\node [style=small box, fill={black!20}] (4) at (1, 0) {$\chi_B$};
		\node [style=none] (5) at (1, -1.25) {};
		\node [style=none] (6) at (-0.5, -1) {};
		\node [style=right label] (7) at (-0.5, -1.25) {$A$};
		\node [style=right label] (8) at (1, -1.25) {$B$};
		\node [style=right label] (9) at (-0.5, 1) {$V_A$};
		\node [style=right label] (10) at (1, 1) {$V_B$};
	\end{pgfonlayer}
	\begin{pgfonlayer}{edgelayer}
		\draw [qWire] (0.center) to (1);
		\draw (1) to (2.center);
		\draw [qWire] (5.center) to (4);
		\draw (4) to (3.center);
	\end{pgfonlayer}
\end{tikzpicture}
.
\end{equation}
This condition, together with invertibility and linearity of the $\chi_A$, easily implies that diagram preservation and indeed all the assumptions of Theorem~\ref{mainthm} are satisfied.

To ensure that Eq.~\eqref{mapwconstr} defines a linear representation that is moreover a quasiprobabilistic representation, as in Def.~\ref{def:quasi}, one must impose that $V_A=\mathds{R}^{\Lambda_A}$ and that
\begin{equation}
\begin{tikzpicture}
	\begin{pgfonlayer}{nodelayer}
		\node [style=none] (0) at (-0.5, -1) {};
		\node [style={small box}, fill={black!20}] (1) at (-0.5, -0) {$\chi_A$};
		\node [style=none] (2) at (-0.5, 1.25) {};
		\node [style=upground] (3) at (-0.5, 1.5) {};
		\node [style=right label] (4) at (-0.5, -1) {$A$};
		\node [style=right label] (5) at (-0.5, 0.75) {$\mathds{R}^{\Lambda_A}$};
	\end{pgfonlayer}
	\begin{pgfonlayer}{edgelayer}
		\draw [qWire] (0.center) to (1);
		\draw (1) to (2.center);
	\end{pgfonlayer}
\end{tikzpicture}
\quad =\quad  \begin{tikzpicture}
	\begin{pgfonlayer}{nodelayer}
		\node [style=none] (0) at (0, -0.75) {};
		\node [style=none] (1) at (0, 0.25) {};
		\node [style=upground] (2) at (0, 0.5) {};
		\node [style=right label] (3) at (0, -0.5) {$A$};
	\end{pgfonlayer}
	\begin{pgfonlayer}{edgelayer}
		\draw [qWire] (0.center) to (1.center);
	\end{pgfonlayer}
\end{tikzpicture},
\end{equation}
which implies that the conditions in Def.~\ref{def:quasi} are satisfied.

Finally, to ensure that Eq.~\eqref{mapwconstr} defines a quasiprobabilistic representation that is moreover an ontological representation, as in Def.~\ref{defnontgpt}, one must introduce a positivity constraint for this map. Specifically, one must have that
\beq
\begin{tikzpicture}
	\begin{pgfonlayer}{nodelayer}
		\node [style={small box}, fill={black!20}] (0) at (0, 1.75) {$\chi_B$};
		\node [style={small box}, fill={black!20}] (1) at (0, -1.75) {$\chi_A^{-1}$};
		\node [style=none] (2) at (0, 3.25) {};
		\node [style=none] (3) at (0, -3.25) {};
		\node [style={right label}] (4) at (0, -3) {$\mathds{R}^{\Lambda_A}$};
		\node [style={right label}] (5) at (0, -1) {$A$};
		\node [style={right label}] (6) at (0, 0.75) {$B$};
		\node [style={right label}] (7) at (0, 2.75) {$\mathds{R}^{\Lambda_B}$};
		\node [style=none] (8) at (0, 0.5) {};
		\node [style=none] (9) at (0, -0.5) {};
		\node [style=none] (10) at (-0.5, 0.5) {};
		\node [style=none] (11) at (0.5, 0.5) {};
		\node [style=none] (12) at (0.5, -0.5) {};
		\node [style=none] (13) at (-0.5, -0.5) {};
	\end{pgfonlayer}
	\begin{pgfonlayer}{edgelayer}
		\draw (3.center) to (1);
		\draw (0) to (2.center);
		\draw [qWire] (0) to (8.center);
		\draw [qWire] (9.center) to (1);
		\draw[thick gray dashed edge] (10.center) to (11.center);
		\draw [thick gray dashed edge](11.center) to (12.center);
		\draw[thick gray dashed edge] (12.center) to (13.center);
		\draw [thick gray dashed edge](13.center) to (10.center);
	\end{pgfonlayer}
\end{tikzpicture}
\quad :: \quad
\begin{tikzpicture}
	\begin{pgfonlayer}{nodelayer}
		\node [style=small box] (0) at (0, 0) {$\widetilde{T}$};
		\node [style=none] (3) at (0, 1) {};
		\node [style=none] (4) at (0, -1) {};
		\node [style=right label] (6) at (0, -1) {$A$};
		\node [style=right label] (7) at (0, 0.75) {$B$};
	\end{pgfonlayer}
	\begin{pgfonlayer}{edgelayer}
		\draw [qWire] (4) to (0);
		\draw [qWire] (0) to (3);
	\end{pgfonlayer}
\end{tikzpicture}
\quad  \mapsto \quad\begin{tikzpicture}
	\begin{pgfonlayer}{nodelayer}
		\node [style=small box] (0) at (0, 0) {$\widetilde{T}$};
		\node [style=small box, fill={black!20}] (1) at (0, 1.75) {$\chi_B$};
		\node [style=small box, fill={black!20}] (2) at (0, -1.75) {$\chi_A^{-1}$};
		\node [style=none] (3) at (0, 3.25) {};
		\node [style=none] (4) at (0, -3.25) {};
		\node [style=right label] (5) at (0, -3) {$\mathds{R}^{\Lambda_A}$};
		\node [style=right label] (6) at (0, -1) {$A$};
		\node [style=right label] (7) at (0, 0.75) {$B$};
		\node [style=right label] (8) at (0, 2.75) {$\mathds{R}^{\Lambda_B}$};
	\end{pgfonlayer}
	\begin{pgfonlayer}{edgelayer}
		\draw [qWire] (2) to (0);
		\draw [qWire] (0) to (1);
		\draw (4.center) to (2);
		\draw (1) to (3.center);
	\end{pgfonlayer}
\end{tikzpicture}
\eeq
defines a positive map from the cone of transformations from $A$ to $B$ in $\widetilde{\Op}$ to the cone of stochastic maps from $\Lambda_A$ to $\Lambda_B$.

This provides a simple recipe for constructing linear representations: one simply needs to choose a family of invertible linear maps for each fundamental system (i.e., one that cannot be further decomposed into subsystems),
 and then define the others as the tensor product of these (so that general $\chi_A$ factorise over subsystems). Similarly, it provides a simple recipe for constructing quasiprobabilistic representations, using the same construction but where the $\chi_A$ preserve the deterministic effect. In the case of noncontextual ontological models, however, the recipe is less simple: one must not only choose invertible linear maps which  factorise over subsystems and preserve the deterministic effect; one must also check the positivity condition (which is nontrivial, since for any particular map $\chi_A$, one must check the condition for every $\chi_B$).

\subsection{Categorical reformulation}

There is an elegant categorical reframing of our structure theorem, as suggested to us by one of the referees:
\begin{theorem}[IV.1 categorical version]
Any convex-linear, empirically adequate and diagram-preserving map (Eq.~\eqref{eq:DPM}), \colorbox{green!20}{$M:\widetilde{\Op}\to \RL$}, where $\widetilde{\Op}$ is tomographically local, is naturally isomorphic to the canonical representation $R:\widetilde{\Op}\to \RL$ defined in Theorem~\ref{OpinRL}.  Moreover, this natural isomorphism $M\implies R$ is unique. 
\end{theorem} 
\proof
The components of the natural isomorphism are given by the $\chi_A$, as these go from $A\to V_A = M(A)$, and where we are abusing notation by denoting $R(A)$ simply as $A$. Eq.~\eqref{eq:repTrans} ensures that these do define a natural transformation.  To see this, first recall that we are abusing notation by suppressing explicitly notating the canonical representation $R$, so Eq.~\eqref{eq:repTrans} really tells us that $M(\widetilde{T}) = \chi_B \circ R(\widetilde{T}) \circ \chi_A^{-1}$ which is equivalent to $M(\widetilde{T})\circ \chi_A = \chi_B \circ R(\widetilde{T})$, which is what we need for this to be a natural transformation. Clearly it is moreover a natural isomorphism as every $\chi_A$ is invertible (Eqs.~\eqref{step3a} and \eqref{step3b}). Eq.~\eqref{parconst} then ensures that this is a monoidal natural isomorphism. To see this, note that the LHS of this equation is not quite the component of the natural isomorphism $\chi_{A\otimes B}$, instead we have $\chi_{AB}:= \mu^{-1}\circ \chi_{A\otimes B} \circ \mu$, and so Eq.~\eqref{parconst} is equivalent to the condition that $\chi_{A\otimes B} \circ \mu = \mu \circ (\chi_A\otimes\chi_B)$, which is exactly what we need for this to be a monoidal natural isomorphism.  Uniqueness of this natural isomorphism follows from the fact that the components $\chi_A$ are uniquely determined, as noted in Theorem~\ref{thm:structure}. 
\endproof

This means that ontological models, should they exist, are essentially unique, in that they are unique up to a unique natural isomorphism. There is, however, an important subtlety going on here. This natural isomorphism is given by viewing ontological models as living in $\RL$, and so the components of the natural isomorphism are just invertible real linear maps.

 Alternatively, however, one could demand a stricter notion of isomorphism between ontological models
 by viewing them as living in $\mathsf{SubStoch}$, in which case the components of a natural isomorphism would be invertible stochastic maps, i.e., permutations of the ontic states. This is a much stricter notion of isomorphism, and it is likely that in this sense there are many different ontological models for a given GPT. In certain situations, however, even this stricter notion can be proven, see, for example, Ref.~\cite{schmid2021only}.

\section{Revisiting our assumptions} \label{revisassump}

We have derived surprisingly strong constraints on the form of noncontextual ontological models of operational theories,
and so it is important to examine the assumptions that went into deriving these constraints.
These concerned both the types of operational theories under consideration and the types of ontological representations of these,
and were summarized in Section~\ref{secassumptions}. The majority of these are ubiquitous and well-motivated.
The only notable restriction on the scope of operational theories we consider is the one induced by our  assumption of tomographic locality  (as discussed further in Section~\ref{neclt}).
Similarly, the only notable restriction
on ontological (and quasiprobabilistic) models that is warranting of further discussion is that they are diagram-preserving.

\subsection{Revisiting diagram preservation}

In the case of ontological models, we will provide below a motivation for the instances of diagram preservation that we required for our proofs. Since we have defined quasiprobabilistic models as representations of operational theories wherein the {\em only} difference from an ontological model is that the probabilities are allowed to become quasiprobabilities (i.e., drawn from the reals rather than the interval $[0,1]$), it follows that these same motivations are also applicable to them.
It is worth noting that among the quasiprobabilistic representations for continuous variable quantum systems that are most studied in the literature, the Wigner representation satisfies our definition\footnote{Strictly speaking, one would need to generalize our definition to the case of infinite-dimensional GPTs.} while the Q~\cite{Husimi} and P representations~\cite{Pfunc1,Pfunc2} do not (as they are defined by overcomplete frames).  There are also examples of both types among quasiprobabilistic representations of finite-dimensional systems in quantum theory.
In particular, Ref.~\cite{citeGHW} defines a family of discrete Wigner representations, some of which satisfy the assumption of diagram preservation, and some of which do not. Of special note among those that satisfy the assumption is Gross' discrete Wigner representation~\cite{gross2006},  which is the unique representation in this family that satisfies a natural covariance property. 
Ref.~\cite{schmid2021only} further shows that this is the unique noncontextual ontological model for stabilizer subtheories in odd dimensions.

Although we endorse diagram preservation in its most general form, it is worth noting that
our main results (given in Section~\ref{mainresults}) require only the following very specific {\em instances} of that assumption:
\begin{enumerate}[label=(\roman*)]
\item diagram preservation for prepare-measure scenarios,
\begin{equation}\label{dp1}
\widetilde{\xi}::\
	\begin{tikzpicture}
	\begin{pgfonlayer}{nodelayer}
		\node [style=copoint] (0) at (0, .75) {$\widetilde{E}$};
		\node [style=point] (1) at (0, -.75) {$\widetilde{P}$};
		\node [style=none] (4) at (0.75, 0.25) {};
		\node [style=none] (13) at (0, -0.25) {};
	\end{pgfonlayer}
	\begin{pgfonlayer}{edgelayer}
		\draw [qWire](1) to (0);
	\end{pgfonlayer}
\end{tikzpicture}
\quad \mapsto\quad \begin{tikzpicture}
	\begin{pgfonlayer}{nodelayer}
		\node [style=copoint] (0) at (0, 1) {$\widetilde{E}$};
		\node [style=point] (1) at (0, -1) {$\widetilde{P}$};
		\node [style=none] (2) at (0.5, 0.375) {\tiny$\tilde{\xi}$};
		\node [style=none] (3) at (0.5, -1.5) {\tiny$\tilde{\xi}$};
		\node [style=none] (4) at (0.75, 0.125) {};
		\node [style=none] (5) at (0.75, 1.75) {};
		\node [style=none] (6) at (-0.75, 1.75) {};
		\node [style=none] (7) at (-0.75, 0.125) {};
		\node [style=none] (8) at (-0.75, -0.25) {};
		\node [style=none] (9) at (-0.75, -1.75) {};
		\node [style=none] (10) at (0.75, -1.75) {};
		\node [style=none] (11) at (0.75, -0.25) {};
		\node [style=none] (12) at (0, 0.25) {};
		\node [style=none] (13) at (0, -0.25) {};
	\end{pgfonlayer}
	\begin{pgfonlayer}{edgelayer}
		\filldraw[fill=Red!20,draw=Red!40](5.center) to (6.center) to (7.center) to (4.center) to cycle;
		\filldraw[fill=Red!20,draw=Red!40](8.center) to (9.center) to (10.center) to (11.center) to cycle;
		\draw [qWire](1) to (13.center);
		\draw (13.center) to (12.center);
		\draw [qWire](12.center) to (0);
	\end{pgfonlayer}
\end{tikzpicture}
.
\end{equation}
and diagram preservation for measure-and-reprepare processes
\beq \label{dp2}
\widetilde{\xi}::\
	\begin{tikzpicture}
	\begin{pgfonlayer}{nodelayer}
		\node [style=none] (0) at (0, 2) {};
		\node [style=none] (1) at (0, -2) {};
		\node [style=copoint] (10) at (0, -1) {$\widetilde{E}$};
		\node [style=point] (11) at (0, 1) {$\widetilde{P}$};
	\end{pgfonlayer}
	\begin{pgfonlayer}{edgelayer}
		\draw [style=qWire] (10) to (1.center);
		\draw [style=qWire] (11) to (0.center);
	\end{pgfonlayer}
\end{tikzpicture}
\quad \mapsto\quad \begin{tikzpicture}
	\begin{pgfonlayer}{nodelayer}
		\node [style=none] (0) at (0, 2) {};
		\node [style=none] (1) at (0, -2) {};
		\node [style=none] (2) at (0.75, -1.75) {\tiny$\tilde{\xi}$};
		\node [style=none] (3) at (0, 2.75) {};
		\node [style=none] (4) at (0, -2.75) {};
		\node [style=none] (5) at (-1, -0.25) {};
		\node [style=none] (6) at (1, -0.25) {};
		\node [style=none] (7) at (1, -2) {};
		\node [style=none] (8) at (-1, -2) {};
		\node [style=copoint] (9) at (0, -1.25) {$\widetilde{E}$};
		\node [style=point] (10) at (0, 1.25) {$\widetilde{P}$};
		\node [style=none] (11) at (0.75, 0.5) {\tiny$\tilde{\xi}$};
		\node [style=none] (12) at (-1, 2) {};
		\node [style=none] (13) at (-1, 0.25) {};
		\node [style=none] (14) at (1, 0.25) {};
		\node [style=none] (15) at (1, 2) {};
	\end{pgfonlayer}
	\begin{pgfonlayer}{edgelayer}
		\filldraw[fill=Red!20,draw=Red!40](5.center) to (6.center) to (7.center) to (8.center) to cycle;
		\filldraw[fill=Red!20,draw=Red!40](12.center) to (13.center) to (14.center) to (15.center) to cycle;
		\draw (3.center) to (0.center);
		\draw (1.center) to (4.center);
		\draw [style=qWire] (9) to (1.center);
		\draw [style=qWire] (10) to (0.center);
	\end{pgfonlayer}
\end{tikzpicture}
,
\eeq
\item  diagram preservation for the identity process
\begin{equation}\label{dp3}
\widetilde{\xi}::\ \begin{tikzpicture}
	\begin{pgfonlayer}{nodelayer}
		\node [style=none] (0) at (0, 1) {};
		\node [style=none] (1) at (0, -1) {};
		\node [style=right label] (2) at (0, -.5) {$A$};
	\end{pgfonlayer}
	\begin{pgfonlayer}{edgelayer}
		\draw [qWire](0.center) to (1.center);
	\end{pgfonlayer}
\end{tikzpicture}
\quad \mapsto\quad \begin{tikzpicture}
	\begin{pgfonlayer}{nodelayer}
		\node [style=none] (0) at (0, 1) {};
		\node [style=none] (1) at (0, -1) {};
		\node [style=right label] (2) at (0, -.5) {$\Lambda_A$};
	\end{pgfonlayer}
	\begin{pgfonlayer}{edgelayer}
		\draw (0.center) to (1.center);
	\end{pgfonlayer}
\end{tikzpicture}.
\end{equation}
\end{enumerate}

These are easily justified.

Eq.~\eqref{dp1} captures the idea that the ontic state is the complete causal mediary between the preparation and the effect.  This assumption is built into the very definition of the standard ontological models framework
(implicitly in early work~\cite{Spekkens2005,Harrigan} and explicitly in later work~\cite{Krishna_2017,Schmid2018}), and is 
 assumed in virtually every past work on ontological models.

Eq.~\eqref{dp2} is a similarly natural assumption.
The natural view of the process
\beq
\begin{tikzpicture}
	\begin{pgfonlayer}{nodelayer}
		\node [style=none] (0) at (0, 2) {};
		\node [style=none] (1) at (0, -2) {};
		\node [style=none] (3) at (0, 2.75) {};
		\node [style=none] (4) at (0, -2.75) {};
		\node [style=none] (5) at (-1, -0.25) {};
		\node [style=none] (6) at (1, -0.25) {};
		\node [style=none] (7) at (1, -2) {};
		\node [style=none] (8) at (-1, -2) {};
		\node [style=copoint] (9) at (0, -1) {$\widetilde{E}$};
		\node [style=point] (10) at (0, 1) {$\widetilde{P}$};
		\node [style=none] (12) at (-1, 2) {};
		\node [style=none] (13) at (-1, 0.25) {};
		\node [style=none] (14) at (1, 0.25) {};
		\node [style=none] (15) at (1, 2) {};
	\end{pgfonlayer}
	\begin{pgfonlayer}{edgelayer}
		\draw [style=qWire] (9) to (1.center);
		\draw [style=qWire] (10) to (0.center);
	\end{pgfonlayer}
\end{tikzpicture}
\eeq
is that one has observed effect $\widetilde{E}$ and then one has independently implemented the preparation $\widetilde{P}$. There need not be any system acting as a causal mediary between $\widetilde{E}$ and $\widetilde{P}$.
  The natural ontological representation, therefore,
is one wherein
there is no ontic state mediating the two processes,
as depicted in Eq.~\eqref{dp2}.  Although we are not aware of this assumption having been made in previous works, it is directly analogous to the
preparation-independence assumption made in Ref.~\cite{Pusey2012} (which involved two independent states, rather than an independent effect and state).

Eq.~\eqref{dp3} can be justified by noting that within the equivalence class of procedures associated to the identity operation in the GPT, there is the one which corresponds to waiting for a vanishing amount of time.
In any reasonable physical theory, no evolution is possible in vanishing time, and hence the only valid ontological representation of such an equivalence class of procedures is the identity map on the ontic state space.

Because we consider the full assumption of diagram preservation to be a natural generalization of all of these specific assumptions, we have endorsed it in our definitions.
 See Appendix B of Ref.~\cite{schmid2020unscrambling} for a more thorough defense of this full assumption.

\subsection{Necessity of tomographic locality} \label{neclt}
The assumption of tomographic locality is common in the GPT literature, so we will not attempt a defense of it here. Nevertheless, it is natural to ask if the assumption is actually necessary to obtain our structure theorems. Here we provide an example which shows that it is. The operational theory we consider in our example is the real-amplitude version of the qutrit stabilizer subtheory of quantum theory\footnote{Note that in order for this to be an operational theory according to our definitions herein, we are here referring to the `convexified' stabilizer subtheory in which all convex combinations of, for example, the pure stabilizer states are permitted.}. In this subtheory, two systems are described by $45$ parameters, whereas only $6^2 = 36$ parameters are available from local measurements, which immediately implies that the theory is not tomographically local (just as the real-amplitude version of the full quantum theory fails to be tomographically local~\cite{hardy2012limited}).

To begin with, consider the standard (complex-amplitude) qutrit stabilizer subtheory. Gross's discrete Wigner function~\cite{gross2006} 
provides a (diagram-preserving) quasiprobability representation for qutrits for which the stabilizer subtheory is positively represented.  By Corollary~\ref{cor:three}, this corresponds to a noncontextual ontological model of the subtheory. 
Indeed, this ontological model has been examined in Ref.~\cite{epistricted}, where it is shown that it can be reconstructed from an ``epistemic restriction''. Since the standard qutrit stabilizer subtheory is tomographically local, these models obey our structure theorems. In particular, the representation of $n$ qutrits uses $9^n$ ontic states, matching the dimension of the relevant space of density matrices.

Now consider the subtheory consisting of only those qutrit stabilizer procedures that can be represented using real amplitudes. This does not introduce any new operational equivalences, and so the model discussed above is still noncontextual when restricted to this subtheory. But now our structure theorem does not hold, because this model still uses $9^n$ ontic states even though the density matrices now live in a $\frac12 3^n(3^n+1)$-dimensional space.

Moreover, we can show that this sort of example is rather generic, that is, that there is no hope of obtaining a structure theorem with our dimension bound for any theory that is not tomographically local. 

Suppose that we have any ontological representation $\widetilde{\xi}$ of some GPT wherein each GPT system of dimension $d$ is represented by an ontic state space of cardinality $d$. Then the GPT is necessarily tomographically local.

To see this, consider the representation of GPT transformations from this state space to itself. Then, the map $\widetilde{\xi}$ is a linear map from the space of transformations to
$d\times d$ substochastic matrices, which are $d^2$ dimensional. By empirical
adequacy $\xi$ is injective, and so the space of GPT transformations is at most $d^2$
dimensional. But the effect-prepare channels already span $d^2$
dimensions, so there cannot be any channels outside this span. Hence,
by Corollary~\ref{corolTL}, the theory is tomographically local.


\section{Outlook}

These results can be directly applied to the study of contextuality in specific scenarios and theories. For instance, we have already seen that our dimension bound is a useful tool for obtaining novel proofs of contextuality (e.g., via Hardy's ontological excess baggage theorem \cite{Hardy2004} or for the 8-state model of Ref.~\cite{8state}), and for providing novel algorithms for deriving noise-robust noncontextuality inequalities (namely, the algorithm in Ref.~\cite{schmid2019characterization} but informed by our dimension bound).
 It remains to be seen whether other algorithms for witnessing nonclassicality, such as those in Ref.~\cite{Schmid2018} or Ref~\cite{Krishna_2017}, could be extended within our framework to more general compositional scenarios.

Our formalism is also ideally suited to understanding the information-theoretic advantages afforded by  contextual operational theories, such as for computational speedup, since it has the compositional flexibility to describe arbitrary scenarios, such as families of circuits which arise in the gate-based model of computation.
In fact, our structure theorem is a major first step in simplifying the proof that contextuality is a necessary resource for the state-injection model of quantum computing~\cite{magic,schmid2021only}. Ref.~\cite{schmid2021only} shows that such a proof can proceed by applying our structure theorem to show that the {\em only} positive quasiprobabilistic models of the (classically-simulable) stabilizer subtheory for odd dimensions are given by Gross's discrete Wigner function~\cite{gross2006}; then, the known fact that the injected resource states necessarily have negative representation in this particular model establishes the result in a direct and elegant fashion.

The key limitation of our results is the assumption that the  GPT associated to the operational theory under consideration is tomographically local.
There are two potential approaches to dealing with this limitation. On the one hand, one could provide an argument that theories which are not tomographically local are undesirable in some principled sense. For example, it seems likely that one can rule them out on the grounds that they violate Leibniz's methodological principle~\cite{Leibniz}. From a practical perspective, wherein the goal is to experimentally verify nonclassicality in a theory-independent manner, one would instead be motivated to seek experimental evidence that nature truly satisfies tomographic locality, independent of the validity of quantum theory. One possible approach to this end would be to extend
 the techniques introduced in \cite{bootstraptomography} to composite systems.

\section*{Acknowledgements}
We thank Matt Leifer, Lorenzo Catani, Shane Mansfield, Martti Karvonen, Alex Wilce, and our reviewers for useful discussions.
We thank Giulio Chiribella for useful discussions on quotiented operational theories and their relation to GPTs.
RWS thanks Bob Coecke for early discussions on the topic of understanding noncontextuality in a category-theoretic way. JHS thanks Lucien Hardy for many discussions on tomographic locality and the duotensor formalism during PSI and since, and thanks John van de Wetering for helpful discussions on his work on quasistochastic representations of quantum theory. 
JHS also thanks Matt Wilson for his help on understanding the relationship between diagram preserving maps and strong monoidal functors. 
We also acknowledge the First Perimeter Institute-Chapman Workshop on Quantum Foundations (PIMan) as the venue that spawned this collaborative project.
DS was supported by a Vanier Canada Graduate Scholarship. MFP is supported by the Royal Commission for the Exhibition of 1851. JHS and DS are supported by the
National Science Centre, Poland (Opus project, Categorical
Foundations of the Non-Classicality of Nature, project
no. 2021/41/B/ST2/03149). This research was supported by Perimeter Institute for Theoretical Physics. Research at Perimeter Institute is supported in part by the Government of Canada through the Department of Innovation, Science and Economic Development Canada and by the Province of Ontario through the Ministry of Colleges and Universities.

\bibliographystyle{apsrev4-2}
\bibliography{context}

\appendix
\onecolumngrid

\section{Context-dependence in representations of GPTs}\label{contextsingpts}

In the main text, we stated that the notion of an ontological model of a GPT that we have defined cannot be said to be either generalized-contextual or generalized-noncontextual (unlike our notion of an ontological model of an operational theory). 
We will now elaborate on this point.

Consider the contexts that one may wish to associate to a GPT state. 
One of the examples which appears in the literature corresponds to different decompositions of the GPT state into mixtures of other GPT states, for example:
\beq
\sum_i p_i s_i = s = \sum_j q_j s'_j.\label{gptconvexcontext}
\eeq
Now consider any ontological representation map which has the GPT as its domain. In the GPT, all three terms in Eq.~\eqref{gptconvexcontext} are {\em strictly equal}, and hence all three map to the same probability distribution over $\Lambda$. As such, there is no possibility for the map to represent an $s$ arising from the LHS mixture differently from how it represents an $s$ arising from the RHS mixture.

A natural question one might consider in light of this is how one should represent {\em ensembles} of states ontologically. The ensembles of relevance in the example just given are $\{(p_i,s_i)\}$, $\{(q_j,s'_j)\}$, and $\{(1,s)\}$; all of these are operationally equivalent:
\beq
\{(p_i,s_i)\} \sim \{(1,s)\} \sim \{(q_j,s'_j)\},
\eeq
but not strictly equal. If one defines a new kind of ontological representation map which acts on such objects, then it {\em could} take these distinct ensembles to distinct probability distributions over $\Lambda$. One could then meaningfully talk about whether such a representation depended on context or not.

However, the notion of ontological representation for a GPT that we have defined herein has as its domain processes within the GPT (such as states), {\em not} ensembles of such processes. This is also true for the more general quasistochastic representations of GPTs. As such, applying the notion of generalized contextuality to them is a category mistake, just as it would be a category mistake to ask whether a variable $X$ depends on another variable $Y$ if $Y$ cannot possibly vary~\cite{schmid2019characterization}. Because standard quasiprobability representations (such as Wigner's or Gross's) are instances of our definition (and in particular, because they take the domain of the representation to be states and effects rather than ensembles of states or ensembles of effects), it is equally meaningless to ask whether they are noncontextual or contextual. 


Of course, one {\em could} define a map which has as its domain {\em the set of ensembles of GPT processes}. For such a map, it would be appropriate to ask whether or not the map is noncontextual. This is similar to what is done in the causal-inferential framework of Ref.~\cite{schmid2020unscrambling}, where the central objects of study are ensembles of processes corresponding to an agent's knowledge of what process occurred (although with the difference that in this case we consider ensembles of {\em unquotiented} processes). In that context, we formalize the resulting notion of an ontological representation, as well as the natural generalization of the notion of `noncontextuality' that arises for it.

A similar story holds for the notion of context that is relevant for the study of Kochen-Specker contextuality. Consider two measurements, $M_1$ and $M_2$, which we conceptualize as processes with a GPT input and a classical output. Suppose that these each have a particular outcome, labeled $a$ and $b$ respectively, which correspond to the same GPT effect:
\beq
\begin{tikzpicture}
	\begin{pgfonlayer}{nodelayer}
		\node [style=small box] (0) at (0, 0) {$M_1$};
		\node [style=copoint] (1) at (0, 1.25) {$a$};
		\node [style=none] (3) at (0, -1) {};
	\end{pgfonlayer}
	\begin{pgfonlayer}{edgelayer}
		\draw [cWire] (1) to (0);
		\draw [qWire] (0) to (3.center);
	\end{pgfonlayer}
\end{tikzpicture}
\quad = \quad \begin{tikzpicture}
	\begin{pgfonlayer}{nodelayer}
		\node [style=small box] (0) at (0, 0) {$M_2$};
		\node [style=copoint] (1) at (0, 1.25) {$b$};
		\node [style=none] (3) at (0, -1) {};
	\end{pgfonlayer}
	\begin{pgfonlayer}{edgelayer}
		\draw [cWire] (1) to (0);
		\draw [qWire] (0) to (3.center);
	\end{pgfonlayer}
\end{tikzpicture}
.
\eeq

The fact that the effect associated to getting outcome $a$ in measurement $M_1$ is {\em strictly equal} to the effect associated to getting outcome $b$ in measurement $M_2$ implies that any map which has the GPT effect space as its domain must represent the two cases identically.  Again, one finds that there is no possibility for a representation map with this specific choice of domain to depend on whether or not the effect was realized using measurement $M_1$ or $M_2$.

But, also as above, one could choose to consider a different kind of ontological representation map in which the domain is no longer the set of GPT processes per se, but something else, which includes, for instance, measurement-outcome pairs. In this particular case, we are interested in pairs $(M_1,a)$ and $(M_2,b)$ which are operationally equivalent, 
\beq
\left(\begin{tikzpicture}
	\begin{pgfonlayer}{nodelayer}
		\node [style=small box] (0) at (0, 0) {$M_1$};
		\node [style=none] (1) at (0, 1) {};
		\node [style=none] (3) at (0, -1) {};
	\end{pgfonlayer}
	\begin{pgfonlayer}{edgelayer}
		\draw [cWire] (1.center) to (0);
		\draw [qWire] (0) to (3.center);
	\end{pgfonlayer}
\end{tikzpicture}
,\begin{tikzpicture}
	\begin{pgfonlayer}{nodelayer}
		\node [style=none] (0) at (0, -0.5) {};
		\node [style=copoint] (1) at (0, 0.5) {$a$};
	\end{pgfonlayer}
	\begin{pgfonlayer}{edgelayer}
		\draw [cWire] (1) to (0.center);
	\end{pgfonlayer}
\end{tikzpicture}
\right) \sim \left(\begin{tikzpicture}
	\begin{pgfonlayer}{nodelayer}
		\node [style=small box] (0) at (0, 0) {$M_2$};
		\node [style=none] (1) at (0, 1) {};
		\node [style=none] (3) at (0, -1) {};
	\end{pgfonlayer}
	\begin{pgfonlayer}{edgelayer}
		\draw [cWire] (1.center) to (0);
		\draw [qWire] (0) to (3.center);
	\end{pgfonlayer}
\end{tikzpicture}
,\begin{tikzpicture}
	\begin{pgfonlayer}{nodelayer}
		\node [style=none] (0) at (0, -0.5) {};
		\node [style=copoint] (1) at (0, 0.5) {$b$};
	\end{pgfonlayer}
	\begin{pgfonlayer}{edgelayer}
		\draw [cWire] (1) to (0.center);
	\end{pgfonlayer}
\end{tikzpicture}
\right),
\eeq
but not strictly equal. 
If one defines a new kind of ontological representation map which acts on such objects, then it {\em could} take distinct such objects 
 to distinct  response functions. One could then meaningfully talk about whether such a representation depended on context or not. This is what is typically done (if only implicitly) in the study of Kochen-Specker noncontextuality.

\section{Proof of the structure theorem (Theorem~\ref{mainthm})} \label{mainproof}

We now complete the proof of Theorem~\ref{mainthm}, as sketched in the main text.

\begin{proof}
Since we are assuming tomographic locality of the GPT, Corollary~\ref{cor:Tdecomp} immediately gives
\begin{equation}
\input{Diagrams/50_Proof2.tikz}\quad =\quad \input{Diagrams/51_Proof3.tikz}.
\end{equation}
Since $M$ is convex-linear  and preserves the zero processes\footnote{ This follows from the fact that one can construct the zero process by composing a state and an effect with the zero scalar as $0_A^B = \widetilde{P}\circ 0 \circ \widetilde{E}$. Then, by empirical adequacy of $M$, one has $M(0)=0$, and so diagram-preservation of $M$ then gives $M(0_A^B) = M(\widetilde{P})\circ M(0) \circ M(\widetilde{E}) = M(\widetilde{P})\circ 0 \circ M(\widetilde{E}) = 0_{V_A}^{V_B}$.},
 and since  
the effect-state channels span the vector space, $M$ can be uniquely extended to a linear map $\hat{M}$.  Hence, 
\begin{equation}
\input{Diagrams/51_Proof3.tikz}\quad =\quad \input{Diagrams/51_Proof3RHS.tikz}.
\end{equation}
Now, using the linearity of $\hat{M}$, we have
\begin{equation}
\input{Diagrams/51_Proof3RHS.tikz}\quad =\ \sum_{ij} r_{ij}\ \input{Diagrams/54_Proof6hat.tikz}.
\end{equation}
Noting that in this diagram, $\hat{M}$ is only applied to objects in the domain of $M$, on which the two maps act identically (by the fact that the former is the linear extension of the latter), one has
\begin{equation}
\sum_{ij} r_{ij} \input{Diagrams/54_Proof6hat.tikz}\quad =\ \sum_{ij} r_{ij}\ \input{Diagrams/54_Proof6.tikz}\quad =\ \sum_{ij} r_{ij}\ \input{Diagrams/55_Proof7.tikz}.
\end{equation}
where the last step follows from the fact that $M$ is diagram-preserving.
In summary, we have shown that
\begin{equation}
\input{Diagrams/50_Proof2.tikz}\ =\ \sum_{ij} r_{ij} \ \input{Diagrams/55_Proof7.tikz},\label{eq:step1conc}
\end{equation}
as claimed in Eq.~\eqref{step1}.

Next, we analyse $M$ in the specific case of a state $\widetilde{P}_i$:
\tikeq[.]{59_Proof11}
Since the DP map $M$  has a unique linear extension which  takes the vector space of GPT states $B$ to the vector space $V_B$, and since both of these are in $\RL$, one can uniquely re-interpret the action of $M$ as a process $\chi$ within $\RL$:
\begin{equation}\input{Diagrams/59_Proof11.tikz}\quad =\quad \begin{tikzpicture}
	\begin{pgfonlayer}{nodelayer}
		\node [style=point] (0) at (0, -1) {$\widetilde{P}_i$};
		\node [style=small box, fill={green!20}] (1) at (0, -0) {$\chi_B$};
		\node [style=none] (2) at (0, 1) {};
	\end{pgfonlayer}
	\begin{pgfonlayer}{edgelayer}
		\draw [qWire] (0) to (1);
		\draw (1) to (2.center);
	\end{pgfonlayer}
\end{tikzpicture}
.\label{eq:step1state}\end{equation}
 In particular, we are using the fact that a linear map $L:\mathcal{L}(\mathds{R},V)\to \mathcal{L}(\mathds{R},V')$ can always be uniquely represented by a linear map $l:V\to V'$ by exploiting the fact that $\mathcal{L}(\mathds{R},V) \cong V$.    The fact that $\chi_B$ is the unique linear map satisfying Eq.~\eqref{eq:step1state}, means that there is no possibility for making other choices for the $\chi_A$ appearing in Eq.~\eqref{eq:repTrans}.

Similarly, $M$ on effects $\widetilde{E}_j$ has a unique linear extension and takes functionals on GPT states to functionals on $V_A$; in other words, $M$ is the adjoint of a process $\phi$ within $\RL$:
\begin{equation}\input{Diagrams/65_Proof17.tikz}\quad =\quad \input{Diagrams/66_Proof18.tikz},\label{eq:step1effect}\end{equation}
 In particular, we are using the fact that a linear map $L:\mathcal{L}(V,\mathds{R}) \to \mathcal{L}(V',\mathds{R})$ can always be uniquely represented by a linear map $l:V'\to V$ by exploiting the fact that $\mathcal{L}(V,\mathds{R})\cong V^*$ and that $\mathcal{L}(V^*,V'^*)\cong \mathcal{L}(V',V)$. 
Combining this with Eq.~\eqref{eq:step1conc}, we have
\begin{equation}
	\input{Diagrams/67_Proof19.tikz}\quad =\ \input{Diagrams/101.tikz}\quad =\quad \input{Diagrams/68_Proof20.tikz}.\label{eq:step2conc}
\end{equation}

All that remains is to show that $\chi_A$ and $\phi_A$ are inverses. Consider the special case that $\widetilde{T}$ is the identity, then Eq.~\eqref{eq:step2conc} becomes
\begin{equation}
	\input{Diagrams/70_Proof22.tikz}\quad =\quad \input{Diagrams/71_Proof23.tikz}.
\end{equation}
Since $M$ is diagram-preserving, it maps identity to identity, and so this becomes
\begin{equation}
	\begin{tikzpicture}
	\begin{pgfonlayer}{nodelayer}
		\node [style=none] (0) at (0, 1) {};
		\node [style=none] (1) at (0, -1) {};
	\end{pgfonlayer}
	\begin{pgfonlayer}{edgelayer}
		\draw (0.center) to (1.center);
	\end{pgfonlayer}
\end{tikzpicture}\quad =\quad \input{Diagrams/71_Proof23.tikz}.\label{eq:step3a}
\end{equation}

Now consider a state $\widetilde{P}$ followed by an effect $\widetilde{E}$. This gives a probability, and since $M$ is empirically adequate it must preserve this probability:
\begin{equation}
	\begin{tikzpicture}
	\begin{pgfonlayer}{nodelayer}
		\node [style=copoint] (0) at (0, 0.75) {$\widetilde{E}$};
		\node [style=point] (1) at (0, -0.75) {$\widetilde{P}$};
	\end{pgfonlayer}
	\begin{pgfonlayer}{edgelayer}
		\draw [style=qWire] (0.center) to (1.center);
	\end{pgfonlayer}
\end{tikzpicture}
\quad =\quad \input{Diagrams/73b_oneM.tikz},
\end{equation}
and since $M$ is diagram-preserving,
\begin{equation}
	\input{Diagrams/73b_oneM.tikz}\quad =\quad \input{Diagrams/73_Proof25.tikz}.
\end{equation}
Combining this with Eqs.~\eqref{eq:step1state} and \eqref{eq:step1effect} gives
\begin{equation}
	\begin{tikzpicture}
	\begin{pgfonlayer}{nodelayer}
		\node [style=copoint] (0) at (0, 0.75) {$\widetilde{E}$};
		\node [style=point] (1) at (0, -0.75) {$\widetilde{P}$};
	\end{pgfonlayer}
	\begin{pgfonlayer}{edgelayer}
		\draw [style=qWire] (0.center) to (1.center);
	\end{pgfonlayer}
\end{tikzpicture}
 = \input{Diagrams/74_Proof26.tikz}.
\end{equation}
Since this holds for all $\widetilde{E}$ and $\widetilde{P}$, tomographic locality implies that the $\widetilde{E}$ span $A^*$ and the $\widetilde{P}$ span $A$, and we have that 
\begin{equation}
	\begin{tikzpicture}
	\begin{pgfonlayer}{nodelayer}
		\node [style=none] (0) at (0, 1) {};
		\node [style=none] (1) at (0, -1) {};
		\node [style=right label] (2) at (0, -0.75) {$A$};
	\end{pgfonlayer}
	\begin{pgfonlayer}{edgelayer}
		\draw [qWire] (0.center) to (1.center);
	\end{pgfonlayer}
\end{tikzpicture}
 = \input{Diagrams/76_Proof28.tikz}.
\end{equation}
Combining this with Eq.~\eqref{eq:step3a} gives that $\chi$ and $\phi$ are inverses of each other.  Hence, we can write that $\phi_A=\chi_A^{-1}$ and so rewrite Eq.~\eqref{eq:step2conc} as
\begin{equation}
	\input{Diagrams/67_Proof19.tikz}\quad =\quad \input{Diagrams/68_Proof20.tikz}\quad = \quad \input{Diagrams/42_QuasiRep2.tikz},
\end{equation}
which completes the proof. 
\end{proof}

\section{Completing the proof of Proposition~\ref{thm:NCOMandOM}} \label{comptheproof}
The key argument required to establish Prop~\ref{thm:NCOMandOM} was given just after the proposition itself, but we now complete the proof.

We now prove that $\xiNC :=   \widetilde{\xi} \circ \sim$ is indeed a valid ontological model of an operational theory if $\widetilde{\xi}$ is a valid ontological model of a GPT. To do so, we show that each of the three properties (enumerated in Definition~\ref{defnontop}) that $\xiNC$ should satisfy is implied by the corresponding property (enumerated in Definition~\ref{defnontgpt}) that $\widetilde{\xi}$ is assumed to satisfy by virtue of being an ontological model of a GPT.

First, recall that we assumed that all deterministic effects in the operational theory are operationally equivalent. Hence, the map $\sim$ will take any such deterministic effect to the unique deterministic effect in the GPT, which (by property 1 of Definition~\ref{defnontgpt}) must be represented by the unit vector $\bf{1}$. Hence, $\xiNC$ represents all deterministic effects in the operational theory appropriately, namely as the unit vector $\bf{1}$.

Second, recall that $\sim$ preserves the operational predictions of the operational theory; hence, the fact that (by property 2 of Definition~\ref{defnontgpt}) $\widetilde{\xi}$ preserves the operational predictions of the GPT implies that $\xiNC :=   \widetilde{\xi} \circ \sim$ preserves the operational predictions of the operational theory.

Third, recall that if, in the operational theory, $P_1$ is a procedure that is a mixture of $P_2$ and $P_3$ with weights $\omega$ and $1-\omega$, then it follows that under $\sim$, one has
\beq
\widetilde{P_1} = \omega \widetilde{P_2} + (1-\omega) \widetilde{P_3}.
\eeq
Hence, the fact that (by property 3 of Definition~\ref{defnontgpt}) the representations of these three processes under $\widetilde{\xi}$ satisfy
 \beq
\begin{tikzpicture}
	\begin{pgfonlayer}{nodelayer}
		\node [style=none] (0) at (0, 0) {$
\widetilde{P_1}$};
		\node [style=none] (1) at (-0.5, 0.5) {};
		\node [style=none] (2) at (0.5, 0.5) {};
		\node [style=none] (3) at (0.5, -0.5) {};
		\node [style=none] (4) at (-0.5, -0.5) {};
		\node [style=none] (5) at (0, 0.5) {};
		\node [style=none] (6) at (0, 1.5) {};
		\node [style=none] (7) at (0, -0.5) {};
		\node [style=none] (8) at (0, -1.5) {};
		\node [style=none] (13) at (1.25, -1.25) {\tiny $\widetilde{\xi}$};
		\node [style=none] (14) at (1.5, -1.5) {};
		\node [style=none] (15) at (1.5, 1.5) {};
		\node [style=none] (16) at (-1.5, 1.5) {};
		\node [style=none] (17) at (-1.5, -1.5) {};
		\node [style=none] (18) at (0, 2) {};
		\node [style=none] (19) at (0, -2) {};
	\end{pgfonlayer}
	\begin{pgfonlayer}{edgelayer}
			\filldraw[fill=Red!20,draw=Red!40](14.center) to (15.center) to (16.center) to (17.center) to cycle;
		\filldraw[fill=white,draw=black] (1.center) to (2.center) to (3.center) to (4.center) to cycle;
		\draw [qWire] (5.center) to (6.center);
		\draw [qWire] (8.center) to (7.center);
		\draw (18.center) to (6.center);
		\draw (8.center) to (19.center);
	\end{pgfonlayer}
\end{tikzpicture}
\ = \omega\
\begin{tikzpicture}
	\begin{pgfonlayer}{nodelayer}
		\node [style=none] (0) at (0, 0) {$
\widetilde{P_2}$};
		\node [style=none] (1) at (-0.5, 0.5) {};
		\node [style=none] (2) at (0.5, 0.5) {};
		\node [style=none] (3) at (0.5, -0.5) {};
		\node [style=none] (4) at (-0.5, -0.5) {};
		\node [style=none] (5) at (0, 0.5) {};
		\node [style=none] (6) at (0, 1.5) {};
		\node [style=none] (7) at (0, -0.5) {};
		\node [style=none] (8) at (0, -1.5) {};
		\node [style=none] (13) at (1.25, -1.25) {\tiny $\widetilde{\xi}$};
		\node [style=none] (14) at (1.5, -1.5) {};
		\node [style=none] (15) at (1.5, 1.5) {};
		\node [style=none] (16) at (-1.5, 1.5) {};
		\node [style=none] (17) at (-1.5, -1.5) {};
		\node [style=none] (18) at (0, 2) {};
		\node [style=none] (19) at (0, -2) {};
	\end{pgfonlayer}
	\begin{pgfonlayer}{edgelayer}
			\filldraw[fill=Red!20,draw=Red!40](14.center) to (15.center) to (16.center) to (17.center) to cycle;
		\filldraw[fill=white,draw=black] (1.center) to (2.center) to (3.center) to (4.center) to cycle;
		\draw [qWire] (5.center) to (6.center);
		\draw [qWire] (8.center) to (7.center);
		\draw (18.center) to (6.center);
		\draw (8.center) to (19.center);
	\end{pgfonlayer}
\end{tikzpicture}
\ +(1-\omega)\
\begin{tikzpicture}
	\begin{pgfonlayer}{nodelayer}
		\node [style=none] (0) at (0, 0) {$
\widetilde{P_3}$};
		\node [style=none] (1) at (-0.5, 0.5) {};
		\node [style=none] (2) at (0.5, 0.5) {};
		\node [style=none] (3) at (0.5, -0.5) {};
		\node [style=none] (4) at (-0.5, -0.5) {};
		\node [style=none] (5) at (0, 0.5) {};
		\node [style=none] (6) at (0, 1.5) {};
		\node [style=none] (7) at (0, -0.5) {};
		\node [style=none] (8) at (0, -1.5) {};
		\node [style=none] (13) at (1.25, -1.25) {\tiny $\widetilde{\xi}$};
		\node [style=none] (14) at (1.5, -1.5) {};
		\node [style=none] (15) at (1.5, 1.5) {};
		\node [style=none] (16) at (-1.5, 1.5) {};
		\node [style=none] (17) at (-1.5, -1.5) {};
		\node [style=none] (18) at (0, 2) {};
		\node [style=none] (19) at (0, -2) {};
	\end{pgfonlayer}
	\begin{pgfonlayer}{edgelayer}
			\filldraw[fill=Red!20,draw=Red!40](14.center) to (15.center) to (16.center) to (17.center) to cycle;
		\filldraw[fill=white,draw=black] (1.center) to (2.center) to (3.center) to (4.center) to cycle;
		\draw [qWire] (5.center) to (6.center);
		\draw [qWire] (8.center) to (7.center);
		\draw (18.center) to (6.center);
		\draw (8.center) to (19.center);
	\end{pgfonlayer}
\end{tikzpicture}
\eeq
implies that the representations of $P_1$, $P_2$, and $P_3$ satisfy
\beq
\begin{tikzpicture}
	\begin{pgfonlayer}{nodelayer}
		\node [style=none] (0) at (0, 0) {$
{P_1}$};
		\node [style=none] (1) at (-0.5, 0.5) {};
		\node [style=none] (2) at (0.5, 0.5) {};
		\node [style=none] (3) at (0.5, -0.5) {};
		\node [style=none] (4) at (-0.5, -0.5) {};
		\node [style=none] (5) at (0, 0.5) {};
		\node [style=none] (6) at (0, 1.5) {};
		\node [style=none] (7) at (0, -0.5) {};
		\node [style=none] (8) at (0, -1.5) {};
		\node [style=none] (13) at (1.25, -1.25) {\tiny ${\xiNC}$};
		\node [style=none] (14) at (1.5, -1.5) {};
		\node [style=none] (15) at (1.5, 1.5) {};
		\node [style=none] (16) at (-1.5, 1.5) {};
		\node [style=none] (17) at (-1.5, -1.5) {};
		\node [style=none] (18) at (0, 2) {};
		\node [style=none] (19) at (0, -2) {};
	\end{pgfonlayer}
	\begin{pgfonlayer}{edgelayer}
			\filldraw[fill=black!30!BurntOrange!30,draw=black!40!BurntOrange!40](14.center) to (15.center) to (16.center) to (17.center) to cycle;
		\filldraw[fill=white,draw=black] (1.center) to (2.center) to (3.center) to (4.center) to cycle;
		\draw [qWire] (5.center) to (6.center);
		\draw [qWire] (8.center) to (7.center);
		\draw (18.center) to (6.center);
		\draw (8.center) to (19.center);
	\end{pgfonlayer}
\end{tikzpicture}
\ = \omega\
\begin{tikzpicture}
	\begin{pgfonlayer}{nodelayer}
		\node [style=none] (0) at (0, 0) {$
{P_2}$};
		\node [style=none] (1) at (-0.5, 0.5) {};
		\node [style=none] (2) at (0.5, 0.5) {};
		\node [style=none] (3) at (0.5, -0.5) {};
		\node [style=none] (4) at (-0.5, -0.5) {};
		\node [style=none] (5) at (0, 0.5) {};
		\node [style=none] (6) at (0, 1.5) {};
		\node [style=none] (7) at (0, -0.5) {};
		\node [style=none] (8) at (0, -1.5) {};
		\node [style=none] (13) at (1.25, -1.25) {\tiny ${\xiNC}$};
		\node [style=none] (14) at (1.5, -1.5) {};
		\node [style=none] (15) at (1.5, 1.5) {};
		\node [style=none] (16) at (-1.5, 1.5) {};
		\node [style=none] (17) at (-1.5, -1.5) {};
		\node [style=none] (18) at (0, 2) {};
		\node [style=none] (19) at (0, -2) {};
	\end{pgfonlayer}
	\begin{pgfonlayer}{edgelayer}
			\filldraw[fill=black!30!BurntOrange!30,draw=black!40!BurntOrange!40](14.center) to (15.center) to (16.center) to (17.center) to cycle;
		\filldraw[fill=white,draw=black] (1.center) to (2.center) to (3.center) to (4.center) to cycle;
		\draw [qWire] (5.center) to (6.center);
		\draw [qWire] (8.center) to (7.center);
		\draw (18.center) to (6.center);
		\draw (8.center) to (19.center);
	\end{pgfonlayer}
\end{tikzpicture}
\ +(1-\omega)\
\begin{tikzpicture}
	\begin{pgfonlayer}{nodelayer}
		\node [style=none] (0) at (0, 0) {$
{P_3}$};
		\node [style=none] (1) at (-0.5, 0.5) {};
		\node [style=none] (2) at (0.5, 0.5) {};
		\node [style=none] (3) at (0.5, -0.5) {};
		\node [style=none] (4) at (-0.5, -0.5) {};
		\node [style=none] (5) at (0, 0.5) {};
		\node [style=none] (6) at (0, 1.5) {};
		\node [style=none] (7) at (0, -0.5) {};
		\node [style=none] (8) at (0, -1.5) {};
		\node [style=none] (13) at (1.25, -1.25) {\tiny ${\xiNC}$};
		\node [style=none] (14) at (1.5, -1.5) {};
		\node [style=none] (15) at (1.5, 1.5) {};
		\node [style=none] (16) at (-1.5, 1.5) {};
		\node [style=none] (17) at (-1.5, -1.5) {};
		\node [style=none] (18) at (0, 2) {};
		\node [style=none] (19) at (0, -2) {};
	\end{pgfonlayer}
	\begin{pgfonlayer}{edgelayer}
			\filldraw[fill=black!30!BurntOrange!30,draw=black!40!BurntOrange!40](14.center) to (15.center) to (16.center) to (17.center) to cycle;
		\filldraw[fill=white,draw=black] (1.center) to (2.center) to (3.center) to (4.center) to cycle;
		\draw [qWire] (5.center) to (6.center);
		\draw [qWire] (8.center) to (7.center);
		\draw (18.center) to (6.center);
		\draw (8.center) to (19.center);
	\end{pgfonlayer}
\end{tikzpicture}.
\eeq

Hence $\xiNC$ satisfies all the properties of an ontological model of an operational theory.

Conversely, we prove that $\widetilde{\xi}:=\xiNC \circ C$ is a valid ontological model of a GPT if $\xiNC$ is a valid noncontextual ontological model of an operational theory. To do so, we show that each of the three properties (enumerated in Definition~\ref{defnontgpt}) that $\widetilde{\xi}$  should satisfy is implied by the corresponding property (enumerated in Definition~\ref{defnontop}) that $\xiNC$ is assumed to satisfy by virtue of being an ontological model of a GPT.

First, consider the unique deterministic effect in the GPT. Applying $C$ to this process yields one of the many deterministic effects in the operational theory. Because (by property 1 of Definition~\ref{defnontop}) $\xiNC$ maps every one of these to the unit vector $\bf{1}$, it follows that $\widetilde{\xi}:=\xiNC \circ C$ maps the unique deterministic effect to the unit vector $\bf{1}$.

Second, recall that the context of a process is irrelevant for the operational predictions it makes, and that consequently, the map $C$ preserves the operational predictions. Given that (by property 2 of Definition~\ref{defnontop}) $\xiNC$ preserves the operational predictions, $\widetilde{\xi}:=\xiNC \circ C$ also preserves the operational predictions.

Third, consider three processes $\widetilde{P}_1$, $\widetilde{P}_2$, and $\widetilde{P}_3$ such that $\widetilde{P}_1 = \omega \widetilde{P}_2 + (1-\omega)\widetilde{P}_3$ in the GPT. Under $C$, one has processes
$C(\widetilde{P}_1)=(\widetilde{P}_1,c_1)$, $C(\widetilde{P}_2)=(\widetilde{P}_1,c_2)$, and $C(\widetilde{P}_3)=(\widetilde{P}_1,c_3)$ in the operational theory, where $c_i$ are arbitrary contexts specified by the map $C$. The fact that $\widetilde{P}_1 = \omega \widetilde{P}_2 + (1-\omega)\widetilde{P}_3$ implies that $C(\widetilde{P}_1)$ is operationally equivalent to the effective procedure $P_{\rm mix}$ defined as the mixture of $C(\widetilde{P}_2)$
and $C(\widetilde{P}_3)$
with weights $\omega$ and $1-\omega$, respectively. ($C(\widetilde{P}_1)$ may not actually {\em be} this mixture, depending on its context $c_i$, which depends on one's choice of $C$.) By property 3 of Definition~\ref{defnontop}, $\xiNC$ must satisfy
\beq
\begin{tikzpicture}
	\begin{pgfonlayer}{nodelayer}
		\node [style=small box] (0) at (0, 0) {$
P_{mix}$};
		\node [style=none] (1) at (-0.5, 0.5) {};
		\node [style=none] (2) at (0.5, 0.5) {};
		\node [style=none] (3) at (0.5, -0.5) {};
		\node [style=none] (4) at (-0.5, -0.5) {};
		\node [style=none] (5) at (0, 0.5) {};
		\node [style=none] (6) at (0, 1.5) {};
		\node [style=none] (7) at (0, -0.5) {};
		\node [style=none] (8) at (0, -1.5) {};
		\node [style=none] (13) at (1.25, -1.25) {\tiny ${\xiNC}$};
		\node [style=none] (14) at (1.5, -1.5) {};
		\node [style=none] (15) at (1.5, 1.5) {};
		\node [style=none] (16) at (-1.5, 1.5) {};
		\node [style=none] (17) at (-1.5, -1.5) {};
		\node [style=none] (18) at (0, 2) {};
		\node [style=none] (19) at (0, -2) {};
	\end{pgfonlayer}
	\begin{pgfonlayer}{edgelayer}
			\filldraw[fill=black!30!BurntOrange!30,draw=black!40!BurntOrange!40](14.center) to (15.center) to (16.center) to (17.center) to cycle;
		\draw [qWire] (5.center) to (6.center);
		\draw [qWire] (8.center) to (7.center);
		\draw (18.center) to (6.center);
		\draw (8.center) to (19.center);
	\end{pgfonlayer}
\end{tikzpicture}
\ = \omega\
\begin{tikzpicture}
	\begin{pgfonlayer}{nodelayer}
		\node [style=small box] (0) at (0, 0) {$
C(\widetilde{P_2})$};
		\node [style=none] (1) at (-0.5, 0.5) {};
		\node [style=none] (2) at (0.5, 0.5) {};
		\node [style=none] (3) at (0.5, -0.5) {};
		\node [style=none] (4) at (-0.5, -0.5) {};
		\node [style=none] (5) at (0, 0.5) {};
		\node [style=none] (6) at (0, 1.5) {};
		\node [style=none] (7) at (0, -0.5) {};
		\node [style=none] (8) at (0, -1.5) {};
		\node [style=none] (13) at (1.25, -1.25) {\tiny ${\xiNC}$};
		\node [style=none] (14) at (1.5, -1.5) {};
		\node [style=none] (15) at (1.5, 1.5) {};
		\node [style=none] (16) at (-1.5, 1.5) {};
		\node [style=none] (17) at (-1.5, -1.5) {};
		\node [style=none] (18) at (0, 2) {};
		\node [style=none] (19) at (0, -2) {};
	\end{pgfonlayer}
	\begin{pgfonlayer}{edgelayer}
			\filldraw[fill=black!30!BurntOrange!30,draw=black!40!BurntOrange!40](14.center) to (15.center) to (16.center) to (17.center) to cycle;
		\draw [qWire] (5.center) to (6.center);
		\draw [qWire] (8.center) to (7.center);
		\draw (18.center) to (6.center);
		\draw (8.center) to (19.center);
	\end{pgfonlayer}
\end{tikzpicture}
\ +(1-\omega)\
\begin{tikzpicture}
	\begin{pgfonlayer}{nodelayer}
		\node [style=small box] (0) at (0, 0) {$
C(\widetilde{P_3})$};
		\node [style=none] (1) at (-0.5, 0.5) {};
		\node [style=none] (2) at (0.5, 0.5) {};
		\node [style=none] (3) at (0.5, -0.5) {};
		\node [style=none] (4) at (-0.5, -0.5) {};
		\node [style=none] (5) at (0, 0.5) {};
		\node [style=none] (6) at (0, 1.5) {};
		\node [style=none] (7) at (0, -0.5) {};
		\node [style=none] (8) at (0, -1.5) {};
		\node [style=none] (13) at (1.25, -1.25) {\tiny ${\xiNC}$};
		\node [style=none] (14) at (1.5, -1.5) {};
		\node [style=none] (15) at (1.5, 1.5) {};
		\node [style=none] (16) at (-1.5, 1.5) {};
		\node [style=none] (17) at (-1.5, -1.5) {};
		\node [style=none] (18) at (0, 2) {};
		\node [style=none] (19) at (0, -2) {};
	\end{pgfonlayer}
	\begin{pgfonlayer}{edgelayer}
			\filldraw[fill=black!30!BurntOrange!30,draw=black!40!BurntOrange!40](14.center) to (15.center) to (16.center) to (17.center) to cycle;
		\draw [qWire] (5.center) to (6.center);
		\draw [qWire] (8.center) to (7.center);
		\draw (18.center) to (6.center);
		\draw (8.center) to (19.center);
	\end{pgfonlayer}
\end{tikzpicture}.
\eeq
But since $\xiNC$ is a noncontextual model and since $C(\widetilde{P}_1)$ is operationally equivalent to $P_{\rm mix}$, it follows that
\beq
\begin{tikzpicture}
	\begin{pgfonlayer}{nodelayer}
		\node [style=small box] (0) at (0, 0) {$
C(\widetilde{P}_1)$};
		\node [style=none] (1) at (-0.5, 0.5) {};
		\node [style=none] (2) at (0.5, 0.5) {};
		\node [style=none] (3) at (0.5, -0.5) {};
		\node [style=none] (4) at (-0.5, -0.5) {};
		\node [style=none] (5) at (0, 0.5) {};
		\node [style=none] (6) at (0, 1.5) {};
		\node [style=none] (7) at (0, -0.5) {};
		\node [style=none] (8) at (0, -1.5) {};
		\node [style=none] (13) at (1.25, -1.25) {\tiny ${\xiNC}$};
		\node [style=none] (14) at (1.5, -1.5) {};
		\node [style=none] (15) at (1.5, 1.5) {};
		\node [style=none] (16) at (-1.5, 1.5) {};
		\node [style=none] (17) at (-1.5, -1.5) {};
		\node [style=none] (18) at (0, 2) {};
		\node [style=none] (19) at (0, -2) {};
	\end{pgfonlayer}
	\begin{pgfonlayer}{edgelayer}
			\filldraw[fill=black!30!BurntOrange!30,draw=black!40!BurntOrange!40](14.center) to (15.center) to (16.center) to (17.center) to cycle;
		\draw [qWire] (5.center) to (6.center);
		\draw [qWire] (8.center) to (7.center);
		\draw (18.center) to (6.center);
		\draw (8.center) to (19.center);
	\end{pgfonlayer}
\end{tikzpicture}
\ = \omega\
\begin{tikzpicture}
	\begin{pgfonlayer}{nodelayer}
		\node [style=small box] (0) at (0, 0) {$
C(\widetilde{P_2})$};
		\node [style=none] (1) at (-0.5, 0.5) {};
		\node [style=none] (2) at (0.5, 0.5) {};
		\node [style=none] (3) at (0.5, -0.5) {};
		\node [style=none] (4) at (-0.5, -0.5) {};
		\node [style=none] (5) at (0, 0.5) {};
		\node [style=none] (6) at (0, 1.5) {};
		\node [style=none] (7) at (0, -0.5) {};
		\node [style=none] (8) at (0, -1.5) {};
		\node [style=none] (13) at (1.25, -1.25) {\tiny ${\xiNC}$};
		\node [style=none] (14) at (1.5, -1.5) {};
		\node [style=none] (15) at (1.5, 1.5) {};
		\node [style=none] (16) at (-1.5, 1.5) {};
		\node [style=none] (17) at (-1.5, -1.5) {};
		\node [style=none] (18) at (0, 2) {};
		\node [style=none] (19) at (0, -2) {};
	\end{pgfonlayer}
	\begin{pgfonlayer}{edgelayer}
			\filldraw[fill=black!30!BurntOrange!30,draw=black!40!BurntOrange!40](14.center) to (15.center) to (16.center) to (17.center) to cycle;
		\draw [qWire] (5.center) to (6.center);
		\draw [qWire] (8.center) to (7.center);
		\draw (18.center) to (6.center);
		\draw (8.center) to (19.center);
	\end{pgfonlayer}
\end{tikzpicture}
\ +(1-\omega)\
\begin{tikzpicture}
	\begin{pgfonlayer}{nodelayer}
		\node [style=small box] (0) at (0, 0) {$
C(\widetilde{P_3})$};
		\node [style=none] (1) at (-0.5, 0.5) {};
		\node [style=none] (2) at (0.5, 0.5) {};
		\node [style=none] (3) at (0.5, -0.5) {};
		\node [style=none] (4) at (-0.5, -0.5) {};
		\node [style=none] (5) at (0, 0.5) {};
		\node [style=none] (6) at (0, 1.5) {};
		\node [style=none] (7) at (0, -0.5) {};
		\node [style=none] (8) at (0, -1.5) {};
		\node [style=none] (13) at (1.25, -1.25) {\tiny ${\xiNC}$};
		\node [style=none] (14) at (1.5, -1.5) {};
		\node [style=none] (15) at (1.5, 1.5) {};
		\node [style=none] (16) at (-1.5, 1.5) {};
		\node [style=none] (17) at (-1.5, -1.5) {};
		\node [style=none] (18) at (0, 2) {};
		\node [style=none] (19) at (0, -2) {};
	\end{pgfonlayer}
	\begin{pgfonlayer}{edgelayer}
			\filldraw[fill=black!30!BurntOrange!30,draw=black!40!BurntOrange!40](14.center) to (15.center) to (16.center) to (17.center) to cycle;
		\draw [qWire] (5.center) to (6.center);
		\draw [qWire] (8.center) to (7.center);
		\draw (18.center) to (6.center);
		\draw (8.center) to (19.center);
	\end{pgfonlayer}
\end{tikzpicture}.
\eeq
Hence we see that $\widetilde{\xi}:=\xiNC \circ C$ satisfies property 3 of Definition~\ref{defnontgpt}, as required.

\section{Proof of Proposition~\ref{prop:structureOM}} \label{propstructureOMproof}

\proof
Since the ontological model $\widetilde{\xi}$ satisfies the requirements of Proposition~\ref{qrepngptstruct} we immediately obtain Eqs.~\eqref{eq:ontchi} and \eqref{eq:ontchidiscard}.

We however also obtain additional constraints arising from the fact that the codomain of $\widetilde{\xi}$ is  $\mathbf{SubStoch}$ rather than $\mathbf{QuasiSubStoch}$. This additional constraint can be viewed as a set of positivity conditions as we will now explain.

 The notion of positivity we require is defined for linear maps between ordered vector spaces. A positive cone $V^+$ for a real vector space $V$ defines an ordering:
 \beq
 v\leq v' \ \ \iff \ \ \exists w \in V^+ \text{ s.t. } v + w = v'.
 \eeq 
 Given two such ordered vector spaces, $(V,V^+)$ and $(W,W^+)$, maps between these ordered vector spaces are  linear maps on the underlying vector spaces $L:V\to W$ and are said to be positive if and only if:
 \beq
 L(V^+)\subseteq W^+.
 \eeq

Now, the question is: what are the relevant ordered vector spaces which we want to consider here?

 In the GPT $\widetilde{\Op}$, we can define an ordered vector space for each pair of systems $(A,B)$ as the vector space spanned by the transformations from $A$ to $B$, which we denote by $\mathsf{Span}[\widetilde{\Op}_A^B]$. The positive cone is defined by the vectors in this space which can be expressed as a positive linear combination of vectors in $\widetilde{\Op}_A^B$:
 \beq
 \mathsf{Cone}[\widetilde{\Op}_A^B]:= \left\{\sum_i r_i \widetilde{f_i} \middle|  r_i \in \mathds{R}^+, \ \widetilde{f_i} \in \widetilde{\Op}_A^B \right\}.
 \eeq
Then it is clear that, $( \mathsf{Span}[\widetilde{\Op}_A^B], \mathsf{Cone}[\widetilde{\Op}_A^B])$ defines an ordered vector space.

 Similarly, in $\mathbf{SubStoch}$ we can define an ordered vector space for each pair of systems $(\mathds{R}^\Lambda, \mathds{R}^{\Lambda'})$ as the vector space spanned by the substochastic maps from $\mathds{R}^\Lambda$ to $\mathds{R}^{\Lambda'}$, which we denote by $\mathsf{Span}[\mathbf{SubStoch}_\Lambda^{\Lambda'}]$ which is a subspace of the real vector space of linear maps from $\mathds{R}^\Lambda$ to $\mathds{R}^{\Lambda'}$. The positive cone is defined by the vectors in this space which can be expressed as a positive linear combination of vectors in $\mathbf{SubStoch}_\Lambda^{\Lambda'}$:
 \beq
 \mathsf{Cone}[\mathbf{SubStoch}_\Lambda^{\Lambda'}] := \left\{\sum_i r_i s_i \middle| r_i \in \mathds{R}^+, \ \ s_i \in \mathbf{SubStoch}_\Lambda^{\Lambda'}\right\}.
\eeq 
Then it is clear that, $( \mathsf{Span}[\mathbf{SubStoch}_\Lambda^{\Lambda'}], \mathsf{Cone}[\mathbf{SubStoch}_\Lambda^{\Lambda'})$ defines an ordered vector space.

Then, for a pair of GPT systems $(A,B)$ the action of the map $\widetilde{\xi}$ from $\widetilde{\Op}_A^B$ to $\mathbf{SubStoch}_{\Lambda_A}^{\Lambda_B}$ can be extended to a linear map from $\mathsf{Span}[\widetilde{\Op}_A^B]$ to $\mathsf{Span}[\mathbf{SubStoch}_{\Lambda_A}^{\Lambda_B}]$. Moreover, it is clear that this will be a positive linear map---in the sense that we defined above---as it maps the positive cone $\mathsf{Cone}[\widetilde{\Op}_A^B]$  into the positive cone  $\mathsf{Cone}[\mathbf{SubStoch}_{\Lambda_A}^{\Lambda_B}]$, that is:
\begin{align}
\widetilde{\xi}\left(\mathsf{Cone}[\widetilde{\Op}_A^B]\right) &= \widetilde{\xi}\left( \left\{\sum_i r_i \widetilde{f_i} \middle|  r_i \in \mathds{R}^+, \ \widetilde{f_i} \in \widetilde{\Op}_A^B \right\}\right)\\
&= \left\{\sum_i r_i \widetilde{\xi}(\widetilde{f_i}) \middle|  r_i \in \mathds{R}^+, \ \widetilde{f_i} \in \widetilde{\Op}_A^B \right\}\\
&= \left\{\sum_i r_i s_i \middle|  r_i \in \mathds{R}^+, \ s_i \in \widetilde{\xi}(\widetilde{\Op}_A^B) \right\} \\
&\subseteq  \left\{\sum_i r_i s_i \middle|  r_i \in \mathds{R}^+, \ s_i \in \mathbf{SubStoch}_{\Lambda_A}^{\Lambda_B} \right\} \\
&= \mathsf{Cone}[\mathbf{SubStoch}_{\Lambda_A}^{\Lambda_B}].
\end{align}
In summary, for each pair $(A,B)$, the representation map $\widetilde{\xi}$ defines a positive linear map from $(\mathsf{Span}[\widetilde{\Op}_A^B], \mathsf{Cone}[\widetilde{\Op}_A^B])$ to $( \mathsf{Span}[\mathbf{SubStoch}_\Lambda^{\Lambda'}], \mathsf{Cone}[\mathbf{SubStoch}_\Lambda^{\Lambda'})$.
\endproof

This positivity condition is all fairly abstract so let us consider some more concrete consequences of this result. If, rather than considering transformations from one GPT system to another, we consider just the states of a single system $A$ then everything simplifies considerably. The vector space we consider in the domain is simply the vector space spanned by the GPT state space, and the positive cone is then just the standard cone of GPT states. The vector space we consider in the codomain is simply the vector space $\mathds{R}^{\Lambda_A}$ with positive cone given by the cone of unnormalised probability distributions. Moreover, the linearly extended action of $\widetilde{\xi}$ is nothing but the linear map $\chi_A$ so we find that $\chi_A$ must be a positive map in the sense defined above.

Similarly, if we consider the contravariant action of $\chi^{-1}_A$ on the space of GPT effects (that is, by composing the effect onto the outgoing wire of $\chi^{-1}_A$) then we arrive at a similar result. Here we find that the contravariant action of $\chi^{-1}_A$ is a positive linear map from the dual of the GPT vector space ordered by the effect cone to the dual of $\mathds{R}^{\Lambda_A}$ ordered by the cone of response functions.

\section{Proofs for preliminaries}

\subsection{Proof that quotienting is diagram-preserving} \label{app:quotDP}

In order to see that the quotienting map is diagram-preserving, we must first define what it means for processes in the quotiented theory to be composed. That is, given a suitable pair of equivalence class $\widetilde{T}$ and $\widetilde{R}$, we must define $\widetilde{R}\circ \widetilde{T}$ (assuming that the relevant type matching constraint is satisfied) and  $\widetilde{R}\otimes \widetilde{T}$. We define these via composition of some choice of representative elements, $r \in \widetilde{R}$ and $t\in \widetilde{T}$, for each equivalence class, as
\beq
\input{Diagrams/n1-seqComp1.tikz}\ \ :=\ \ \input{Diagrams/n2-seqComp2.tikz} \qquad \text{and}\qquad \input{Diagrams/n3-parComp1.tikz}\ \ :=\ \ \input{Diagrams/n4-parComp2.tikz}. 
\eeq
For this to be well defined, it must be independent of the choices of representatives, i.e. for any $t_1, t_2 \in \widetilde{T}$ and $r_1, r_2\in \widetilde{R}$, one has
\beq 
\begin{tikzpicture}
	\begin{pgfonlayer}{nodelayer}
		\node [style=none] (0) at (0, 0) {$\widetilde{r_1\circ t_1}$};
		\node [style=none] (2) at (0, -1.25) {};
		\node [style=none] (3) at (0, 1.25) {};
		\node [style=none] (4) at (-1.25, 0.75) {};
		\node [style=none] (5) at (1.25, 0.75) {};
		\node [style=none] (6) at (1.25, -0.75) {};
		\node [style=none] (7) at (-1.25, -0.75) {};
	\end{pgfonlayer}
	\begin{pgfonlayer}{edgelayer}
		\draw [qWire] (0.center) to (2.center);
		\draw [qWire] (3.center) to (0.center);
		\draw [fill=white] (6.center)
			 to (7.center)
			 to (4.center)
			 to (5.center)
			 to cycle;
	\end{pgfonlayer}
\end{tikzpicture}
\ \ = \ \ 
\begin{tikzpicture}
	\begin{pgfonlayer}{nodelayer}
		\node [style=none] (0) at (0, 0) {$\widetilde{r_2\circ t_2}$};
		\node [style=none] (2) at (0, -1.25) {};
		\node [style=none] (3) at (0, 1.25) {};
		\node [style=none] (4) at (-1.25, 0.75) {};
		\node [style=none] (5) at (1.25, 0.75) {};
		\node [style=none] (6) at (1.25, -0.75) {};
		\node [style=none] (7) at (-1.25, -0.75) {};
	\end{pgfonlayer}
	\begin{pgfonlayer}{edgelayer}
		\draw [qWire] (0.center) to (2.center);
		\draw [qWire] (3.center) to (0.center);
		\draw [fill=white] (6.center)
			 to (7.center)
			 to (4.center)
			 to (5.center)
			 to cycle;
	\end{pgfonlayer}
\end{tikzpicture}
\qquad\text{and}\qquad
\begin{tikzpicture}
	\begin{pgfonlayer}{nodelayer}
		\node [style=none] (7) at (-0.75, -0.75) {};
		\node [style=none] (8) at (-0.75, -1.5) {};
		\node [style=none] (10) at (0.75, -0.75) {};
		\node [style=none] (11) at (0.75, -1.5) {};
		\node [style=none] (12) at (-0.75, 1.5) {};
		\node [style=none] (13) at (-0.75, 0.75) {};
		\node [style=none] (14) at (0.75, 1.5) {};
		\node [style=none] (15) at (0.75, 0.75) {};
		\node [style=none] (16) at (0, 0) {$\widetilde{r_1\otimes t_1}$};
		\node [style=none] (17) at (-1.25, 0.75) {};
		\node [style=none] (18) at (-1.25, -0.75) {};
		\node [style=none] (19) at (1.25, -0.75) {};
		\node [style=none] (20) at (1.25, 0.75) {};
	\end{pgfonlayer}
	\begin{pgfonlayer}{edgelayer}
		\draw [qWire] (7.center) to (8.center);
		\draw [qWire] (10.center) to (11.center);
		\draw [qWire] (12.center) to (13.center);
		\draw [qWire] (14.center) to (15.center);
		\draw [fill=white] (18.center)
			 to (17.center)
			 to (20.center)
			 to (19.center)
			 to cycle;
	\end{pgfonlayer}
\end{tikzpicture}
\ \ = \ \ 
\begin{tikzpicture}
	\begin{pgfonlayer}{nodelayer}
		\node [style=none] (7) at (-0.75, -0.75) {};
		\node [style=none] (8) at (-0.75, -1.5) {};
		\node [style=none] (10) at (0.75, -0.75) {};
		\node [style=none] (11) at (0.75, -1.5) {};
		\node [style=none] (12) at (-0.75, 1.5) {};
		\node [style=none] (13) at (-0.75, 0.75) {};
		\node [style=none] (14) at (0.75, 1.5) {};
		\node [style=none] (15) at (0.75, 0.75) {};
		\node [style=none] (16) at (0, 0) {$\widetilde{r_2\otimes t_2}$};
		\node [style=none] (17) at (-1.25, 0.75) {};
		\node [style=none] (18) at (-1.25, -0.75) {};
		\node [style=none] (19) at (1.25, -0.75) {};
		\node [style=none] (20) at (1.25, 0.75) {};
	\end{pgfonlayer}
	\begin{pgfonlayer}{edgelayer}
		\draw [qWire] (7.center) to (8.center);
		\draw [qWire] (10.center) to (11.center);
		\draw [qWire] (12.center) to (13.center);
		\draw [qWire] (14.center) to (15.center);
		\draw [fill=white] (18.center)
			 to (17.center)
			 to (20.center)
			 to (19.center)
			 to cycle;
	\end{pgfonlayer}
\end{tikzpicture}
\eeq
or equivalently
\beq \label{LHEappone}
\input{Diagrams/n10-cong1.tikz}\ \  \sim \ \  \input{Diagrams/n11-cong2.tikz} \qquad \text{and} \qquad \input{Diagrams/n12-cong3.tikz}\ \  \sim \ \  \input{Diagrams/n14-cong4.tikz}.
\eeq
If this is the case, then the quotienting map is a structure-preserving equivalence relation, or \emph{congruence relation}, for the process theory.

It is straightforward to show that the first equality in Eq.~\eqref{LHEappone} is equivalent to the conditions
\beq \label{firstequivappone}
\forall\ r \quad \input{Diagrams/n19-congSeq1.tikz} \ \sim  \ \input{Diagrams/n20-congSeq2.tikz}
\quad \text{and}  \quad \forall\ t \quad 
\input{Diagrams/n21-congSeq3.tikz} \ \sim  \ \input{Diagrams/n22-congSeq4.tikz}
\ . 
\eeq
To verify the nontrivial direction of this equivalence, consider the special case of these where $r=r_1$ (in the first) and where $t=t_2$ (in the second); then, one has
\beq
\begin{tikzpicture}
	\begin{pgfonlayer}{nodelayer}
		\node [style=small box] (0) at (0, -0.75) {$t_1$};
		\node [style=small box] (1) at (0, 0.75) {$r_1$};
		\node [style=none] (2) at (0, -1.75) {};
		\node [style=none] (3) at (0, 1.75) {};
	\end{pgfonlayer}
	\begin{pgfonlayer}{edgelayer}
		\draw [qWire] (3.center) to (1);
		\draw [qWire] (1) to (0);
		\draw [qWire] (0) to (2.center);
	\end{pgfonlayer}
\end{tikzpicture}
\ \ \sim \ \ 
\begin{tikzpicture}
	\begin{pgfonlayer}{nodelayer}
		\node [style=small box] (0) at (0, -0.75) {$t_2$};
		\node [style=small box] (1) at (0, 0.75) {$r_1$};
		\node [style=none] (2) at (0, -1.75) {};
		\node [style=none] (3) at (0, 1.75) {};
	\end{pgfonlayer}
	\begin{pgfonlayer}{edgelayer}
		\draw [qWire] (3.center) to (1);
		\draw [qWire] (1) to (0);
		\draw [qWire] (0) to (2.center);
	\end{pgfonlayer}
\end{tikzpicture}
\ \ = \ \ 
\begin{tikzpicture}
	\begin{pgfonlayer}{nodelayer}
		\node [style=small box] (0) at (0, -0.75) {$t_2$};
		\node [style=small box] (1) at (0, 0.75) {$r_1$};
		\node [style=none] (2) at (0, -1.75) {};
		\node [style=none] (3) at (0, 1.75) {};
	\end{pgfonlayer}
	\begin{pgfonlayer}{edgelayer}
		\draw [qWire] (3.center) to (1);
		\draw [qWire] (1) to (0);
		\draw [qWire] (0) to (2.center);
	\end{pgfonlayer}
\end{tikzpicture}
\ \ \sim \ \
\begin{tikzpicture}
	\begin{pgfonlayer}{nodelayer}
		\node [style=small box] (0) at (0, -0.75) {$t_2$};
		\node [style=small box] (1) at (0, 0.75) {$r_2$};
		\node [style=none] (2) at (0, -1.75) {};
		\node [style=none] (3) at (0, 1.75) {};
	\end{pgfonlayer}
	\begin{pgfonlayer}{edgelayer}
		\draw [qWire] (3.center) to (1);
		\draw [qWire] (1) to (0);
		\draw [qWire] (0) to (2.center);
	\end{pgfonlayer}
\end{tikzpicture}
\eeq
Similarly, the second equality in Eq.~\eqref{LHEappone} is equivalent to the conditions
\beq
\forall\ r\quad \input{Diagrams/n15-congPar1.tikz} \ \sim  \ \input{Diagrams/n16-congPar2.tikz}
\quad \text{and} \quad \forall\ t \quad
\input{Diagrams/n17-congPar3.tikz} \ \sim  \ \input{Diagrams/n18-congPar4.tikz}\ .
\eeq
We now prove that the first of these four conditions (namely, the first equivalence in Eq.~\eqref{firstequivappone}) holds:
\begin{align}
\begin{tikzpicture}
	\begin{pgfonlayer}{nodelayer}
		\node [style=small box] (4) at (0, 0) {$t_1$};
		\node [style=none] (5) at (0, -1) {};
		\node [style=none] (6) at (0, 1) {};
	\end{pgfonlayer}
	\begin{pgfonlayer}{edgelayer}
		\draw [qWire] (4) to (5.center);
		\draw [qWire] (6.center) to (4);
	\end{pgfonlayer}
\end{tikzpicture}
\ \sim\ \begin{tikzpicture}
	\begin{pgfonlayer}{nodelayer}
		\node [style=small box] (4) at (0, 0) {$t_2$};
		\node [style=none] (5) at (0, -1) {};
		\node [style=none] (6) at (0, 1) {};
	\end{pgfonlayer}
	\begin{pgfonlayer}{edgelayer}
		\draw [qWire] (4) to (5.center);
		\draw [qWire] (6.center) to (4);
	\end{pgfonlayer}
\end{tikzpicture}
 \quad &\iff \quad \forall \tau \ \ \input{Diagrams/n25-proof3.tikz} \  = \  \input{Diagrams/n26-proof4.tikz} \\
                    & \implies \quad \forall \tau' \ \ \input{Diagrams/n27-proof5.tikz} \  = \  \input{Diagrams/n28-proof6.tikz} \\
                    & \iff \quad \input{Diagrams/n19-congSeq1.tikz} \ \sim  \ \input{Diagrams/n20-congSeq2.tikz}
\end{align}
where in the second step we are noting that
\beq
\input{Diagrams/n29-proof7.tikz}
\eeq
is an example of a tester $\tau$ for any $\tau'$ and $r$. 
The argument for the other three conditions is analogous.

This establishes that the notion of composition that we have defined is independent of the choice of representative elements, and so we can simply write 
\beq
\input{Diagrams/n1-seqComp1.tikz}\ \ =\ \ \input{Diagrams/n2-seqComp2.tikz}\ \ = \ \ 
\begin{tikzpicture}
	\begin{pgfonlayer}{nodelayer}
		\node [style=none] (0) at (0, 0) {$\widetilde{R\circ T}$};
		\node [style=none] (2) at (0, -1.25) {};
		\node [style=none] (3) at (0, 1.25) {};
		\node [style=none] (4) at (-1.25, 0.75) {};
		\node [style=none] (5) at (1.25, 0.75) {};
		\node [style=none] (6) at (1.25, -0.75) {};
		\node [style=none] (7) at (-1.25, -0.75) {};
	\end{pgfonlayer}
	\begin{pgfonlayer}{edgelayer}
		\draw [qWire] (0.center) to (2.center);
		\draw [qWire] (3.center) to (0.center);
		\draw [fill=white] (6.center)
			 to (7.center)
			 to (4.center)
			 to (5.center)
			 to cycle;
	\end{pgfonlayer}
\end{tikzpicture}
 \qquad \text{and}\qquad \input{Diagrams/n3-parComp1.tikz}\ \ =\ \ \input{Diagrams/n4-parComp2.tikz} \ \ = \ \ 
 \begin{tikzpicture}
	\begin{pgfonlayer}{nodelayer}
		\node [style=none] (7) at (-0.75, -0.75) {};
		\node [style=none] (8) at (-0.75, -1.5) {};
		\node [style=none] (10) at (0.75, -0.75) {};
		\node [style=none] (11) at (0.75, -1.5) {};
		\node [style=none] (12) at (-0.75, 1.5) {};
		\node [style=none] (13) at (-0.75, 0.75) {};
		\node [style=none] (14) at (0.75, 1.5) {};
		\node [style=none] (15) at (0.75, 0.75) {};
		\node [style=none] (16) at (0, 0) {$\widetilde{R\otimes T}$};
		\node [style=none] (17) at (-1.25, 0.75) {};
		\node [style=none] (18) at (-1.25, -0.75) {};
		\node [style=none] (19) at (1.25, -0.75) {};
		\node [style=none] (20) at (1.25, 0.75) {};
	\end{pgfonlayer}
	\begin{pgfonlayer}{edgelayer}
		\draw [qWire] (7.center) to (8.center);
		\draw [qWire] (10.center) to (11.center);
		\draw [qWire] (12.center) to (13.center);
		\draw [qWire] (14.center) to (15.center);
		\draw [fill=white] (18.center)
			 to (17.center)
			 to (20.center)
			 to (19.center)
			 to cycle;
	\end{pgfonlayer}
\end{tikzpicture}
\eeq
Now, recalling that
\beq
\input{Diagrams/n30-proofextra1.tikz}\ \ := \ \ \input{Diagrams/n31-proofextra2.tikz}
\eeq
it follows that the quotienting map is diagram-preserving:
\begin{align}
\input{Diagrams/n5-seqCompPres1.tikz}\ \  
&= \ \ \input{Diagrams/n6-seqCompPres3.tikz}\ \ = \ \
\begin{tikzpicture}
	\begin{pgfonlayer}{nodelayer}
		\node [style=none] (0) at (0, 0) {$\widetilde{R\circ T}$};
		\node [style=none] (2) at (0, -1.25) {};
		\node [style=none] (3) at (0, 1.25) {};
		\node [style=none] (4) at (-1.25, 0.75) {};
		\node [style=none] (5) at (1.25, 0.75) {};
		\node [style=none] (6) at (1.25, -0.75) {};
		\node [style=none] (7) at (-1.25, -0.75) {};
	\end{pgfonlayer}
	\begin{pgfonlayer}{edgelayer}
		\draw [qWire] (0.center) to (2.center);
		\draw [qWire] (3.center) to (0.center);
		\draw [fill=white] (6.center)
			 to (7.center)
			 to (4.center)
			 to (5.center)
			 to cycle;
	\end{pgfonlayer}
\end{tikzpicture}\ \ = \ \ \ \ \ \input{Diagrams/n1-seqComp1.tikz} \ \ \ \ \ \ = \ \ \ \ \ \ \input{Diagrams/n6-seqCompPres2.tikz}\ , \\
\input{Diagrams/n9-parCompPres1.tikz} \ \  &= \ \ \input{Diagrams/n8-parCompPres2.tikz} \ \ = \ \  \begin{tikzpicture}
	\begin{pgfonlayer}{nodelayer}
		\node [style=none] (7) at (-0.75, -0.75) {};
		\node [style=none] (8) at (-0.75, -1.5) {};
		\node [style=none] (10) at (0.75, -0.75) {};
		\node [style=none] (11) at (0.75, -1.5) {};
		\node [style=none] (12) at (-0.75, 1.5) {};
		\node [style=none] (13) at (-0.75, 0.75) {};
		\node [style=none] (14) at (0.75, 1.5) {};
		\node [style=none] (15) at (0.75, 0.75) {};
		\node [style=none] (16) at (0, 0) {$\widetilde{R\otimes T}$};
		\node [style=none] (17) at (-1.25, 0.75) {};
		\node [style=none] (18) at (-1.25, -0.75) {};
		\node [style=none] (19) at (1.25, -0.75) {};
		\node [style=none] (20) at (1.25, 0.75) {};
	\end{pgfonlayer}
	\begin{pgfonlayer}{edgelayer}
		\draw [qWire] (7.center) to (8.center);
		\draw [qWire] (10.center) to (11.center);
		\draw [qWire] (12.center) to (13.center);
		\draw [qWire] (14.center) to (15.center);
		\draw [fill=white] (18.center)
			 to (17.center)
			 to (20.center)
			 to (19.center)
			 to cycle;
	\end{pgfonlayer}
\end{tikzpicture} \ \ = \ \  \input{Diagrams/n3-parComp1.tikz} \ \ = \ \ \input{Diagrams/n7-parCompPres3.tikz}\ .
\end{align}

\subsection{Proof of Lemma~\ref{bilinearityprf}}\label{sec:bilinearity}
We now prove Lemma~\ref{bilinearityprf}, restated here:
\begin{lemma}
The operation $\smallsquare$ can be uniquely extended to a bilinear map
\beq
\smallsquare:\left(\mathds{R}^{m^{B\to C}},\mathds{R}^{m^{A\to B}}\right)\to \mathds{R}^{m^{A\to C}},
\eeq
and the operation $\smallboxtimes$ can be uniquely extended to  a bilinear map
\beq
\smallboxtimes : \left(\mathds{R}^{m^{A\to B}}, \mathds{R}^{m^{C\to D}}\right)\to \mathds{R}^{m^{{AC}\to{BD}}}.
\eeq
\end{lemma}

\proof
 
Here we show the proof for $\smallsquare$. The proof for $\smallboxtimes$ follows similarly.

To begin, note that the vectors $\mathbf{R}_{\widetilde{T}}$ with $\widetilde{T}:A\to B$ span the vector space $\mathds{R}^{m^{A\to B}}$, as we have taken $\mathcal{F}^{A\to B}$ to be a \emph{minimal} fiducial set of testers. Consequently, we can always (nonuniquely) write an arbitrary $\mathbf{U}\in\mathds{R}^{m^{B\to C}}$ as $\sum_i u_i \mathbf{R}_{\widetilde{T}_i}$ for some transformations $\widetilde{T}_i:B\to C$, $u_i \in \mathds{R}$, and can write an arbitrary $\mathbf{V}\in \mathds{R}^{m^{A\to B}}$ as $\sum_j v_j \mathbf{R}_{\widetilde{T}'_j}$ for some transformations $\widetilde{T}'_j:A\to B$, $v_j\in \mathds{R}$. Hence, we propose the linear extension be defined by $\mathbf{U}\smallsquare \mathbf{V}:= \sum_{ij} u_i v_j (\mathbf{R}_{\widetilde{T}_i} \smallsquare \mathbf{R}_{\widetilde{T'}_j})= \sum_{ij} u_i v_j \mathbf{R}_{\widetilde{T}_i\circ \widetilde{T'}_j}$. For this to be a valid definition, however, it must be the case that this is independent of the chosen decomposition of $\mathbf{U}$ and $\mathbf{V}$. We now show that this is indeed the case.

To begin, let us consider two distinct decompositions in the second argument of $\smallsquare$. That is, given
\begin{equation} \label{linconstraint2}
\sum_i u_i \mathbf{R}_{\widetilde{T}_i}= \mathbf{U}=\sum_j {u'}_j \mathbf{R}_{\widetilde{T'}_j},
\end{equation}
we want to show that 
\beq
\sum_i u_i (\mathbf{R}_{\widetilde{T}}\smallsquare\mathbf{R}_{\widetilde{T}_i}) = \sum_j {u'}_j (\mathbf{R}_{\widetilde{T}}\smallsquare\mathbf{R}_{\widetilde{T'}_j})
\eeq
for all $\widetilde{T}$. 

To begin, we use linearity of $E^{A\to B}$ (as defined in Section~\ref{Edefn}) to give us that 
\beq
\sum_i u_i \mathbf{K}_{\widetilde{T}_i}=\sum_j {u'}_j \mathbf{K}_{\widetilde{T'}_j}.
\eeq
Unpacking the definition of $\mathbf{K}$ gives us that 
 for all testers $\widetilde{\tau}$,
\beq \label{eq:proofII1}
\begin{tikzpicture}
	\begin{pgfonlayer}{nodelayer}
		\node [style=none] (17) at (6.5, 0.75) {};
		\node [style=none] (19) at (7, 0.75) {};
		\node [style=none] (20) at (8.25, 1.25) {};
		\node [style=none] (21) at (7.75, 0.75) {};
		\node [style=small box] (22) at (7, 0) {$\widetilde{T'}_j$};
		\node [style=none] (23) at (6.5, -0.75) {};
		\node [style=none] (25) at (7, -0.75) {};
		\node [style=none] (26) at (8.25, -1.25) {};
		\node [style=none] (27) at (6.5, 1.25) {};
		\node [style=none] (29) at (6.5, -1.25) {};
		\node [style=none] (30) at (7.75, -0.75) {};
		\node [style=none] (31) at (8, 0) {$\widetilde{\tau}$};
		\node [style=none] (34) at (5, 0) {$=\sum_j {u'}_j$};
		\node [style=none] (35) at (1.5, 0.75) {};
		\node [style=none] (36) at (2, 0.75) {};
		\node [style=none] (37) at (3.25, 1.25) {};
		\node [style=none] (38) at (2.75, 0.75) {};
		\node [style=small box] (39) at (2, 0) {$\widetilde{T}_i$};
		\node [style=none] (40) at (1.5, -0.75) {};
		\node [style=none] (41) at (2, -0.75) {};
		\node [style=none] (42) at (3.25, -1.25) {};
		\node [style=none] (43) at (1.5, 1.25) {};
		\node [style=none] (44) at (1.5, -1.25) {};
		\node [style=none] (45) at (2.75, -0.75) {};
		\node [style=none] (46) at (3, 0) {$\widetilde{\tau}$};
		\node [style=none] (47) at (0, 0) {$\sum_i u_i$};
	\end{pgfonlayer}
	\begin{pgfonlayer}{edgelayer}
		\draw [fill=white, draw=black] (29.center)
			 to (26.center)
			 to (20.center)
			 to (27.center)
			 to (17.center)
			 to (21.center)
			 to (30.center)
			 to (23.center)
			 to cycle;
		\draw [style=qWire] (22) to (19.center);
		\draw [style=qWire] (22) to (25.center);
		\draw [fill=white, draw=black] (44.center)
			 to (42.center)
			 to (37.center)
			 to (43.center)
			 to (35.center)
			 to (38.center)
			 to (45.center)
			 to (40.center)
			 to cycle;
		\draw [style=qWire] (39) to (36.center);
		\draw [style=qWire] (39) to (41.center);
	\end{pgfonlayer}
\end{tikzpicture}
.
\eeq
Now, define
\beq
\begin{tikzpicture}
	\begin{pgfonlayer}{nodelayer}
		\node [style=none] (3) at (-1, 1.25) {};
		\node [style=none] (4) at (-1, 1.75) {};
		\node [style=none] (5) at (1, 1.25) {};
		\node [style=none] (6) at (1, -1.25) {};
		\node [style=none] (7) at (2, 1.75) {};
		\node [style=none] (8) at (2, -1.75) {};
		\node [style=none] (9) at (-1, -1.75) {};
		\node [style=none] (10) at (-1, -1.25) {};
		\node [style=none] (11) at (1.5, 0) {$\widetilde{\tau}'$};
		\node [style=none] (12) at (0, -1.25) {};
		\node [style=none] (13) at (0, -0.75) {};
		\node [style=none] (14) at (0, 0.75) {};
		\node [style=none] (15) at (0, 1.25) {};
	\end{pgfonlayer}
	\begin{pgfonlayer}{edgelayer}
		\draw (4.center) to (3.center);
		\draw (3.center) to (5.center);
		\draw (5.center) to (6.center);
		\draw (6.center) to (10.center);
		\draw (10.center) to (9.center);
		\draw (9.center) to (8.center);
		\draw (8.center) to (7.center);
		\draw (7.center) to (4.center);
		\draw [qWire] (13.center) to (12.center);
		\draw [qWire] (15.center) to (14.center);
	\end{pgfonlayer}
\end{tikzpicture}
\ \ :=\ \
\begin{tikzpicture}
	\begin{pgfonlayer}{nodelayer}
		\node [style=small box] (0) at (0, 0.75) {$\widetilde{T}$};
		\node [style=none] (1) at (0, -0.25) {};
		\node [style=none] (2) at (0, 1.75) {};
		\node [style=none] (3) at (-1, 1.75) {};
		\node [style=none] (4) at (-1, 2.25) {};
		\node [style=none] (5) at (1, 1.75) {};
		\node [style=none] (6) at (1, -1.75) {};
		\node [style=none] (7) at (2, 2.25) {};
		\node [style=none] (8) at (2, -2.25) {};
		\node [style=none] (9) at (-1, -2.25) {};
		\node [style=none] (10) at (-1, -1.75) {};
		\node [style=none] (11) at (1.5, -0.25) {$\widetilde{\tau}$};
		\node [style=none] (12) at (0, -1.75) {};
		\node [style=none] (13) at (0, -1.25) {};
	\end{pgfonlayer}
	\begin{pgfonlayer}{edgelayer}
		\draw [qWire] (1.center) to (0);
		\draw [qWire] (0) to (2.center);
		\draw (4.center) to (3.center);
		\draw (3.center) to (5.center);
		\draw (5.center) to (6.center);
		\draw (6.center) to (10.center);
		\draw (10.center) to (9.center);
		\draw (9.center) to (8.center);
		\draw (8.center) to (7.center);
		\draw (7.center) to (4.center);
		\draw [qWire] (13.center) to (12.center);
	\end{pgfonlayer}
\end{tikzpicture}
\eeq
for some arbitrary $\widetilde{\tau}$ and transformation $\widetilde{T}$. Substituting tester $\widetilde{\tau}'$ into Eq.~\eqref{eq:proofII1}, we find that
\beq
\sum_i u_i \begin{tikzpicture}
	\begin{pgfonlayer}{nodelayer}
		\node [style=small box] (0) at (0, 0.75) {$\widetilde{T}$};
		\node [style=small box] (1) at (0, -0.75) {$\widetilde{T}_i$};
		\node [style=none] (2) at (0, 1.75) {};
		\node [style=none] (3) at (-1, 1.75) {};
		\node [style=none] (4) at (-1, 2.25) {};
		\node [style=none] (5) at (1, 1.75) {};
		\node [style=none] (6) at (1, -1.75) {};
		\node [style=none] (7) at (2, 2.25) {};
		\node [style=none] (8) at (2, -2.25) {};
		\node [style=none] (9) at (-1, -2.25) {};
		\node [style=none] (10) at (-1, -1.75) {};
		\node [style=none] (11) at (1.5, -0.25) {$\widetilde{\tau}$};
		\node [style=none] (12) at (0, -1.75) {};
		\node [style=none] (13) at (0, -0.75) {};
	\end{pgfonlayer}
	\begin{pgfonlayer}{edgelayer}
		\draw [qWire] (1) to (0);
		\draw [qWire] (0) to (2.center);
		\draw (4.center) to (3.center);
		\draw (3.center) to (5.center);
		\draw (5.center) to (6.center);
		\draw (6.center) to (10.center);
		\draw (10.center) to (9.center);
		\draw (9.center) to (8.center);
		\draw (8.center) to (7.center);
		\draw (7.center) to (4.center);
		\draw [qWire] (13.center) to (12.center);
	\end{pgfonlayer}
\end{tikzpicture}
\ \ = \sum_j v_j
\begin{tikzpicture}
	\begin{pgfonlayer}{nodelayer}
		\node [style=small box] (0) at (0, 0.75) {$\widetilde{T}$};
		\node [style=small box] (1) at (0, -0.75) {$\widetilde{T'}_j$};
		\node [style=none] (2) at (0, 1.75) {};
		\node [style=none] (3) at (-1, 1.75) {};
		\node [style=none] (4) at (-1, 2.25) {};
		\node [style=none] (5) at (1, 1.75) {};
		\node [style=none] (6) at (1, -1.75) {};
		\node [style=none] (7) at (2, 2.25) {};
		\node [style=none] (8) at (2, -2.25) {};
		\node [style=none] (9) at (-1, -2.25) {};
		\node [style=none] (10) at (-1, -1.75) {};
		\node [style=none] (11) at (1.5, -0.25) {$\widetilde{\tau}$};
		\node [style=none] (12) at (0, -1.75) {};
		\node [style=none] (13) at (0, -0.75) {};
	\end{pgfonlayer}
	\begin{pgfonlayer}{edgelayer}
		\draw [qWire] (1) to (0);
		\draw [qWire] (0) to (2.center);
		\draw (4.center) to (3.center);
		\draw (3.center) to (5.center);
		\draw (5.center) to (6.center);
		\draw (6.center) to (10.center);
		\draw (10.center) to (9.center);
		\draw (9.center) to (8.center);
		\draw (8.center) to (7.center);
		\draw (7.center) to (4.center);
		\draw [qWire] (13.center) to (12.center);
	\end{pgfonlayer}
\end{tikzpicture}
.
\eeq
As this holds for all $\widetilde{\tau}$,
and so, in particular, for our fiducial testers, we therefore have that 
\beq
 \sum_i u_i \mathbf{R}_{\widetilde{T}\circ\widetilde{T}_i}  = \sum_j {u'}_j \mathbf{R}_{\widetilde{T}\circ\widetilde{T'}_j}.
\eeq
Finally, using the fact that $\mathbf{R}_{\widetilde{T}\circ \widetilde{T_i}} = \mathbf{R}_{\widetilde{T}}\smallsquare \mathbf{R}_{\widetilde{T}_i}$ and similarly that $\mathbf{R}_{\widetilde{T}\circ \widetilde{T'}_j} = \mathbf{R}_{\widetilde{T}}\smallsquare \mathbf{R}_{\widetilde{T'}_j}$ we find 
\beq
\sum_i u_i (\mathbf{R}_{\widetilde{T}}\smallsquare\mathbf{R}_{\widetilde{T}_i})=\sum_j {u'}_j (\mathbf{R}_{\widetilde{T}}\smallsquare\mathbf{R}_{\widetilde{T'}_j}),
\eeq
which is our desired result. 

One can similarly show linearity in the first argument, namely,
\beq
\sum_k v_k (\mathbf{R}_{\widetilde{T}''_k} \smallsquare \mathbf{R}_{\widetilde{T}}) = \sum_{l} {v'}_l (\mathbf{R}_{\widetilde{T}'''_l} \smallsquare \mathbf{R}_{\widetilde{T}}).
\eeq
Putting these together, we obtain full bilinearity of $\smallsquare$, namely
\beq
\mathbf{U}\smallsquare \mathbf{V}= \sum_{ij} u_i v_j (\mathbf{R}_{\widetilde{T}_i} \smallsquare \mathbf{R}_{\widetilde{T'}_j}) = \sum_{kl} {u'}_k {v'}_l (\mathbf{R}_{\widetilde{T}''_k} \smallsquare \mathbf{R}_{\widetilde{T}'''_l}),
\eeq
as required.
\endproof

\subsection{Proof of Lemma~\ref{lemmadecomp}} \label{proofTL}

We now prove Lemma~\ref{lemmadecomp}, restated here:
\begin{lemma}
 A GPT is tomographically local if and only if one can decompose the identity process for every system $A$, denoted $\widetilde{\mathds{1}}_A$, as
\begin{equation}
\begin{tikzpicture}
	\begin{pgfonlayer}{nodelayer}
		\node [style=none] (0) at (0, 1) {};
		\node [style=none] (1) at (0, -1) {};
		\node [style=right label] (2) at (0, -0.75) {$A$};
	\end{pgfonlayer}
	\begin{pgfonlayer}{edgelayer}
		\draw [qWire] (0.center) to (1.center);
	\end{pgfonlayer}
\end{tikzpicture}

\quad = \quad \sum_{ij} [\mathbf{M}_{\widetilde{\mathds{1}}_A}]_i^j
\input{Diagrams/measprepA.tikz},
\label{eq:identitydecomp2}
\end{equation}
where $\mathbf{M}_{{\widetilde{\mathds{1}}}_A}$ is the matrix inverse of the transition matrix $\mathbf{N}_{{\widetilde{\mathds{1}}}_A}$ of the identity process, that is,
\beq
\mathbf{M}_{{\widetilde{\mathds{1}}}_A}:= \mathbf{N}_{\widetilde{\mathds{1}}_A}^{-1}.
\eeq
\label{lem:identitydecomp2}
\end{lemma}

\begin{proof}
First, we prove that if a GPT satisfies tomographic locality, then the identity has a decomposition of the form in Eq.~\eqref{eq:identitydecomp2}. We do this by defining a particular process $f$ as a linear expansion into states and effects with the carefully chosen set of coefficients $[\mathbf{M}_{{\widetilde{\mathds{1}}}_A}]_i^j$, and then we prove that $f=\widetilde{\mathds{1}}_A$.

Take any minimal spanning set $\{ \widetilde{P}_i^A \}_i$ of GPT states and spanning set $\{ \widetilde{E}_j^A \}_j$ of GPT effects, and consider the transition matrix $\mathbf{N}_{\widetilde{\mathds{1}}_A}$ with entries given by
\beq
[\mathbf{N}_{\widetilde{\mathds{1}}_A}]_i^j:= \ \
\begin{tikzpicture}
	\begin{pgfonlayer}{nodelayer}
		\node [style=point] (0) at (0, -0.75) {$\widetilde{P}_i^A$};
		\node [style=copoint] (1) at (0, 0.75) {$\widetilde{E}_j^A$};
	\end{pgfonlayer}
	\begin{pgfonlayer}{edgelayer}
		\draw [qWire] (1) to (0);
	\end{pgfonlayer}
\end{tikzpicture}.
\label{eq:gdef}
\eeq
Next, define $\mathbf{M}_{{\widetilde{\mathds{1}}}_A}^{-1} := \mathbf{N}_{\widetilde{\mathds{1}}_A}$, that is the inverse $\mathbf{N}_{\widetilde{\mathds{1}}_A}$ with matrix elements  $[\mathbf{M}_{{\widetilde{\mathds{1}}}_A}]_i^j$ satisfying
\begin{equation} \label{invertgf}
\sum_j [\mathbf{M}_{{\widetilde{\mathds{1}}}_A}]_i^j [\mathbf{N}_{{\widetilde{\mathds{1}}}_A}]_j^k = \delta_{ik},
\end{equation}

 The matrix inverse of $\mathbf{N}_{\widetilde{\mathds{1}}_A}$ exists, since the rows of $\mathbf{N}_{\widetilde{\mathds{1}}_A}$ are linearly independent. 
That is, we show that $\sum_i a_i [\mathbf{N}_{\widetilde{\mathds{1}}_A}]_i^j=0$ if and only if $a_i=0$ for all $i$. First, note that we have \[0=\sum_ia_i[\mathbf{N}_{\widetilde{\mathds{1}}_A}]_i^j = \sum_i a_i 
,
\] but as the $\widetilde{E}_j^A$ span the space of effects and composition is bilinear (see Lemma~\ref{bilinearityprf}), this means that $\sum_i a_i \widetilde{P}_i^A = 0$; then, since the $\widetilde{P}_i^A$ are linearly independent, we have $a_i=0$ for all $i$, as desired.  

 Next we use this inverse to define a process $f$ as
\beq
%
.
\eeq
A priori, there is no reason why this process must be a physical GPT process; however, it turns out to be the identity process, as we will now show.
Consider the expression
\beq
%
\eeq
for some $\widetilde{P}_k^A$ and some $\widetilde{E}_l^A$ from the minimal spanning sets above. Substituting the expansion of $f$ followed by applying the definition of $\mathbf{M}_{{\widetilde{\mathds{1}}}_A}$ and $\mathbf{N}_{{\widetilde{\mathds{1}}}_A}$, one has
\beq
%
\ =\sum_{ij}[\mathbf{M}_{{\widetilde{\mathds{1}}}_A}]_i^j [\mathbf{N}_{{\widetilde{\mathds{1}}}_A}]_j^l [\mathbf{N}_{{\widetilde{\mathds{1}}}_A}]_k^i.
\eeq

But now it follows from Eq.~\eqref{invertgf} that
\begin{equation}
\sum_{ij} [\mathbf{M}_{{\widetilde{\mathds{1}}}_A}]_i^j [\mathbf{N}_{{\widetilde{\mathds{1}}}_A}]_j^l [\mathbf{N}_{{\widetilde{\mathds{1}}}_A}]_k^i = \sum_i \delta_{il} [\mathbf{N}_{{\widetilde{\mathds{1}}}_A}]_k^i = [\mathbf{N}_{{\widetilde{\mathds{1}}}_A}]_k^l= \ %
.
\end{equation}
Hence, it holds that
\beq
%
\eeq
  for all $\widetilde{P}_k^A$ and $\widetilde{E}_l^A$ in the minimal spanning sets above. By the fact that these sets span the state and effect spaces respectively, it follows that
\beq
%
\eeq
for all $\widetilde{P}$ and $\widetilde{E}$.
Now, in any GPT which satisfies tomographic locality, namely Eq.~\eqref{loctom}, two channels which give the same statistics on all local inputs and outputs are equal, and hence $f$ is in fact the identity transformation. Hence, the identity transformation has a linear expansion, of the form given by Eq.~\eqref{eq:identitydecomp2}, namely
\beq
%
.
\eeq

Next, we prove the converse: if the identity has a linear expansion as in Eq.~\eqref{eq:identitydecomp2} in a given GPT, then that GPT satisfies tomographic locality. To see this, consider two bipartite processes $\widetilde{T}$ and $\widetilde{T}'$ which give rise to the same statistics on all local inputs, so that
\beq
\forall P_1,P_2,E_1,E_2\quad
%
.
\label{eq:samelocal}
\eeq

For any tester $\widetilde{\tau}$ of the appropriate type, one can write the probability generated by composing with $\widetilde{T}$ as
\beq \label{diagsplits1}
%
\eeq
simply by inserting the linear expansion of the identity on each system.
Similarly, one can write
\beq \label{diagsplits2}
%
.
\eeq
Noting that the RHS of Eq.~\eqref{diagsplits1} splits into two disconnected diagrams, and the same holds for the RHS of Eq.~\eqref{diagsplits2}, it follows from Eq.~\eqref{eq:samelocal} that
\beq
%
.
\eeq
Since this is true for any two processes satisfying Eq.~\eqref{eq:samelocal}, the principle of tomographic locality (Eq.~\eqref{loctom}) is satisfied.
\end{proof}

\subsection{Proof of Eqs.~\eqref{parallelcompr} and \eqref{nastycomp}}  \label{comprepn}

To prove Eq.~\eqref{parallelcompr}, one can decompose the four identities in the diagram and perform some simple manipulations of the resulting expression.
\allowdisplaybreaks
\begin{align}
%
.
\end{align}

To prove Eq.~\eqref{nastycomp}, one can insert four decompositions of the identity into the following diagram:
\begin{align}
%
.
\end{align}

\end{document}